\documentclass[12pt]{article}
\usepackage{graphicx}
\usepackage{euscript}
\usepackage{color}
\usepackage{amsmath}
\usepackage{amssymb}
\usepackage{amscd}
\usepackage{amsthm}
\usepackage{amsopn}
\usepackage{xspace}
\usepackage{verbatim}
\usepackage{amsmath}
\usepackage{amsfonts}
\usepackage{epic}
\usepackage{eepic}
\usepackage{empheq}
\usepackage[active]{srcltx}
\usepackage[a4paper,margin=3cm]{geometry}
\definecolor{red}{rgb}{1,0,0}
\definecolor{green}{rgb}{0,1,0}
\definecolor{SeaGreen}{RGB}{46,139,87}
\definecolor{Maroon}{RGB}{128,0,0}

\newcommand{\N}{\mathbb{N}}
\newcommand{\Z}{\mathbb{Z}}

\newcommand{\C}{{\mathbb{C}}}
\newcommand{\R}{{\mathbb{R}}}

\newcommand{\A}{{\mathcal A}}
\newcommand{\B}{\mathcal B}

\def\Dg {{\mathcal D}}
\newcommand{\F}{\mathcal F}
\newcommand{\FF}{\EuScript F}

\def\Jg {{\mathcal J}}

\newcommand{\LL}{\mathcal L}

\def\Mg {{\mathcal M}}
\newcommand{\OO}{\mathcal O}

\def\PP{\mathcal P}

\def\Sg {{\mathcal S}}

\def\Vg {{\mathcal V}}

\def\hg {\mathfrak h}

\def\Div{\text{\rm div\,}}

\def\0{\mathbf  0}

\def\XXint#1#2#3{{\setbox0=\hbox{$#1{#2#3}{\int}$ }
\vcenter{\hbox{$#2#3$ }}\kern-.6\wd0}}

\newcommand{\Union}{\mathop{\bigcup}\limits}

 

\errorcontextlines=0 \numberwithin{equation}{section}
\theoremstyle{plain}
\newtheorem{theorem}{Theorem}[section]
\newtheorem{lemma}[theorem]{Lemma}

\newtheorem{proposition}[theorem]{Proposition}
\newtheorem{assumption}[theorem]{Assumption}

\newtheorem{remark}[theorem]{Remark}
\newtheorem{corollary}[theorem]{Corollary}

\title{On the spectrum of some Bloch-Torrey vector operators}
  \author{ Y. Almog$^*$, Department of
  Mathematics, \\ Ort Braude College, \\ 
    Carmiel 2161002, Israel \\~\\
  and \\~\\
\noindent   B. Helffer, Laboratoire de Math\'ematiques Jean Leray, \\CNRS and Universit\'e de Nantes, \\
  2 rue de la Houssini\`ere, 44322 Nantes Cedex France}
\date{}
\begin{document}
\maketitle
\bibliographystyle{siam}
\begin{abstract}
  We consider the Bloch-Torrey operator in $L^2(I,\R^3)$ where
  $I\subseteq\R$.  In contrast with the $L^2(I,\R^2)$ (as well as the 
  $L^2(\R^k,\R^2)$) case considered in previous works. We obtain that
  $\R_+$ is in the continuous spectrum for $I=\R$ as well as discrete
  spectrum outside the real line. For a finite interval we find the
  left margin of the spectrum. In addition, we prove that the
  Bloch-Torrey operator must have an essential spectrum for a rather
  general setup in $\R^k$, and find an effective description for its
  domain. 
\end{abstract}
\section{Introduction}
\subsection{The Bloch-Torrey operator}
We consider a simplified version of the Bloch-Torrey equation
\cite[Eq.  (4)]{To}, that is commonly used to model
Diffusion-Weighted Magnetic Resonance Imaging (DW-MRI). For infinite
relaxation times and constant diffusivity it assumes the form
 \begin{equation}
\label{eq:2}
 \partial_t {\bf m}  =- \gamma {\bf b} \times {\bf m}  + D \Delta {\bf m}\,.
  \end{equation}
  This time-dependent equation describes the evolution  in time of a
  vector field ${\bf m}$ on $\mathbb R^3$, representing the
  magnetization vector
  under the action of an external magnetic field ${\bf b}$. \\

  To obtain any information on the semigroup associated with
  \eqref{eq:2}, we need to analyze the resolvent of a suitable
  realization of the differential operator $ - D \Delta + \gamma {\bf b} \times$.
  After dilation and a change of notation we write
  \begin{equation}
\label{eq:4}
  B_\epsilon(x,d_x) := -\epsilon^2 \Delta + {\bf b} \times\,.
  \end{equation}
  In the sequel we denote the magnetization ${\bf m}$ by ${\bf u}$.
  We begin by considering the general problem in $\mathbb R^3$ 
    providing first a precise definition of this spectral problem. 
    \begin{proposition} \label{prop1.1} Let ${\bf b}\in C^\infty(\mathbb
      R^3;\mathbb R^3)$. Then, the closure $\B_\epsilon$ of the operator
      $B_\epsilon(x,d_x)$ which is
      priori defined on $C_0^\infty(\mathbb R^3,\R^3)$ is maximally accretive as an unbounded operator in $L^2(\mathbb R^3,\mathbb R^3)$.\\
  \end{proposition}
    The proof  is given in Section \ref{s2}.
\begin{remark}
  Note that $\B_\epsilon$ can be extended as an unbounded operator  in  $L^2(\mathbb R^3,\mathbb C^3)$ which commutes with the complex conjugation.
  Hence its spectrum is invariant to complex conjugation.
  \end{remark}

  By the Hille-Yosida theorem (see \cite[Theorem 8.3.2]{da07}) there
  exists a continuous semi-group associated with $B_\epsilon$ and it is
  therefore natural to attempt to obtain some of its properties in the
  sequel.

  In Section 3 we focus on the case when $\bf b (x)$ depends only on one variable (say $x_1$). \\
  In this case, we apply a partial Fourier transform in the $x_2$ and
  $x_3$ direction which leads to the following family of $(1D)$
  operators depending on $ ((\xi_2,\xi_3)\in \mathbb R^2)$
               \begin{equation}
\label{eq:68}
                B_\epsilon\Big(x_1,\frac{d}{dx_1},\xi_2,\xi_3\Big)   := -\epsilon^2 \frac{d^2}{dx_1^2}\otimes I 
      + 
       \begin{pmatrix}
       \epsilon^2 (\xi_2^2+\xi_3^2)&-b_3 & b_2\\
       b_3 & \epsilon^2 (\xi_2^2+\xi_3^2)&-b_1\\
       -b_2&b_1& \epsilon^2 (\xi_2^2+\xi_3^2)
       \end{pmatrix}
     \,.
              \end{equation}
              The above operator (after
              reduction to the case $\xi_2=\xi_3=0$) is considered,
              assuming linearity of ${\bf b}$, i.e,
      \begin{displaymath}
             {\bf b} (x_1) = {\bf b_0} +  x_1  \, {\bf b_1} \,.
      \end{displaymath}
     where ${\bf b_0}\in \mathbb R^3$ and ${\bf b_1}\in \mathbb R^3\setminus \{
     0\}$.\\
 Using translation we can obtain ${\bf b_1}\perp{\bf b_0}$. Then, after 
 rotation and renormalization,  denoting  the canonical basis in
 $\mathbb R^3$ or $\mathbb C^3$  by $(\hat{\bf i}_1,\hat{\bf
   i}_2,\hat{\bf i}_3)$, it is sufficient to consider the case when 
\begin{displaymath}
  {\mathbf b}_0 = b_0 \, \hat{\bf i}_1\mbox{  and  } {\mathbf b}_1=\hat{\bf i}_3\,.
\end{displaymath}
 Note that ${\bf b}$ is divergence free, as is required from
  magnetic fields by Maxwell
  equations.  The operator $\B_\epsilon$ becomes
\begin{equation}
\label{eq:1}
  B_\epsilon\Big(x,\frac {d}{dx}\Big): =-\epsilon^2 \frac{d^2}{dx^2} -  [x{\mathbf b}_1+{\mathbf b}_0]\times \,,
\end{equation}
Note that the case $b_{0}=0$ reduces to the two-dimensional case. More
generally, if we suppose that ${\bf b}=b(x_1,x_2,x_3)\hat{\bf i}_3$
(though in the case of a divergence free field we get $ {\bf b}
=b(x_1,x_2) \,\hat{\bf i}_3\,.$) then the skew-symmetric matrix
associated with ${\bf b}\times$ is given by
\begin{equation}
\label{eq:92}
\Mg=
         \begin{pmatrix}
       0 &-b & 0\\
       b & 0 &0\\
       0 &0 & 0
       \end{pmatrix}\,.
\end{equation}
The eigenvectors associated with $\Mg$ are $\hat {\bf i}_3$, $v$, and
$\bar{v}$, where,
\begin{equation} \label{eq:defv}
  {\bf v}=\frac{1}{\sqrt{2}} (- i \, \hat{\bf i}_1 + \hat{\bf i}_2)\,.
\end{equation}
Since $\{{\bf v},\bar{\bf v}, \hat{\bf i}_3\}$ form an orthonormal basis for
$\C^3$ we may apply rotation to $-\epsilon^2\Delta +\Mg$ to obtain
in this new basis  the operator
\begin{equation}
\label{eq:93}
\widetilde{\mathcal B}_\epsilon:=
  \begin{pmatrix}
    -\epsilon^2\Delta +ib & 0 & 0 \\
    0 & -\epsilon^2\Delta - ib & 0 \\
    0 & 0 & -\epsilon^2\Delta
  \end{pmatrix}\,.
\end{equation}
Obviously, in this basis $-\epsilon^2\Delta +\Mg$ can be considered as three
separate scalar operators. The spectral properties of $-\epsilon^2\Delta +ib$
have been considered in \cite{AGH,al08,hen15}. Note that if we define
$-\epsilon^2\Delta +\Mg$  on $L^2(\R^3,\R^3)$ for $b=x_1$ we obtain that the
spectrum is $\R_+$ (which is precisely $\sigma(-\Delta)$ on $L^2(\R^3)$) given
that $\sigma (-\epsilon^2\Delta+ix)=\emptyset$ on $L^2(\R^3)$.

\subsection{Main statements}
We now present the main results of this work. In the case where $\B_\epsilon$ is defined in 
 $\R$ we obtain 
\begin{theorem}
\label{thm:unbounded}
  Let $\B_\epsilon$ be defined by (\ref{eq:1}), on the domain
  \begin{displaymath}
    D(\B_\epsilon)= \{{\bf u}\in H^2(\R,\C^3)\,|\,{\bf b}\times{\bf
      u}\in L^2(\R,\C^3)\,\}\,.
  \end{displaymath}
Then we have:
\begin{itemize}
\item $\Lambda \in \sigma(\B_\epsilon)\Leftrightarrow\overline{\Lambda} \in \sigma (\B_\epsilon)$.
\item  $\R_+\subset\sigma(\B_\epsilon)$.
\item Let for $n\in \mathbb N^*$ and $\epsilon >0$, $\kappa_n^{0} (\epsilon):= i+
  \frac{2n-1}{2}(1+i)\epsilon \,.$ Then for any $N\in\N^*$ there exist
  positive $\epsilon_0$ and $\hat C $, such that for all $0<\epsilon \leq \epsilon_0$ a
  sequence \break $\{\kappa_n (\epsilon)\}_{n=1}^N\subset\sigma(\B_\epsilon)$ can be found, satisfying
  \begin{equation}
 \label{eq:105}
\Big|\kappa_n(\epsilon) -\kappa_n^0(\epsilon) \Big| \leq   \hat C \epsilon^2\,,\, \mbox{ for } n=1 \dots,N \,.
  \end{equation}

\item  Let $\varrho>0$ and  $\hat R>0$.  Let further 
  \begin{equation}
    \label{eq:131}
N_\varrho=\Big[\frac{2\varrho+1}{2}\Big]\,,
  \end{equation}
where $[\cdot]$ denotes the integer part. Set now
\begin{equation}
 \label{eq:dg} 
  \Dg(\hat R,\varrho, \epsilon)=\{\Lambda\in\C\setminus\Union_{n=1}^{N_\varrho} \big(B(\kappa_n^0(\epsilon),\hat R\epsilon^2)\cup
  B(\overline{\kappa_n^0(\epsilon)},\hat R\epsilon^2)\big)\,\} \cap \{  \Re\Lambda\leq \varrho\epsilon\,\} \cap \{\Im \Lambda \neq 0\} \,.
\end{equation}
Then, there exist positive $C$ and $\hat R_0 >1$ such that for all $\hat
R_0<\hat R< [\sqrt{2}\epsilon]^{-1}$ and
$\Lambda\in\Dg(\hat R,\varrho,\epsilon)$ it holds that
\begin{equation}
\label{eq:107}
\|(\B_\epsilon-\Lambda)^{-1}\|\leq C\Big(1+  \frac{\epsilon^{2/3}}{|\Im \Lambda|^2}+\frac{1}{\hat R\epsilon^\frac 53}\Big)
 \,.
\end{equation}
\end{itemize}
\end{theorem}
\begin{remark}
  We note that for $\Lambda_r<0$, it holds, since $\B_\epsilon$ is accretive, that
  \begin{displaymath}
    \|(\B_\epsilon-\Lambda)^{-1}\|\leq \frac{1}{|\Lambda_r|} \,.
  \end{displaymath}
\end{remark}

We now state our main result for the Dirichlet realization $\B_\epsilon^I$
of the operator $B_\epsilon(x,d_x)$ (see \eqref{eq:1}) in $I= (a,b)$. 
\begin{theorem}
\label{prop:spectrum-bounded}
The domain of $\B_\epsilon^I$ is given by
  \begin{displaymath}
  D(\B_\epsilon^I)=H^2((a,b),\C^3)\cap H^1_0((a,b),\C^3) \,,
\end{displaymath}
Furthermore, let
  \begin{displaymath}
\rho_0 = \inf_{w\in H^1_0(a,b)} \frac{I(w)}{\|[x^2+1]^{1/2}w\|_2^2}\,,
\end{displaymath}
where
\begin{displaymath}
  I(w)= \|(xw)^\prime\|_2^2+\|w^\prime\|_2^2 \,.
\end{displaymath}
Then 
  \begin{equation}
\label{eq:127}
\lim_{\epsilon\to0}\epsilon^{-2}\Re \sigma(\B_\epsilon^I)=\rho_0 \,.
  \end{equation}
\end{theorem}

 The two-dimensional setting described by \eqref{eq:92} has
  received significant attention in the literature (cf. for instance
  \cite{AGH,Grebenkov07,Grebenkov17,callaghan1993principles})
  sometimes with time-dependent magnetic field. An example of a
  divergence-free magnetic field whose direction varies in space has
  been presented in the classical work of Torrey \cite{To}, where
\begin{displaymath}
  {\bf b}=(gx,gy,-2gz)
\end{displaymath}
for some $g\in\R$. 

In the mathematical literature \eqref{eq:4} with varying field
direction has been considered in the context of vector Schr\"odinger
operators \cite{kuetal20,kuetal19}. Thus, from the results in
\cite{kuetal20} we can conclude that a contraction semigroup is
associated with \eqref{eq:4} (an immediate conclusion of Proposition
\ref{prop1.1}). It should be noted that the maximal domain of
\eqref{eq:4} is not found in \cite{kuetal20}. In contrast, in
\cite{kuetal19} the maximal domain is found under the assumption that
$|\nabla{\bf b}|\,|{\bf b}|^{-\alpha}$ is bounded in $\R^d$ for some 
$0\leq\alpha<1/2$. We bring more general results (in our context) in Section
2. 

The rest of this contribution is arranged as follows. In the next
section we address the general operator \eqref{eq:4}, in the context
of  Schr\"odinger operators with matrix potential. In particular
we give conditions for the existence of an essential spectrum and
obtain the maximal domain for a rather general setting. In Section 3
we prove Theorem \ref{thm:unbounded}. Finally, in Section 4 we prove
Theorem \ref{prop:spectrum-bounded}.

    \section{Properties of Schr\"odinger operators with matrix-valued
      potentials}
\label{s2}
In this section we derive some basic properties of the operator
$\B_\epsilon$ given by \eqref{eq:4}, in settings significantly more general
than that of \eqref{eq:1}.  The analysis applies in particular to the
general differential operator (\ref{eq:4}) and hence also to the
one-dimensional operator (\ref{eq:1}).

\subsection{A more general operator}
One can generalize
\eqref{eq:4} even further by considering the operator
\begin{equation}
\label{eq:74}
   \PP(x,d_x):=  -\Delta \otimes I_d + M (x)\,,
\end{equation}
     where $I_d$ is the identity matrix acting on $\mathbb R^d$ and
     $M\in C^\infty(\R^k,M_d(\R))$, where $M_d(\R)$ denotes the set of all
     $d\times d$ matrices with real entries. We set $\epsilon=1$ as the value
     of $\epsilon$ does not have any effect on the properties which we consider in
     this section.\\
Set
\begin{displaymath}
    M_s = \frac 12 (M+ M^t)\,,\, M_{as} = \frac 12 (M-M^t)\,.
\end{displaymath}
We further assume that
    \begin{equation}
\label{eq:57}
    M_s \geq 0\,,
    \end{equation}
which is certainly the case in \eqref{eq:4}, where $M$ is
skew-symmetric. 
  
  \subsection{Accretivity}
  In this subsection, we extend maximal accretivity results that have
  been established for the selfadjoint operator $-\Delta +V$ (see
  \cite[Theorem 6.6.2 ]{hebook02}) and also for two interesting  non-selfadjoint operators: 
 the  Fokker-Planck operator \cite{HeNi} and the complex Schr\"odinger
  operator  $-\Delta + i V$ \cite{hebook13}.
  
    \begin{proposition}
\label{prop2.1}
      Let $\PP$, given by \eqref{eq:74}, be defined on
      $C_0^\infty(\R^k,\R^d)$  and satisfy \eqref{eq:57}. Then, its closure, under the graph norm,
      denoted by ${\bf P}$, is maximally accretive as an unbounded
      operator in $L^2(\mathbb R^k,\mathbb R^d)$. Moreover, 
\begin{displaymath}
  D({\bf P}) \subset H^1(\mathbb R^k,\mathbb R^d)\,.
\end{displaymath}
     \end{proposition}
      \begin{proof}
    We first observe that $\PP (x,d_x)$ is accretive on $C_0^\infty (\mathbb
    R^k,\mathbb R^d)$. To this end it is sufficient to note that
\begin{displaymath}
  \langle  \PP {\bf u}\,,\, {\bf u} \rangle_{L^2(\mathbb  R^k,\mathbb R^d)} \geq 0\,,
\end{displaymath}
   which holds since $\PP$ is the sum of the non negative
   operator $(-\Delta) \otimes I_d + M_s (x) $ and the antisymmetric matrix $M_{as}(x)$\,. 
We can then follow the proof in \cite[exercise 13.7]{hebook13} (which
refers to Theorem 13.14 and the proof of Theorem 9.15). 
      \end{proof}
   \begin{remark}   
     Proposition \ref{prop1.1} follows as a particular case of Proposition
     \ref{prop2.1} for $k=d=3$ and $M_s=0$.
       \end{remark}
      \subsection{Essential spectrum}
      If ${\rm dim}\,{\rm ker}\,M(x)>0$ for all $x\in\R^k$ (or for a
      suitable sequence of balls with centers tending to $+\infty$) as in
      the case $M_s\equiv0$ we may attempt to exploit the fact that
      locally, in any of the directions spanning ${\rm ker}\,M$, $\PP$
      is expected to behave like $\Delta$ to show that its resolvent is
      not compact and even that $\R_+\subseteq\sigma(\PP)$. We begin with the
      following proposition, establishing non-compactness of
      $(\PP-\lambda)^{-1}$.
    \begin{proposition}
    Let  $\{a_n\}_{n=1}^\infty\subset\R^k$ satisfy $|a_n| \to +\infty$ and $|a_n-a_m|\geq 1$
for all $n\neq m$. Suppose that there exists $C>0$ and a 
     unit vector field ${\bf c} (x) $  such
        that,   for all $n\in\N$ and $ x\in B(a_n,\frac 12)$,
       \begin{subequations}
\label{eq:78}
         \begin{equation}
    M(x)  {\bf c} (x) =0\,,\
\end{equation} 
    and
\begin{equation}
     |d_x^\alpha c_j (x)| \leq C , \mbox{ for all }  \alpha \mbox{  s.t } 1\leq|\alpha| \leq 2, j=1,2,\dots,d\,.
\end{equation}
       \end{subequations}
     Then the resolvent of ${\bf P}$ is not compact.
     \end{proposition}
      \begin{proof}  
One  looks for an  infinite  orthonormal  family ($n\geq N$) in the form
\begin{displaymath}
  \Phi_n (x):= 2^d{\bf c} (x) \phi (2(x - a_n))\,,
\end{displaymath}
where $\phi\in C_0^\infty(B(a,b))$ is of unity norm (i.e., $\|\phi\|_{L^2} =1$).
Note that the above construction guarantees that $\|\Phi_n\|_2=1$. \\
  As $M\Phi_n\equiv0$, it can be easily verified that
      $\{{\bf P} \Phi_n\}_{n=1}^\infty$ is uniformly bounded. It follows that the
      resolvent of ${\bf P}$ cannot be compact.
 \end{proof}
 \begin{remark}
\label{rem:unit}
  If we do not assume that ${\bf c}(x)$ is  a unit vector. Then  \eqref{eq:78} can be replaced by  (assuming ${\bf c}(x)\neq0$ for $|x| \geq R$) the existence of $C>0$ such that:
   \begin{equation}
     \label{eq:79}
 |d_x^\alpha c_j (x)| \leq C\,  |{\bf c}(x)|, \mbox{ for all } \alpha \mbox{  s.t } 1\leq|\alpha| \leq 2, j=1,2,\dots,d\,,
   \end{equation}
which is normally easier to verify than (\ref{eq:78}b). 
 \end{remark}
     We note that for \eqref{eq:68} we have $k=1,d=3$, $M_s=0$ and ${\bf c}={\bf b}$.
     In this case we may conclude that 
     \begin{corollary}
   Suppose that for  $|x|\geq R$ it
         holds that ${\bf b}\neq0$ and that for some $C>0$ 
    \begin{equation}
    |d_x^\alpha b_j (x)| \leq C \, |{\bf b} (x)|, \forall \alpha \mbox{  s.t } 1\leq|\alpha| \leq 2\,, j=1,2,3\,.
  \end{equation}

  Then the resolvent of $\mathcal B_\epsilon$, given by \eqref{eq:78}, is
  not compact.
   \end{corollary}
 Making slightly stronger assumptions on ${\bf c}$ we now prove the
existence of an essential spectrum for ${\bf P}$.
    \begin{proposition}
     Let  $\{a_n\}_{n=1}^\infty\subset\R^k$ satisfy $|a_n| \to +\infty$ and
       \begin{displaymath}
      r_n=\inf_{m\neq n}|a_n-a_m|\xrightarrow[n\to\infty]{}\infty \,.
       \end{displaymath}

Suppose that there exists a  unit vector field  ${\bf
         c} $ and   $R>0$ such that,  for all $n\in\N$ and $x\in B(a_n,r_n)$,
       \begin{subequations}
\label{eq:87}
         \begin{equation}
    M(x)  {\bf c} (x) =0\,,
\end{equation}
and
\begin{equation}
    |d_x^\alpha c_j (x)|\xrightarrow[|x|\to\infty]{}0 ,\, \forall \alpha \mbox{  s.t }\, 1\leq|\alpha| \leq 2,\, j=1,2,\dots,d\,.
\end{equation}
       \end{subequations}
     Then, $\R_+\subseteq \sigma({\bf P})$, where $\mathbb R_+:=[0,+\infty)$.
     \end{proposition}
      \begin{proof}  
One  looks for an  infinite  orthonormal  family ($n\geq N$) in the form
\begin{displaymath}
  \Phi_n (x):= r_n^{-d} {\bf c} (x) \phi (2r_n^{-1}(x - a_n))\,,
\end{displaymath}
   where 
    $\phi\in C_0^\infty(B(0,1)$ is of unity norm (i.e., $\|\phi\|_{L^2} =1$).
    
      Let $\lambda\in\R_+$, and $\xi\in\R^k$ satisfy $|\xi|^2=\lambda$. Let further
      $\Psi_n=e^{i\xi\cdot x}\Phi_n$. As $M\Psi_n\equiv0$, it can be easily verified that
      \begin{displaymath}
        \|(\PP-\lambda)\Psi_n\|_2\to0\,.
      \end{displaymath}
       It follows that $\lambda\in\sigma({\bf P})$ and hence $\R_+\subseteq\sigma({\bf P})$.
 \end{proof}
 
 \begin{remark}
\label{rem:essential-unit}
  If  ${\bf c}(x)$ is not assumed to be a unit vector, then  (\ref{eq:87}c)
   should be replaced (assuming ${\bf c}\neq0$) by,
   \begin{equation}
\label{eq:88}
 |d_x^\alpha c_j (x)| \leq  \delta(x)|{\bf c}(x)|,\, \forall \alpha \mbox{  s.t }\, 1\leq |\alpha| \leq 2,\, j=1,2,\dots,d\,,
   \end{equation}
where $\delta(x)\xrightarrow[|x|\to\infty]{}0$. 
 \end{remark}

 Assuming a linear magnetic field (as in (\ref{eq:1})) we consider a
 field ${\bf b}$ satisfying
\begin{equation}
\label{eq:89}
  {\bf b}= A\, x + {\bf f} \,,
\end{equation}
 where $A\neq0$ is a $d \times k$ matrix and ${\bf f}\in\R^d$. Let ${\bf w}$
denote an eigenvector $of A^TA$ corresponding to a non-zero eigenvalue
(which clearly exists since $A\neq0$).  Choosing $a_n=s_n{\bf w}$, where
$r_n=s_n-s_{n-1}\uparrow\infty$, it can be easily verified that there exists
$C>0$ such $|A\, x |\geq C|x|$ in $B(a_n,r_n/2)$. Consequently,
${\bf b}$ satisfies \eqref{eq:88}, and hence $\R_+\subseteq\sigma(\mathcal B_\epsilon)$
whenever
${\bf b}$ satisfies \eqref{eq:89}. Note that if $A=0$, then, by (\ref{eq:93}),
\begin{displaymath}
\sigma({\bf P})=\R_+\cup \{ \R_++i |{\bf f}|\} \cup \{ \R_+-i |{\bf f}|\}\,.
\end{displaymath}

Of particular interest is the case
\eqref{eq:1}. Here $ A =\hat {\bf i}_3$ and hence $\R_+$ is in the essential
spectrum of $\mathcal B_\epsilon$.

    \subsection{Maximal estimates} 
    A natural question is the effective description of $D({\bf P})$,
    which is currently defined as the closure of $C_0^\infty(\R^k,\R^d)$
    under the graph norm. While far from having an optimal result we
    can still determine the domain in two different cases: The first
    of them concerns matrices with positive symmetric part and
    coefficients of bounded derivatives.  In the second one we assume
    a more general class of skew-symmetric matrices.
   \begin{proposition}
   Suppose that $M$ satisfies \eqref{eq:57} and suppose that there exist $C>0$ and $R>0$ such that,
    for  $|x|\geq R$ it holds that
\begin{equation}
\label{eq:90}
     |d_xM_{i,j}(x)| \leq C  \,,\,  \forall   \,i,j =1,2,\dots,d\,,
\end{equation}
Then,
\begin{equation}
\label{eq:94}
    D({\bf P})= \{{\bf u}\in H^2(\mathbb R^k,\mathbb R^d)\,,\, M {\bf u} \in L^2(\mathbb R^k,\mathbb R^d)\}\,.
\end{equation}
\end{proposition}
  \begin{proof}   
Let ${\bf u} \in C_0^\infty(\mathbb R^k,\mathbb R^d)$. Clearly,
\begin{equation*}
\|\nabla {\bf u}\|^2_2+\langle u,M_s{\bf u}\rangle  = \langle {\bf u},\PP {\bf u}\rangle\leq \frac 12  \left( \| \PP {\bf u}\|^2_2 + \|{\bf u}\|^2_2\right)\,.
\end{equation*}
By  \eqref{eq:57} we can conclude that
\begin{equation}\label{eq:h1}
\|\nabla {\bf u}\|^2_2 = \langle {\bf u},\PP {\bf u}\rangle\leq \frac 12  \left( \| \PP {\bf u}\|^2_2 + \|{\bf u}\|^2_2\right)\,.
\end{equation}
Next, we write
\begin{equation} \label{eq:h2}
\begin{array}{ll}
\| M {\bf u}\|^2_2 &=\langle M {\bf u},\PP {\bf u}\rangle+
\langle M {\bf u},\Delta {\bf u}\rangle \\ & \leq 
\| \PP {\bf u} \|_2\, \| M {\bf u} \|_2 - \langle M_s \nabla {\bf u},\nabla {\bf
  u}\rangle-\langle (\nabla M) {\bf u},\nabla {\bf u}\rangle\\
  &\leq  \| \PP {\bf u} \|_2\, \| M {\bf u} \|_2 - \langle (\nabla M) {\bf u},\nabla {\bf u}\rangle
  \,.
  \end{array}
 \end{equation}
 We can then conclude from \eqref{eq:90}, \eqref{eq:h1} and
 \eqref{eq:h2} the existence of $C$ such that
 \begin{equation}\label{eq:h3}
 \| M {\bf  u} \|^2 \leq C ( \|\mathcal P {\bf u}\|^2 + \|{\bf u}\|^2)\,,\, \forall {\bf u} \in C_0^\infty(\mathbb R^k,\mathbb R^d)\,.
 \end{equation}
 
 Let $({\bf u},\PP {\bf u}) \in[L^2(\R^k,\R^d)]^2$ and $\{{\bf
   u}_n\}_{n=1}^\infty \subset C_0^\infty(\mathbb R^k,\mathbb R^d)$ satisfy ${\bf
   u}_n\to{\bf u}$ in the graph norm.  Then, by \eqref{eq:h3}, $\{M
 {\bf u}_n\}_{n=1}^\infty$ is a Cauchy sequence in $L^2(\mathbb
 R^k,\mathbb R^d)$.  Since $M {\bf u}_n\to M{\bf u}$ in $L^2(\Omega,\R^d)$
 for any $\Omega\Subset \R^k$ we may conclude that the limit of $\{M {\bf
   u}_n\}_{n=1}^\infty$ in $L^2(\mathbb R^k,\mathbb R^d)$ is $M{\bf u}\in
 L^2(\R^k,\R^d)$.  By subtraction from $\PP {\bf u}$, we obtain
 \linebreak $\Delta {\bf u} \in L^2(\R^k,\R^d)$, thus $ {\bf u}\in
 H^2(\mathbb R^k,\R^d)$ and hence $ {\bf u}\in D({\bf P})$.
\end{proof}
We note that the above result is a particular case of \cite[Theorem
3.2]{kuetal19}.

We now obtain a stronger result for the case where $M$ is skew-symmetric.
  \begin{proposition}
\label{prop:maximal}
   Suppose that $ M=M_{as}\in C^2(\R^k,M_d(\R))$ is a skew symmetric matrix. Let
   \begin{displaymath}
     \Sg(x) = \inf_{\lambda_j \in \sigma(M(x))\setminus \{0\}} |\lambda_j |\,,
   \end{displaymath}
and suppose that there exist $C>0$ and $R>0$ such that,
    for  $|x|\geq R$ it holds that
\begin{equation}
\label{eq:90a}
  | d_x^\alpha  M_{i,j}(x)| \leq C\,  \Sg(x) \,,\,  \forall   \,i,j =1,2,\dots,d, 1 \leq |\alpha| \leq 2\,.
\end{equation}
and that  
   ${\rm dim\,}{\rm Ker\, } M (x)$ is constant for $|x|\geq R$. 
Then,
\begin{equation}
\label{eq:94a}
    D({\bf P})= \{{\bf u}\in H^2(\mathbb R^k,\mathbb R^d)\,,\, M {\bf u} \in L^2(\mathbb R^k,\mathbb R^d)\}\,.
\end{equation}
     \end{proposition}
 \begin{proof}
   We note that under our assumptions, all eigenvalues in $\sigma(M(x))$
   are purely imaginary, $\{0\}\in\sigma(M(x))$ when $d$ is odd, and that
   $\Sg(x)$ is continuous. Note further that by the condition on ${\rm
     dim}\,{\rm Ker\, } M (x)$ we have either $M(x)=0$ or $\Sg(x)>0$ for
   all $|x|>R$. The treatment of the first case being evident, we
   treat the second case. We introduce for any $x$, $\Pi_0 (x)$ the
   projector on the kernel of $M(x)$ and \break $\Pi (x):=
   (I-\Pi_0(x))$.  When $|x|\geq R$, we note that $\Pi_0(x)$ (hence
   $\Pi(x)$) depends smoothly on $x$ and note that we have
\begin{displaymath}
\Pi_0 (x):= \frac{1}{2\pi i} \int_{\gamma} (z-M(x))^{-1} dz \,,
\end{displaymath}
for any positively oriented circle $\gamma$ of radius strictly smaller
than $\Sg (x)$.

We now estimate $\partial_{x_j} \Pi_0 (x)$ for $|x|>R$.  Let $x_0\in\R^d$. For
any $\gamma$ of radius smaller than $\Sg (x_0)$, we have
\begin{displaymath}
  (\partial_{x_j} \Pi_0 ) (x_0) = \frac{1}{ 2\pi i} \int_{\gamma} (z-M(x_0))^{-1} (\partial_{x_j} M)(x_0) (z-M(x_0))^{-1}dz 
\end{displaymath}
Since
\begin{displaymath}
   \|(z-M(x))^{-1}\| \leq \frac{1}{\min(|z|,\Sg(x)-|z|)} \,,
\end{displaymath}
we obtain, by choosing $\gamma$ to be of radius $\frac 12 \Sg(x_0)$, with
the aid of \eqref{eq:90a} and the fact that the length of $\gamma$ is $\pi \Sg(x_0)$, 
 the existence of $C>0$ such that, for all $|x| \geq R$,
\begin{equation}\label{comm1}
\| (\partial_{x_j} \Pi_0 ) (x) \|_{M_d(\mathbb R)} 
 \leq  C \,,\, \forall j=1,\dots,k\,.
\end{equation}

Similarly, for any $\gamma$ of radius $< \Sg (x_0)$, it holds that
\begin{displaymath}
  \begin{array}{l}
(\partial_{x_j} \partial_{x_\ell}\Pi_0 ) (x_0)\\
\quad  = \frac{1}{ 2\pi i} \int_{\gamma} (z-M(x_0))^{-1} (\partial_{x_j x_\ell } M)(x_0) (z-M(x_0))^{-1}dz \\ \qquad  + 
\frac{1}{2 \pi i} \int_{\gamma} (z-M(x_0))^{-1} (\partial_{x_j} M)(x_0) (z-M(x_0))^{-1}(\partial_{x_\ell} M)(x_0) (z-M(x_0))^{-1}dz\,.
\end{array}
\end{displaymath}
Using \eqref{comm1} we establish the existence of $C>0$ such
that, for $|x| \geq R$,
\begin{equation}\label{comm2}
\| (\partial_{x_j x_\ell} \Pi_0 ) (x) \|_{M_d(\mathbb R)} 
 \leq  C \,,\, \forall j,\ell =1,\dots,k\,.
\end{equation}

We now introduce $\chi\in C_0^\infty(\mathbb R^k)$ such that $\chi =1$ on
$B(0,R)$ and let $\tilde \chi=1-\chi$.  Next, we write for any ${\bf
  u}\in C_0^\infty(\R^k,\R^d)$
\begin{equation}
\label{eq:95}
\begin{array}{ll}
\| M {\bf u}\|^2 &  = \langle M{\bf u}, \mathcal Pu \rangle + \langle \chi M{\bf u}, \Delta {\bf u}\rangle + \langle \tilde \chi M {\bf u}, \Delta {\bf u}\rangle \\
&=  \langle M{\bf u}, \mathcal Pu \rangle  - \langle (\nabla\chi) M{\bf u}, \nabla {\bf u}\rangle  -
\langle \chi \nabla  M{\bf u}, \nabla {\bf u}\rangle + \langle \tilde \chi M{\bf u}, \Delta
{\bf u}\rangle 
\end{array}
\end{equation}
Since $M\in C^2(\R^k,M_d(\R))$ we can conclude from \eqref{eq:h1} that
\begin{equation}
\label{eq:103}
  - \langle (\nabla\chi) M{\bf u}, \nabla {\bf u}\rangle  -
\langle \chi \nabla  M{\bf u}, \nabla {\bf u}\rangle \leq C\|{\bf u}\|_2\|\nabla{\bf u}\|_2\leq
C\|{\bf u}\|_2(\| \PP {\bf u}\|_2 + \|{\bf
    u}\|_2)
\end{equation}
To bound the last term on the right-hand-side of \eqref{eq:95} we
first observe that
\begin{displaymath}
  M =M\Pi =\Pi M= \Pi M \Pi \,.
\end{displaymath}
Hence,
\begin{displaymath}
\langle \tilde \chi M{\bf u}, \Delta
{\bf u}\rangle = \langle \tilde \chi M \Pi{\bf u}, \Delta {\bf u}\rangle =  
 \langle \tilde \chi  M {\bf u},[ \Pi, \Delta ] {\bf u}\rangle
  +  \langle \tilde \chi  M \Pi {\bf u}, \Delta \Pi  {\bf u}\rangle
\end{displaymath}
By \eqref{comm1}, \eqref{comm2}, and \eqref{eq:h1} we have that
\begin{displaymath}
  \|[ \Pi, \Delta ] {\bf u}\|_2\leq C(\|{\bf u}\|_2+\|\nabla{\bf u}\|_2)\leq C(\| \PP {\bf u}\|_2 + \|{\bf
    u}\|_2)\,,
\end{displaymath}
and hence
\begin{equation}
\label{eq:104}
   |\langle \tilde \chi  M {\bf u},[ \Pi, \Delta ] {\bf u}\rangle|\leq C\, \|M{\bf u}\|_2\, (\| \PP {\bf u}\|_2 + \|{\bf
    u}\|_2)\,.
\end{equation}
Finally, we write
\begin{displaymath}
\begin{array}{ll}
\langle \tilde \chi  M \Pi {\bf u}, \Delta \Pi  {\bf u}\rangle & = - \langle [\nabla,  \tilde \chi  M ]\Pi {\bf u}, \nabla \Pi  {\bf u}\rangle \\
&= - \langle (\nabla \tilde \chi)   M  {\bf u}, \nabla \Pi  {\bf u}\rangle 
 - \langle\tilde \chi  (\nabla M) \Pi {\bf u}, \nabla \Pi  {\bf u}\rangle 
\end{array}
\end{displaymath}
 By \eqref{eq:90a} and \eqref{comm1} we obtain, for any $\eta \in (0,1)$
\begin{displaymath}
|\langle \tilde \chi  M \Pi {\bf u}, \Delta \Pi  {\bf u}\rangle| 
\leq \eta \, \| M{\bf u}\|^2_2 + C_\eta (\| \nabla {\bf u}\|^2_2 + \| {\bf u}\|^2_2) 
\end{displaymath}
Substituting the above (with sufficiently small $\eta$), together with
\eqref{eq:104} and \eqref{eq:103} into \eqref{eq:95} yields
\begin{equation}
\label{eq:91}
\| M {\bf u}\|_2 \leq C(\| \PP {\bf u}\|_2 + \|{\bf
    u}\|_2)\,,\, \forall {\bf u} \in C_0^\infty(\mathbb R^k,\mathbb R^d)\,.
\end{equation}
We complete the proof in the same manner as in the proof of the
previous proposition.
   \end{proof}  
As mentioned before, in \cite{kuetal19} the authors consider the
operator 
\begin{displaymath}
  \PP = -\Div Q\, \nabla +M \,,
\end{displaymath}
for the case where there exists $\beta \in \R$ such that
  \begin{equation}
\label{eq:133}
    {\mathbf \xi}\cdot M({\bf x}){\mathbf \xi}\geq \beta|\xi|^2 \,,
  \end{equation}
  for all ${\bf \xi}\in\R^d$ and ${\bf x}\in\R^k$.  In the case $Q=I$, it
  is shown in \cite{kuetal19} that when $\nabla M\circ M^{-\gamma}$ is bounded in
  $L^\infty(\R^d)$ for some $0\leq\gamma<1/2$, then 
  \begin{displaymath}
    D(\PP)=\{u\in
  H^2(\R^d)\,|\,Mu\in L^2(\R^d)\}\,.
  \end{displaymath}
(It should be mentioned that the
  results in \cite{kuetal19} are stated in $L^p$ for any \break $p\in(1,\infty)$
  whereas here we consider only the case $p=2$.)  We note that while
  \eqref{eq:133} clearly holds in the case where $M$ is
  skew-symmetric, \eqref{eq:90a} applies to cases where \break $\nabla M \circ
  M^{-\gamma}$ is unbounded in $L^\infty(\R^d)$ for all $0\leq\gamma<1/2$. Thus, for
  instance, in the case $d=2$ we may consider (see \cite[Example
  2.4]{kuetal19})
\begin{displaymath}
  M=
  \begin{bmatrix}
    0 & 1+|x|^r \\
    -(1+|x|^r) & 0
  \end{bmatrix}
\end{displaymath}
Since   for $\nabla M \circ  M^{-\gamma}$ to be bounded we must have $\gamma\geq1-r^{-1}$, one
can apply the results in  \cite{kuetal19} for $r<2$ only, whereas
\eqref{eq:90a} holds for $r\geq2$ as well.  
   \begin{corollary}
\label{cor:cross-product}
Let $d=k=3$ and $M{\bf u}={\bf b}\times{\bf u}$. Then, if there exist $C>0$
and $R>0$ such that for all $|x|\geq R$ it holds that ${\bf b(x)}\neq 0$
and
    \begin{equation}
      \label{eq:96}
|d_x^\alpha b_j (x)| \leq  C\, |{\bf b}(x)|, \forall \alpha \mbox{  s.t } 1\leq |\alpha| \leq 2, j=1,2,3\,,
    \end{equation}
then
\begin{displaymath}
    D({\bf P})= \{{\bf u}\in H^2(\mathbb R^3,\mathbb R^3)\,,\, {\bf b}\times{\bf u} \in L^2(\mathbb R^3,\mathbb R^3)\}\,.
\end{displaymath}
   \end{corollary}
   \begin{proof}
Since $\Sg(x)=|{\bf b}(x)|$ we can easily conclude \eqref{eq:90} from
\eqref{eq:96}. 
   \end{proof}
   
   \subsection{Bounded components}
   In this section we consider the case $d=3$, where $M$ is the matrix
   associated with a vector product with ${\bf b}$. Assuming that two
   components of ${\bf b}$ are bounded, we may obtain $D({\bf P})$
   even in cases where the third component does not satisfy
   \eqref{eq:96}.
   
   \subsubsection{Characterization of the domain}
   To this end we use the
   results in \cite{HeNo}, obtained for the scalar operator $-\Delta + i
   V(x)$.
        \begin{proposition}
\label{prop:bounded-caract-doma}
          Let $\B_1$ denote the closure under the graph norm of
          \eqref{eq:4} with $\epsilon=1$, where ${\bf b}=(b_1,b_2,b_3)$.
          Suppose that $b_1$ and $b_2$ belong to $L^\infty(\R^k)$. Suppose
          further for some $r\in\Z_+$ that $b_3\in C^{r+1}(\R^k)$
          satisfies
      \begin{equation}
\label{eq:98}
\max_{|\beta|=r+1} |D_x^\beta b_3(x)| \leq C_0\, \sqrt{ \sum _{|\alpha| \leq r}  |D_x^\alpha b_3(x)|^2 +1}\, .
\end{equation}
Then  
\begin{equation}
\label{eq:99}
      D(\B_1)= \{{\bf  u} \in H^2 (\mathbb R^k,\mathbb
      R^3)\,,\,{\bf b}\times {\bf u}        \in L^2(\mathbb R^k,\mathbb R
      ^3) \}\,. 
\end{equation}
      \end{proposition}
      \begin{proof}
       Let ${\bf v}$ be defined by \eqref{eq:defv}.   In the basis $\{{\bf v},\bar{\bf v}, \hat{\bf i}_3\}$,
       the skew-symmetric matrix  $M$ assumes the form 
  \begin{displaymath}
   \tilde M := \left(
       \begin{array}{ccc}
       -i b_3&0 & -(b_1-i b_2)/\sqrt{2}\\
       0 & i b_3 &-(b_1+i b_2) /\sqrt{2}\\
       (b_1+ i b_2)/\sqrt{2} &(b_1- i b_2)/\sqrt{2}&0
       \end{array}
       \right)\,. 
  \end{displaymath}
and 
\begin{displaymath}
      \tilde{\B}_1:= -\Delta \otimes I_3 + \tilde M \,.
\end{displaymath}
Consider $\tilde {\bf u} \in L^2(\R^k,\R^d)$ satisfying $ {\bf \tilde
  \B}_1\tilde{\bf u}=\tilde {\bf f} \in L^2(\R^k,\R^3)$. Then, it holds that
\begin{equation}
\begin{array}{l}
-\Delta \tilde u_1 - i b_3 \tilde u_1 = \tilde g_1\\
- \Delta \tilde u_2 +  i b_3 \tilde u_2  = \tilde g_2\\
- \Delta \tilde u_3 =\tilde g_3\,,
\end{array}
\end{equation}
where
\begin{equation}
\begin{array}{ll}
\tilde g_1 &=\tilde f_1 + \frac{b_1-i b_2}{\sqrt 2} \tilde u_3, \\
\tilde g_2 & =\tilde f_2 + \frac{b_1+i b_2}{\sqrt 2} \tilde u_3 \\
\tilde g_3 & = \tilde f_3 - \frac{b_1+ i b_2}{\sqrt{2}} \tilde u_1  - \frac{b_1- i b_2}{\sqrt{2}}   \tilde u_2\,.
\end{array}
\end{equation}
Clearly, $\tilde g_i \in L^2(\mathbb R^k)$ for all $i\in\{1,2,3\}$. By
standard elliptic estimates we then conclude that $\tilde
u_3\in H^2(\R^k)$. We
can then apply on  the two first lines \cite[Theorem 5]{HeNo} to
conclude that
\begin{displaymath}
  \|b_3\tilde{u}_i\|_2 \leq  C (\|\tilde{g}_i\|_2+\|\tilde{u}_i\|_2)\,,
  \quad i=1,2 \,.
\end{displaymath}
Hence $\Delta\tilde{u}_i\in L^2(\R^k)$ for $i=1,2$, and by standard elliptic estimates $\tilde
u_i\in H^2(\R^k)$ for $i=1,2$. We can thus conclude that
 \begin{displaymath}
   D(\tilde{\B}_1)= \{ \tilde {\bf u} \in H^2(\mathbb R^k,\mathbb R ^3)\,,\, b_3 \tilde u_1 \in L^2(\mathbb R^k)\,,\, b_3 \tilde u_2\in L^2 (\mathbb R^k)\}\,.
 \end{displaymath}

An inverse transformation to the original basis establishes \eqref{eq:99}.
  \end{proof}
         \subsubsection{Resolvent estimates}
         We now obtain estimates for the resolvent of $\B_1$, using
         well known resolvent estimates obtained for $-\Delta + i
         b_3$. Better estimates are obtained in the next sections for
         the particular case where ${\bf b}$ is given by \eqref{eq:1}.\\
         
         The resolvent equation $(\tilde \B_1- \lambda) {\tilde u} =
         {\tilde f}$, takes the form (dropping the accent in the sequel):
         \begin{subequations}
\label{eq:2.15}
      \begin{alignat}{2}
     f_1 &=     (-\Delta -i b_3-\lambda) u_1 - \frac{b_1-i b_2}{\sqrt 2} u_3
     \\
    f_2&= (-\Delta + ib _3 -\lambda) u_2  - \frac{b_1+i b_2}{\sqrt 2} u_3\\
    f_3& = (-\Delta -\lambda) u_3 +  \frac{b_1+ i b_2}{\sqrt{2}} u_1  +  \frac{b_1- i b_2}{\sqrt{2}}   u_2\,.
    \end{alignat}
         \end{subequations}
    Assuming that $\lambda \notin \sigma ( -\Delta \pm i b_3)$ we write
\begin{displaymath}
  R_{\pm} (\lambda) =  ( -\Delta \mp i b_3 - \lambda)^{-1}  
\end{displaymath}
to obtain the following equation for $u_3$:
\begin{equation} 
\label{eq:101}
     \left( (-\Delta-\lambda)  - c  R_-(\lambda) \bar c -  \bar c  R_+ (\lambda)  c
     \right) u_3 = f_3 -  c  R_+(\lambda)  f_1 -  \bar c   R_-(\lambda)   f_2\,,
\end{equation}
where
\begin{displaymath}
 c:= (b_1 + i b_2)/\sqrt{2}\,. 
\end{displaymath}
Assuming further $\lambda \notin \mathbb R^+$ (note that ${\bf b}$ does not
necessarily meet the condition set in Remark \ref{rem:essential-unit})
we can attempt to estimate the norm of the well-defined bounded
operator
\begin{equation}
\label{eq:100}
K_\lambda:=  \frac 12   (-\Delta -\lambda)^{-1} \left( c  R_-(\lambda)  \bar c  + \bar c  R_+(\lambda)  c \right)\,.
\end{equation} 

To proceed further we need to introduce the following assumption on
the resolvent of $-\Delta \pm i b_3$
\begin{assumption}\label{ass2.11} For a given interval $I$, 
there exist $s <1$,  $D_1>0$ and $D_2 >0$ such that, if $\Re \lambda \in I $ and $|\Im \lambda| \geq D_1$, then 
\begin{displaymath}
    \| R_\pm (\lambda)\| \leq  D_2 \,  |\Im \lambda|^s  \,.
\end{displaymath}
\end{assumption}
\begin{remark}
  The above bound applies for $I=(-\infty,\tau)$ for every $\tau\in\R$, and
  $b_3(x) = x_1$ (in which case $s=0$ due to translation invariance, 
  see \cite[ Proposition 14.11]{hebook13}) or $b_3(x)=x_1^2$ (where
  $s=-\frac 13$, see \cite[ Proposition 14.13]{hebook13}).
   \end{remark}

Given the above assumption we can obtain the following resolvent estimate
  \begin{proposition} 
\label{prop:bounded}
    Let $b_3\in C^r(\R^k)$ satisfy assumption \ref{ass2.11}, for some
    given interval $I$. Then, there exist $C_1>0$ and $C_2 >0$ such
    that, for all $\lambda\in\C$ satisfying $\Re \lambda \in I $ and $|\Im \lambda| \geq
    C_1$, it holds that $\lambda \notin \sigma(\B_1)$ and
\begin{equation}
\label{eq:102}
 \|  (\B_1-\lambda)^{-1}\| \leq C_2\, |\Im \lambda|^s\,.
\end{equation}
\end{proposition}
\begin{proof} 
Given the fact that $\|(-\Delta-\lambda)^{-1}\|\leq|\Im\lambda|^{-1}$ we obtain from
\eqref{eq:100}, Assumption \ref{ass2.11} and the boundedness of $c$ that
\begin{displaymath}
    \|K_\lambda\| \leq C_3 |\Im \lambda|^{s-1}\,.
\end{displaymath}
Hence, for sufficiently large $|\Im \lambda|$, $I + K_\lambda$ is
invertible. Applying $(-\Delta-\lambda)^{-1}$ to \eqref{eq:101}, however,
yields
\begin{equation}
\label{eq:51}
  (I + K_\lambda)u_3=(-\Delta-\lambda)^{-1}(f_3 -  c R_-(\lambda)  f_1 -  \bar c  R_+(\lambda)   f_2)\,.
\end{equation}
Hence, we get
\begin{displaymath}
   \| u_3\| \leq C |\Im \lambda|^{\max(s,0)-1} \|{\bf f}\| \,.
\end{displaymath}
From the first two lines of \eqref{eq:2.15} we then conclude
\eqref{eq:102} for $u_1$ and $u_2$.
   \end{proof}
  \subsubsection{Point spectrum}
   \begin{proposition}\label{prop.pointsp} 
     Let $\Omega=\C\setminus( \mathbb R_+ \cup \sigma (-\Delta \pm i b_3))$.  Under the
     assumptions of Proposition \ref{prop:bounded-caract-doma}, it
     holds, for any $\lambda\in \Omega$, that $\lambda \in \sigma (\mathcal B_1)$ if and
     only if $-1 \in \sigma (K_\lambda)$.  Moreover, if $-\Delta \pm i b_3$ has a
     compact resolvent, then $\lambda$ is an eigenvalue of $\B_1$. Finally,
     $\lambda$ is an isolated eigenvalue of $\B_1$ of finite multiplicity.
    \end{proposition}
    \begin{proof}
      Let $\lambda \in \Omega$. We begin by the trivial observation that if
      $-1\in \sigma (K_\lambda)$ then by \eqref{eq:101} $(\mathcal B_1-\lambda)^{-1}$
      must be unbounded. If $-1\not\in \sigma (K_\lambda)$, we may conclude from
      (\ref{eq:2.15}a,b) and the boundedness of $R_\pm(\lambda)$ that
      $\lambda\in\rho(\B_1)$.

Suppose now that $R_\pm(\lambda)$ is compact. Then so is $K_\lambda$, and hence,
if $-1 \in \sigma (K_\lambda)$, then $-1$ is an eigenvalue of $K_\lambda$ and $\lambda$ is
an eigenvalue of $\B_1$. Consider the family $\Omega \ni \lambda \mapsto (\tilde{\mathcal
B}_1-\lambda) \in \mathcal L ( D(\tilde{\B}_1,\mathcal L^2(\mathbb
R^k,\mathbb C^3))$ .  By the foregoing discussion,
\begin{displaymath}
  {\rm dim} \, {\rm ker}(\tilde{\B}_1-\lambda)={\rm dim} \, {\rm ker} (K_\lambda
  -1)<\infty \,.
\end{displaymath}
As
\begin{displaymath}
    \tilde{\B}_1^*= -\Delta \otimes I_3 + \tilde M^* \,.
\end{displaymath}
we may apply the same arguments to obtain that $ {\rm dim} \, {\rm
  ker} (\tilde{\mathcal B}_1^*- \bar \lambda)<\infty$, and hence ${\rm dim} \,
{\rm coker} (\tilde{\mathcal B}_1- \lambda)<\infty$. We can then conclude that
$(\tilde{\mathcal B}_1- \lambda)$ is a Fredholm operator for each $\lambda\in\Omega$.
Clearly, $(\tilde{\mathcal B}_1- \lambda)$ is invertible for sufficiently
large $|\Im \lambda|$ or negative $\Re\lambda$. Hence the index, which is constant
in $\Omega$ is zero and we can use either \cite[Proposition 2.3]{SZ} or
\cite[Theorem 2.1]{ka12} (relying on \cite{gosi71}), to obtain that $
(\tilde{\mathcal B}_1-\lambda)^{-1}$ is, (see \cite{ka12}) a
finite-meromorphic family in $\Omega$. The
proposition is proved.
  \end{proof}

    \begin{remark}\label{rem.pointsp}
    In the next sections we consider the case $k=1$ and $b_3(x)=x$. In
    this case, $\sigma (-\Delta \pm i b_3)=\emptyset$ and $\R_+\subset\sigma(\tilde{B}_1)$. By
    Proposition \ref{prop.pointsp} it follows that
    $\sigma(\tilde{B}_1)\cap(\C\setminus R_+)$ is discrete
    \end{remark}

\begin{remark}\label{rem.pointsp.extension}
  We note that one can establish Proposition \ref{prop.pointsp}
  whenever $\bar c R_+(\lambda) c$ is compact. For example, if we consider
  the above case where $k=1$ and $b_3(x)=x$, we may allow for
  $|c|=\OO(|x|^\gamma)$ as $|x|\to\infty$ for some $0\leq\gamma < 1/2$.
    \end{remark}
  
   In the next sections, assuming $k=1$, constant $b_1$, $b_2=0$, and
   $b_3(x) =x$, we obtain much more precise results for $\B_\epsilon$ in the asymptotic
   regime $\epsilon \to 0$.
     \subsection{$\tilde{B}_1$ in the presence of boundary}
\label{sec:case-with-boundary}
     Consider a bounded smooth open set $\mathcal O \subset \mathbb R^k$
     and the Dirichlet realization of $\mathcal P$ in $\mathcal O$
     denoted by $\bf P^\mathcal O$. For $ M\in C^2(\R^k,M_d(\R))$, $\bf
     P^\mathcal O$ is a bounded perturbation of $-\Delta$ and hence $D(\bf
     P^{\mathcal O})=H^2(\mathcal O,\mathbb R^d)\cap H_0^1(\mathcal
     O,\mathbb R^d)$, and $\bf P^{\mathcal O}$ has a compact
     resolvent.  Nevertheless, in the presence of a small
     parameter $\epsilon$, i.e. when
     \begin{displaymath}
      {\bf  P}^{\mathcal O}=-\epsilon^2\Delta+M
     \end{displaymath}
     the behaviour of spectrum and the resolvent in the limit $\epsilon\to0$
     could probably be understood from the analysis of linearized
     operators acting on $\R^k$. This is precisely the case in the two
     dimensional setting \eqref{eq:93} when $\bf P^\mathcal O$ is
     equivalent to the Dirichlet realization of $-\epsilon^2 \Delta + i V$ (see
     for instance \cite{hen15,AGH,Ahen}).

\section{The (1D)-model in $\mathbb R$.}\label{s3}
In this section we consider the operator \eqref{eq:1} acting  on
$\R$. 
\subsection{Problem setting}
In the standard basis of $\mathbb C^3$, the system \eqref{eq:1} reads 
 for ${\bf u} = (u_1,u_2,u_3)$ 
 \begin{equation}
 \B_\epsilon{\bf u}:=\left(
 \begin{array}{ccc}
 -\epsilon^2 \frac{d^2}{dx^2} & x  & 0 \\
 - x & -\epsilon^2 \frac{d^2}{dx^2} & 1 \\
0&  - 1 & -\epsilon^2 \frac{d^2}{dx^2} 
 \end{array}
 \right) 
 \left(\begin{array}{c} u_1\\u_2\\u_3
 \end{array}
 \right)
 \end{equation}
 It has been established in a more general context in either
 Proposition \ref{prop:maximal} or Proposition \ref{prop:bounded} that
 \begin{proposition}
$  \B_\epsilon$ is a closed operator in $L^2(\mathbb R,\mathbb R^3)$ whose domain  is
\begin{equation}\label{eq:1b}
  D(\B_\epsilon)= \{ {\bf u}\in H^2(\R,\mathbb R^3) \,| \, [x{\mathbf b}_1+{\mathbf b}_0]\times
  {\bf u}\in L^2(\R,\mathbb R^3) \,\} \,. 
\end{equation}
 \end{proposition}

 Let ${\bf v}$ be given by \eqref{eq:defv}. We begin by rewriting
 $\B_\epsilon$ in the basis $( {\bf v},\bar{ {\bf v}}, {\mathbf b}_1)$ of
 $\C^3$.  We thus set
\begin{displaymath}  
  \tilde u_1={\bf u}\cdot {\bf v}=\frac{1}{\sqrt{2}}(-i u_1 +u_2) \quad ; \quad   \tilde u_2={\bf u}\cdot \bar{\bf v}=\frac{1}{\sqrt{2}}(iu_1+u_2)
  \quad ; \quad \tilde u_3 = {\bf u}\cdot {\bf b}_1=u_3 \,. 
\end{displaymath}
In this new basis the operator $\widetilde \B_\epsilon$  becomes 
\begin{equation}
\label{eq:124}
 \widetilde B_\epsilon:=\left(
 \begin{array}{ccc}
 -\epsilon^2 \frac{d^2}{dx^2}+ ix & 0  & \frac{1}{\sqrt 2}  \\
 0 & -\epsilon^2 \frac{d^2}{dx^2} - ix & \frac{1}{\sqrt 2}  \\
-\frac{1}{\sqrt 2}  &  -\frac{1}{\sqrt 2}  & -\epsilon^2 \frac{d^2}{dx^2} 
 \end{array}
 \right) \,.
 \end{equation}
We attempt to obtain resolvent estimates for the  problem
  \begin{equation}
\label{eq:65}
  (\widetilde \B_\epsilon - \Lambda) \tilde{\bf u} = \epsilon^\frac 23 \tilde{\bf f}\,.
  \end{equation}

Applying the transformation
\begin{displaymath}
  x\to\epsilon^{2/3}x 
\end{displaymath}
yields, 
\begin{equation}
\label{eq:60}
  (\check \B_\varepsilon-\lambda){\bf \check  u} = {\bf \check  f}\,,
\end{equation}
where, with 
\begin{equation}
\label{eq:49}
\varepsilon =\epsilon^\frac 43 \mbox{ and } \lambda = \epsilon^{-\frac 23} \Lambda\,
\end{equation}
  $\check \B_\varepsilon$ is given by
\begin{equation}
 \check  B_\varepsilon\Big(x,\frac{d}{dx}\Big) :=\left(
 \begin{array}{ccc}
 - \frac{d^2}{dx^2}+ ix & 0  & \frac{1}{\sqrt 2} \varepsilon^{-\frac 12} \\
 0 & - \frac{d^2}{dx^2} - ix & \frac{1}{\sqrt 2}\varepsilon^{-\frac 12}  \\
-\frac{1}{\sqrt 2} \varepsilon^{-\frac 12} &  -\frac{1}{\sqrt 2} \varepsilon^{-\frac 12} & -\frac{d^2}{dx^2} 
 \end{array}
 \right) \,.
 \end{equation}
Equivalently we may write the spectral equation in the form
 \begin{equation}  \label{eq:3} 
\left\{ \begin{array}{c}
    (-\frac{d^2}{dx^2} +i \, x -\lambda) \check u_1+  \frac{1}{\sqrt 2} \varepsilon^{-\frac 12} \check u_3 = \check  f_1 \,,\\
    (-\frac{d^2}{dx^2}- i \,x -\lambda) \check u_2+  \frac{1}{\sqrt 2}\varepsilon^{-\frac 12}  \check u_3 = \check f_2 \,, \\
   ( -\frac{d^2}{dx^2}-\lambda) \check u_3 -\frac{1}{\sqrt{2}}\varepsilon^{-\frac 12} (\check u_1+\check u_2) = \check  f_3 \,.
    \end{array}
    \right.
      \end{equation}
We omit the accents of $u_i$ and $f_i$ ($i=1,2,3$) in the sequel. 

Let the Fourier transform of ${\bf u}$ be defined in the following
manner
\begin{displaymath}
 \hat{u}_i(\omega) = \F(u_i)(\omega) =\frac{1}{\sqrt{2\pi}}\int_\Re e^{-i\omega x}u_i(x)\,dx
  \quad i=1,2,3\,.
\end{displaymath}
Applying $\F$ to \eqref{eq:3} yields
\begin{subequations}
\label{eq:5}
  \begin{empheq}[left={\empheqlbrace}]{alignat=2}
    &(\omega^2-\lambda)\hat{u}_1 - \frac{d\hat{u}_1}{d\omega}+  \frac{1}{\sqrt 2} \varepsilon^{-1/2}\hat{u}_3 = \hat{f}_1 & \text{ in } \R \\
    &(\omega^2-\lambda)\hat{u}_2 + \frac{d\hat{u}_2}{d\omega}+
    \frac{1}{\sqrt 2} \varepsilon^{-1/2}\hat{u}_3 = \hat{f}_2 & \text{ in } \R
    \\ 
    &(\omega^2-\lambda)\hat{u}_3 -\frac{\varepsilon^{-1/2}}{\sqrt{2}}(\hat{u}_1+\hat{u}_2) = \hat{f}_3 & \text{ in } \R \,.
  \end{empheq}
\end{subequations}
We search for solutions in $X=X_1\times X_1\times X_3$, where
\begin{equation}\label{eq:defX}
  X_1=\{u\in H^1(\R)\,|\, \omega^2u\in L^2(\R)\}\quad ; \quad  X_3=\{u\in L^2(\R)\,|\, \omega^2u\in L^2(\R)\}\,.
\end{equation}
We now set
\begin{equation}\label{eq:defhatu}
  \hat{u}_d=\hat{u}_1 -\hat{u}_2 \quad ; \quad  \hat{u}_s =\hat{u}_1 +\hat{u}_2 \quad ;\quad 
    \hat{f}_d=\hat{f}_1 -\hat{f}_2 \quad ; \quad  \hat{f}_s =\hat{f}_1 +\hat{f}_2 \;.
\end{equation}
Subtracting (\ref{eq:5}a) from  (\ref{eq:5}b) yields
\begin{equation}
  \label{eq:6}
- \frac{d\hat{u}_s}{d\omega} + (\omega^2-\lambda)\hat{u}_d= \hat{f}_d \,.
\end{equation}
Summing up (\ref{eq:5}b) and  (\ref{eq:5}a) yields with the aid of
(\ref{eq:5}c), assuming $\lambda \notin \mathbb R_+$,
\begin{equation}
  \label{eq:7}
\frac{d\hat u_d}{d\omega}- \Big[(\omega^2-\lambda)+\frac{\varepsilon ^{-1}}{\omega^2-\lambda}\Big]\hat{u}_s =
-\hat{f}_s + \sqrt{2} \varepsilon^{-1/2}\frac{\hat{f}_3}{\omega^2-\lambda} \,.
\end{equation}
Extracting $\hat{u}_d$ from \eqref{eq:6} and then substituting into
\eqref{eq:7} we obtain
\begin{equation}
\label{eq:8}
- \frac{d}{d\omega}\Big(\frac{1}{\omega^2-\lambda}\frac{d\hat{u}_s}{d\omega}\Big) +
 \Big[(\omega^2-\lambda) +  \frac{\varepsilon^{-1}}{\omega^2-\lambda}\Big]\hat{u}_s =  g \,,
 \end{equation}
 with
 \begin{equation}\label{eq:8b} 
 g:=  \hat{f}_s + \varepsilon^{-1/2}\frac{\hat{f}_3}{\omega^2-\lambda}
 +  \frac{d}{d\omega}\Big(\frac{\hat{f}_d}{\omega^2-\lambda}\Big)\,.  
\end{equation}

\subsection{The case $ 0 < | \Im \lambda |<\varepsilon^{-1/2}$}
 We write
\begin{displaymath}
\lambda_r =\Re \lambda \mbox{ and } \lambda_i = \Im \lambda\,.
\end{displaymath}
\begin{proposition} 
\label{prop:strip-estimate}
For any  $0<\delta\leq1/2$, there exist $\varepsilon_0>0$ and
  $C>0$ such that for all $0 < \varepsilon<\varepsilon_0$ and any triple $({\bf u},{\bf
    f},\lambda)$ satisfying $ 0 < | \lambda_i|\leq (1-2\delta^4)^{1/2}\varepsilon^{-1/2}$ and
  \eqref{eq:3}, it holds that
  \begin{equation}
    \label{eq:9} 
\|u_1+u_2\|_2 \leq C\varepsilon^{1/2}\Big(1+\frac{[1+\varepsilon^{1/2}(\lambda_r)_+^{1/2}]}{| \lambda_i|}\Big)\|{\bf f}\|_2 \,.
  \end{equation}
\end{proposition}
\begin{proof}
 Without any loss of generality we assume $\Im\lambda>0$, as the
transformation
\begin{displaymath}
  \lambda\to\bar{\lambda} \quad ; \quad \hat{u}_s\to\bar{\hat u}_s \quad ; \quad {\mathbf f}
  \to \bar{\mathbf f}
\end{displaymath}
leaves \eqref{eq:8} unaltered. \\

We split the discussion into three different cases depending on the
value of $\lambda_r$. We begin, however, by obtaining some
identities and inequalities that are valid in all cases. \\
 
 \paragraph{Preliminary inequalities}~\\
 Taking the
inner product, in $L^2(\R)$, of \eqref{eq:8} with $\hat{u}_s$ yields,
for the imaginary part 
\begin{displaymath}
  \Im\Big\langle \hat{u}_s^\prime,\frac{\hat{u}_s^\prime}{\omega^2-\lambda}\Big\rangle -  \lambda_i\|\hat{u}_s\|_2^2 +
  \varepsilon^{-1}\Im\Big\langle \hat{u}_s,\frac{\hat{u}_s}{\omega^2-\lambda}\Big\rangle =- \Im\langle \hat{u}_s,g\rangle \,,
\end{displaymath}
Consequently, it holds that
\begin{equation}
\label{eq:12}
  \Big\|\frac{\hat{u}_s^\prime}{\omega^2-\lambda}\Big\|_2^2 - \|\hat{u}_s\|_2^2 +
  \varepsilon^{-1}\Big\|\frac{\hat{u}_s}{\omega^2-\lambda}\Big\|_2^2 =- \frac{1}{ \lambda_i}\Im\langle \hat{u}_s,g\rangle \,.
\end{equation}
Next, we estimate the inner product $\langle \hat{u}_s,g\rangle$ with $g$ defined in \eqref{eq:8b}. Clearly,
\begin{displaymath}
  \Big|\Big \langle \hat{u}_s,\frac{\hat{f}_3}{\omega^2-\lambda}\Big\rangle\Big|\leq
  \Big\|\frac{\hat{u}_s}{\omega^2-\lambda}\Big\|_2 \|f_3\|_2 \,. 
\end{displaymath}
Furthermore, integration by parts yields that
\begin{displaymath}
  \Big|\Big
  \langle \hat{u}_s,\Big(\frac{\hat{f}_d}{\omega^2-\lambda}\Big)^\prime\Big\rangle\Big|\leq
  \Big\|\frac{\hat{u}_s^\prime}{\omega^2-\lambda}\Big\|_2 \|f_d\|_2 \,. 
\end{displaymath}
Hence,
\begin{equation}
\label{eq:109}
  |\langle \hat{u}_s,g\rangle|\leq
\|\hat{u}_s\|_2\|f_s \|_2+\Big\|\frac{(\hat{u}_s)^\prime}{\omega^2-\lambda}\Big\|_2
\|f_d\|_2 + \varepsilon^{-1/2}\Big\|\frac{\hat{u}_s}{\omega^2-\lambda}\Big\|_2 \|f_3\|_2 \,. 
\end{equation}
By \eqref{eq:109} we have
\begin{equation*}
\begin{array}{ll}
  |\langle \hat{u}_s,g\rangle| & \leq \sqrt{    
\|\hat{u}_s\|_2^2+\Big\|\frac{(\hat{u}_s)^\prime}{\omega^2-\lambda}\Big\|_2^2 + \varepsilon^{-1}\Big\|\frac{\hat{u}_s}{\omega^2-\lambda}\Big\|_2^2}  \sqrt{ \|f_3\|^2 + \|f_s\|^2 + \|f_d\|^2} \\ & \leq \sqrt{    
\|\hat{u}_s\|_2^2+\Big\|\frac{(\hat{u}_s)^\prime}{\omega^2-\lambda}\Big\|_2^2 + \varepsilon^{-1}\Big\|\frac{\hat{u}_s}{\omega^2-\lambda}\Big\|_2^2} \,  \sqrt{2} \,  \|{\bf  f} \|_2 \,. 
\end{array}
\end{equation*}
With the aid of \eqref{eq:12} we can conclude that
\begin{equation*}
  |\langle \hat{u}_s,g\rangle|\leq \sqrt{2}\, \sqrt{    
2 \|\hat{u}_s\|_2^2+ \frac{1}{|\lambda_i|}|\langle \hat{u}_s,g\rangle| } \,   \| {\bf f} \|_2 \,. 
\end{equation*}
Taking the square then yields
\begin{equation*}
  |\langle \hat{u}_s,g\rangle|^2 \leq    2 
\Big(2 \|\hat{u}_s\|_2^2+ \frac{1}{|\lambda_i|}|\langle \hat{u}_s,g\rangle| \Big) \,  \| {\bf f} \|_2^2 \,, 
\end{equation*}
which can be rewritten in the form
\begin{equation*}
 ( |\langle \hat{u}_s,g\rangle| - \frac{1}{\lambda_i} \|{\bf f}  \|_2^2) ^2 \leq (4 \|\hat{u}_s\|_2^2 + \frac{\|{\bf f} \|^2}{\lambda_i^2} ) \|{\bf f} \|_2^2 \,.
\end{equation*}
Consequently,
\begin{equation*}
 |\langle \hat{u}_s,g\rangle|  \leq  \left(\frac{\|{\bf f}\|_2}{ \lambda_i}   + \sqrt{4
     \|\hat{u}_s\|_2^2 + \frac{\|{\bf f}\|_2^2}{ \lambda_i^2}} \right)\,
 \|{\bf f}\|_2   \,, 
\end{equation*}
and hence
\begin{equation} \label{eq:15}
 |\langle \hat{u}_s,g\rangle|  \leq 2  \left(\frac{\|{\bf f}\|_2}{\lambda_i}   +  \|\hat{u}_s\|_2 \right) \,  \|{\bf f}\|_2  \,.
\end{equation}

Next, taking the inner product of \eqref{eq:8} from the left with
$(\omega^2-\bar{\lambda})\hat{u}_s$ yields for the real part,
\begin{multline}
\label{eq:22}
 \Re\langle(\omega^2-\bar{\lambda})\hat{u}_s,g\rangle =   \|(\hat{u}_s)^\prime\|_2^2 +
 \Re \Big\langle2\omega \hat{u}_s,\frac{(\hat{u}_s)^\prime}{\omega^2-\lambda}\Big\rangle\\  +(\varepsilon^{-1}- \lambda_i^2)\|\hat{u}_s\|_2^2 +
  \|(\omega^2-\lambda_r)\hat{u}_s\|_2^2 \,.
\end{multline}
Finally, by \eqref{eq:8b} we have that
\begin{displaymath} 
  |\langle(\omega^2-\bar{\lambda})\hat{u}_s,g\rangle| \leq \|(\omega^2-\lambda)\hat{u}_s\|_2\|f_s
  \|_2+\Big(\|\hat{u}_s^\prime\|_2+ \Big\|\frac{2\omega}{\omega^2-\lambda}\hat{u}_s\Big\|_2\Big)
\|f_d\|_2 + \varepsilon^{-1/2}\|\hat{u}_s\|_2 \|f_3\|_2 \,,
\end{displaymath}
and since 
\begin{equation}
\label{eq:17} 
\Big\|\frac{\omega\,  \hat{u}_s}{\omega^2-\lambda}\Big\|_2\leq
(\lambda_r)_+^{1/2}\Big\|\frac{\hat{u}_s}{\omega^2-\lambda}\Big\|_2+\Big\|\frac{\hat{u}_s}{|\omega^2-\lambda|^{1/2}}\Big\|_2
\,, 
\end{equation}
we obtain that
\begin{multline}
\label{eq:110}
  |\langle(\omega^2-\bar{\lambda})\hat{u}_s,g\rangle| \leq \varepsilon^{-1/2}\|\hat{u}_s\|_2\|f_3\|_2 +
  \|(\omega^2-\lambda)\hat{u}_s\|_2\|f_s\|_2\\ +\Big(\|\hat{u}_s^\prime\|_2+ 
  (\lambda_r)_+^{1/2}\Big\|\frac{\hat{u}_s}{\omega^2-\lambda}\Big\|_2+
  \Big\|\frac{\hat{u}_s}{|\omega^2-\lambda|^{1/2}}\Big\|_2 \Big) \|f_d\|_2\,. 
\end{multline}
We now observe that
\begin{equation}
\label{eq:116}
  \Big|\Big\langle2\omega \hat{u}_s,\frac{\hat{u}_s^\prime}{\omega^2-\lambda}\Big\rangle\Big|\leq
  2\|\omega \hat{u}_s\|_2\Big\|\frac{\hat{u}_s^\prime}{\omega^2-\lambda}\Big\|_2 \,.
\end{equation}
Set,
\begin{equation}
\label{eq:115}
  \lambda_{r,m}=\max(\lambda_r,1)\,.
\end{equation}
Next we write, with the aid of Cauchy's inequality,
\begin{equation}
\label{eq:Cauchy}
\begin{array}{ll}
    2\|\omega \hat{u}_s\|_2 & \leq2\|\omega^2\hat{u}_s\|^{1/2}\|\hat{u}_s\|_2^{1/2}\\
    & \leq\lambda_{r,m}^{-1/2}\|\omega^2\hat{u}_s\|_2+\lambda_{r,m}^{1/2}\|\hat{u}_s\|_2\\& \leq 
    \lambda_{r,m}^{-1/2}[\|(\omega^2-(\lambda_r)_+\hat{u}_s\|_2+\lambda_{r,m}^{-1/2}(\lambda_r)_+ \|\hat{u}_s\|_2
    ]+\lambda_{r,m}^{1/2}\|\hat{u}_s\|_2 \\ &\leq 
    \lambda_{r,m}^{-1/2}\|(\omega^2-\lambda_r)\hat{u}_s\|_2+ 2\lambda_{r,m}^{1/2}\|\hat{u}_s\|_2 \,.
    \end{array}
\end{equation}
Substituting \eqref{eq:110} into \eqref{eq:22} yields, with the aid of
\eqref{eq:116} and \eqref{eq:Cauchy} that
\begin{equation}
\label{eq:111}
\begin{array}{l}
    \|(\hat{u}_s)^\prime\|_2^2 +(\varepsilon^{-1}- \lambda_i^2)\|\hat{u}_s\|_2^2 +
    \|(\omega^2-\lambda_r)\hat{u}_s\|_2^2 \\
   \qquad \quad   \leq \varepsilon^{-1/2}\|\hat{u}_s\|_2\|f_3\|_2 + \|(\omega^2-\lambda)\hat{u}_s\|_2\|f_s\|_2 \\
    \qquad    \qquad   + (\lambda_{r,m}^{-1/2}\|(\omega^2- \lambda_r )\hat{u}_s\|_2+
  2\lambda_{r,m}^{1/2}\|\hat{u}_s\|_2)\Big\|\frac{\hat{u}_s^\prime}{\omega^2-\lambda}\Big\|_2 \\
     \qquad    \qquad  +\Big(\|\hat{u}_s^\prime\|_2+ 
  \lambda_{r,m}^{1/2}\Big\|\frac{\hat{u}_s}{\omega^2-\lambda}\Big\|_2+\Big\|\frac{\hat{u}_s}{|\omega^2-\lambda|^{1/2}}\Big\|_2 \Big) \|f_d\|_2\,. 
\end{array}
\end{equation}
By \eqref{eq:12} and \eqref{eq:15} we have
\begin{displaymath}
   \Big\|\frac{\hat{u}_s^\prime}{\omega^2-\lambda}\Big\|_2^2 \leq
   \|\hat{u}_s\|_2^2+\frac{1}{\lambda_i}  |\langle \hat{u}_s,g\rangle|\leq
   \|\hat{u}_s\|_2^2+\frac{2}{\lambda_i} \left(\frac{\|{\bf f}\|_2}{\lambda_i}
     +  \|\hat{u}_s\|_2 \right) \,  \|{\bf f}\|_2  \,. 
\end{displaymath}
Consequently,
\begin{displaymath}
  \Big\|\frac{\hat{u}_s^\prime}{\omega^2-\lambda}\Big\|_2^2\leq
  \Big(\|\hat{u}_s\|_2+\frac{\|{\bf f}\|_2}{\lambda_i}\Big)^2+\frac{\|{\bf
  f}\|_2^2}{\lambda_i^2} \,,
\end{displaymath}
and hence
\begin{equation}
\label{eq:112}
  \Big\|\frac{\hat{u}_s^\prime}{\omega^2-\lambda}\Big\|_2 \leq \|\hat{u}_s\|_2 +
  \frac{2}{ \lambda_i}\|\bf f\|_2 \,. 
\end{equation}
By Cauchy inequality  we have that for any $\alpha >0$ 
\begin{displaymath}
  \Big\|\frac{\hat{u}_s}{|\omega^2-\lambda|^{1/2}}\Big\|_2\leq \frac{1}{\sqrt{2}} \left(\alpha\Big\|\frac{\hat{u}_s}{\omega^2-\lambda}\Big\|_2+
  \frac 1 \alpha \|\hat{u}_s\|_2\right)\,,
\end{displaymath}
which in particular implies
\begin{equation}\label{eq:3.23}
  \Big\|\frac{\hat{u}_s}{|\omega^2-\lambda|^{1/2}}\Big\|_2\leq 
 \lambda_{r,m}^{1/2} \, \Big\|\frac{\hat{u}_s}{\omega^2-\lambda}\Big\|_2+
  \lambda_{r,m}^{-1/2}\|\hat{u}_s\|_2
\end{equation}
Using \eqref{eq:12} and \eqref{eq:15} once again we obtain that
\begin{equation}
\label{eq:117}
   \Big\|\frac{\hat{u}_s}{\omega^2-\lambda}\Big\|_2 \leq \varepsilon^{1/2}\|\hat{u}_s\|_2 +
  \frac{2\varepsilon^{1/2}}{ \lambda_i}\|{\bf f}\|_2 \,. 
\end{equation}
Substituting \eqref{eq:117}
into \eqref{eq:3.23} yields
\begin{equation}
\label{eq:118}
  \Big\|\frac{\hat{u}_s}{|\omega^2-\lambda|^{1/2}}\Big\|_2\leq
  (\varepsilon^{1/2}\lambda_{r,m}^{1/2}+\lambda_{r,m}^{-1/2})\|\hat{u}_s\|_2 +
  \frac{2\varepsilon^{1/2}\lambda_{r,m}^{1/2}}{ \lambda_i}\|{\bf f}\|_2    \,.
\end{equation}
Substituting \eqref{eq:118}, 
 together with \eqref{eq:117} and
\eqref{eq:112} into \eqref{eq:111},  we obtain that
\begin{equation*}
\begin{array}{l}
   \|(\hat{u}_s)^\prime\|_2^2 +(\varepsilon^{-1}-\lambda_i^2)\|\hat{u}_s\|_2^2+
    \|(\omega^2-\lambda_r)\hat{u}_s\|_2^2\\
    \qquad \leq  \varepsilon^{-1/2}\|\hat{u}_s\|_2\|f_3\|_2
      +\|(\omega^2-\lambda_r)\hat{u}_s\|_2\|f_s\|_2\\
     \quad  \qquad +\lambda_i\|\hat{u}_s\|_2\|f_s\|_2+ [\lambda_{r,m}^{-1/2}\|(\omega^2-\lambda_r)\hat{u}_s\|_2+
  2\lambda_{r,m}^{1/2}\|\hat{u}_s\|_2]\Big(\|\hat{u}_s\|_2 +
  \frac{2}{ \lambda_i}\|{\bf f}\|_2\Big)\\ 
 \quad  \qquad +\Big(\|\hat{u}_s^\prime\|_2+(2\varepsilon^{1/2}\lambda_{r,m}^{1/2}+\lambda_{r,m}^{-1/2})\|\hat{u}_s\|_2 +
  \frac{4\varepsilon^{1/2}\lambda_{r,m}^{1/2}}{ \lambda_i}\|{\bf f}\|_2  \Big) \|f_d\|_2\,.  
\end{array}
\end{equation*}
From the above  and using the condition on $\lambda_i$, we can conclude that there exists $C>0$ such that
\begin{displaymath}
  \|(\hat{u}_s)^\prime\|_2+ \|(\omega^2-\lambda_r)\hat{u}_s\|_2\leq
  C\Big(\lambda_{r,m}^{1/4}\|\hat{u}_s\|_2+ \Big[1+
  \frac{ \varepsilon^{1/2}\lambda_{r,m}^{1/2}+ \lambda_{r,m}^{-\frac 12} }{ \lambda_i} \Big]\|{\bf f}\|_2\Big) 
\end{displaymath}
which, when substituted, together with \eqref{eq:117} and
\eqref{eq:118}, into \eqref{eq:110}  yields 
\begin{equation*}
\begin{array}{l}
  |\langle(\omega^2-\bar{\lambda})\hat{u}_s,g\rangle| \\
  \qquad \leq C\Big(\Big[1+ \frac{ \varepsilon^{1/2}\lambda_{r,m}^{1/2}+
    \lambda_{r,m}^{-\frac 12} }{ \lambda_i}   \Big]\|{\bf
    f}\|_2+\lambda_{r,m}^{1/4}\|\hat{u}_s\|_2\Big)(\|f_s\|_2 + \|f_d\|_2)\\  \quad \qquad   +
   \varepsilon^{-1/2}\|\hat{u}_s\|_2\|f_3\|_2 +
  \Big[(2\varepsilon^{1/2}\lambda_{r,m}^{1/2}+\lambda_{r,m}^{-1/2})\|\hat{u}_s\|_2 +
  \frac{4\varepsilon^{1/2}\lambda_{r,m}^{1/2}}{ \lambda_i}\|{\bf f}\|_2  \Big]\|f_d\|_2 \,. 
  \end{array}
\end{equation*}
Hence,
\begin{multline}
\label{eq:24}
  |\langle(\omega^2-\bar{\lambda})\hat{u}_s,g\rangle|  \leq C\Big(\Big[1+\frac{\varepsilon^{1/2}\lambda_{r,m}^{1/2} +\lambda_{r,m}^{-\frac 12} }{\lambda_i}   \Big]\|{\bf
    f}\|_2+  \big(  \lambda_{r,m}^{1/4}+\varepsilon^{-1/2}+ \varepsilon^{1/2}\lambda_{r,m}^{1/2} +  \lambda_{r,m}^{-\frac 12}  \big)\|\hat{u}_s\|_2\Big)\|{\bf f}\|_2
\end{multline}

\paragraph{Case 1: $ 0\leq \lambda_r\leq \delta^2 \varepsilon^{-\frac 12 }/8$.}~\\
Let
$\eta\in C_0^\infty(\R,[0,1])$,  satisfy
$\tilde{\eta}=\sqrt{1-\eta^2}\in C^\infty(\R,[0,1])$ and 
\begin{equation}
\label{eq:11}
  \eta(x)=
  \begin{cases}
    0 & x<\frac{1}{2} \\
    1 & x>1 \,.
  \end{cases}
\end{equation}
Let further 
\begin{displaymath}
\eta_\varepsilon(\omega)=\eta(\varepsilon^{\frac{1}{4}} |\omega|/\delta) \mbox{  and }  \tilde{\eta}_\varepsilon(\omega)=\tilde{\eta}(\varepsilon^\frac 14 |\omega|/\delta)\,.
\end{displaymath}
  Taking the
inner product, in $L^2(\R)$, of \eqref{eq:8} with $\eta^2_\varepsilon \, \hat{u}_s$
yields, for the real part
\begin{multline}
\label{eq:20}
  \Big\|\frac{\eta_\varepsilon \, \hat{u}_s^\prime}{\omega^2-\lambda}[\omega^2-\lambda_r]^{1/2}\Big\|_2^2 + 2\Re
  \Big\langle\eta_\varepsilon^\prime  \hat{u}_s,\frac{\eta_\varepsilon \hat{u}_s^\prime}{\omega^2-\lambda}\Big\rangle + \\
 + \varepsilon^{-1} \Big\|\frac{\eta_\varepsilon \hat{u}_s}{\omega^2-\lambda}[\omega^2-\lambda_r]^{1/2}\Big\|_2^2 +
  \|[\omega^2-\lambda_r]^{1/2}\eta_\varepsilon \hat{u}_s\|_2^2= \Re\langle\eta_\varepsilon \hat{u}_s,\eta_\varepsilon g\rangle \,.
\end{multline}
As $  \lambda_r\leq\delta^2\varepsilon^{-\frac 12}/8$, it follows that $ (\omega^2-\lambda_r)
\geq \delta^2\varepsilon^{-\frac 12}/8$ on the support of $\eta_\epsilon$, and hence there exists
$C>0$ such that
\begin{equation}
\label{eq:10}
  \|\eta_\varepsilon \hat{u}_s\|_2^2 \leq C\varepsilon^{\frac 12}(\varepsilon \|\hat{u}_s\|_2^2 +
  |\langle\eta_\varepsilon \hat{u}_s,\eta_\varepsilon g\rangle| ) \,,
\end{equation}
From \eqref{eq:12} we can thus conclude that, for sufficiently small $\varepsilon$, 
\begin{equation}\label{eq:12a}
\begin{array}{ll}
  \|\tilde{\eta}_\varepsilon \hat{u}_s\|_2^2 & \leq
  (1-\delta^4)\varepsilon^{-1}\Big\|\tilde{\eta}_\varepsilon\frac{\hat{u}_s}{\omega^2-\lambda}\Big\|_2^2 \\ &\leq 
  (1-\delta^4)\varepsilon^{-1}\Big\|\frac{\hat{u}_s}{\omega^2-\lambda}\Big\|_2^2\\ &  \leq
  (1-\delta^4)\|\hat{u}_s\|_2^2 + \frac{1}{\lambda_i}|\langle \hat{u}_s,g\rangle| 
\end{array}
\end{equation}
To obtain the inequality we need to use the fact that on the support
of $\tilde \eta_{\varepsilon}$ we have, for $\varepsilon$ sufficiently small,
  \begin{displaymath}
    |\omega^2-\lambda|^2 \leq \delta^4\varepsilon^{-1}+(1-2\delta^4)\varepsilon^{-1}  
 = (1-\delta^4)\varepsilon^{-1}\,.
  \end{displaymath}
We use \eqref{eq:12} to obtain the above inequality.\\

Combining \eqref{eq:12a}  with \eqref{eq:10} yields, for any $\delta \in (0,\frac 12]$,  the existence of
$\varepsilon_0>0$ and $C>0$ such that, for $\varepsilon \in (0,\varepsilon_0)$,  
\begin{equation}
\label{eq:13}
  \|\hat{u}_s\|_2^2 \leq C(\varepsilon^{\frac 12}|\langle\eta_\varepsilon \hat{u}_s,\eta_\varepsilon g\rangle|+
  \frac{1}{ \lambda_i}|\langle \hat{u}_s,g\rangle|) \,.
\end{equation}
Combining the above with \eqref{eq:12} yields
\begin{equation}
  \label{eq:14}
  \Big\|\frac{(\hat{u}_s)^\prime}{\omega^2-\lambda}\Big\|_2^2 +
  \varepsilon^{-1}\Big\|\frac{\hat{u}_s}{\omega^2-\lambda}\Big\|_2^2 \leq C\Big(\epsilon^{\frac 12}|\langle\eta_\varepsilon \hat{u}_s,\eta_\varepsilon g\rangle|+\frac{1}{ \lambda_i}|\langle \hat{u}_s,g\rangle|\Big) \,.
\end{equation}
Next, we estimate $\langle\eta_\varepsilon \hat{u}_s,\eta_\varepsilon g\rangle$. As in the proof of
\eqref{eq:15} we may write that
\begin{displaymath}
   \Big|\Big \langle\eta_\varepsilon \hat{u}_s,\eta_\epsilon\frac{\hat{f}_3}{\omega^2-\lambda}\Big\rangle\Big|\leq
   \Big\|\frac{\hat{u}_s}{\omega^2-\lambda}\Big\|_2 \|f_3\|_2 \,.  
\end{displaymath}
Furthermore, we have that 
\begin{displaymath}
   \Big|\Big
  \langle\eta_\varepsilon^2\hat{u}_s,\Big(\frac{\hat{f}_d}{\omega^2-\lambda}\Big)^\prime\Big\rangle\Big|\leq
  \Big\|\frac{\hat{u}_s^\prime}{\omega^2-\lambda}\Big\|_2 \|f_d\|_2+C\varepsilon^\frac 14  \Big\|\frac{\eta'(\epsilon^{1/4} \,\cdot\,)\hat{u}_s}{\omega^2-\lambda}\Big\|_2 \|f_d\|_2 \,. 
\end{displaymath}
Consequently we may write, for sufficiently small $\varepsilon$
\begin{equation}
\label{eq:36}
  |\langle\eta_\varepsilon \hat{u}_s,\eta_\varepsilon g\rangle| \leq (1 + C \varepsilon^\frac 14) \|\hat{u}_s\|_2\|f_s\|_2+\Big\|\frac{\hat{u}_s^\prime}{\omega^2-\lambda}\Big\|_2
\|f_d\|_2 + 2\varepsilon^{-1/2}\Big\|\frac{\hat{u}_s}{\omega^2-\lambda}\Big\|_2 \|f_3\|_2 \,. 
\end{equation}
Substituting the above  together with \eqref{eq:15} into
\eqref{eq:13} and \eqref{eq:14} then yields 
\begin{equation}
  \label{eq:16}
 \Big\|\frac{(\hat{u}_s)^\prime}{\omega^2-\lambda}\Big\|_2^2 +
  \varepsilon^{-1}\Big\|\frac{\hat{u}_s}{\omega^2-\lambda}\Big\|_2^2 +   \|\hat{u}_s\|_2^2
  \leq C \left( \frac{1}{ \lambda_i^2} +\varepsilon\right) \|{\mathbf f}\|_2^2 \leq   \frac{\hat C}{ \lambda_i^2} \,  \|{\mathbf f}\|_2^2\,,
\end{equation}
where for the last inequality we have used the assumption on $|\lambda_i|$.\\
Consequently, by \eqref{eq:22}, as $ | \lambda_i
|<(1-2\delta^4)^{1/2}\varepsilon^{-1/2}$, and by \eqref{eq:17} it holds that
\begin{equation}\label{eq:17a}
\begin{array}{l}
  \|\hat{u}_s^\prime\|_2^2 +   \| (\omega^2-\lambda_r) \hat{u}_s\|_2^2 +\delta^4\varepsilon^{-1}\|\hat{u}_s\|_2^2\\[1.5ex]
  \qquad  \leq  2
  \|\hat{u}_s^\prime\|_2\Big(\lambda_r^{1/2}\Big\|\frac{\hat{u}_s}{\omega^2-\lambda}\Big\|_2 +
  \Big\|\frac{\hat{u}_s}{|\omega^2-\lambda|^{1/2}}\Big\|_2 \Big)+
  2|\langle(\omega^2-\bar{\lambda})\hat{u}_s,g\rangle|  \,.
  \end{array}
\end{equation} 
From \eqref{eq:16} we get
\begin{displaymath}
  \Big\|\frac{\hat{u}_s}{|\omega^2-\lambda|^{1/2}}\Big\|_2^2 \leq
  \Big\|\frac{\hat{u}_s}{\omega^2-\lambda }\Big\|_2\|\hat{u}_s\|_2\leq {C\, \frac{\varepsilon^{\frac 12}}{ \lambda_i^2} \, \|{\mathbf f}\|_2^2\,}
\end{displaymath}
and, for  $0 \leq \lambda_r \leq  \delta^2 \varepsilon^{-\frac 12 }/8\,$,
\begin{equation}\label{eq:3.33}
\lambda_r \Big\|\frac{\hat{u}_s}{\omega^2-\lambda}\Big\|_2^2 \leq \hat C \lambda_r \frac{\varepsilon}{\lambda_i^2} \| f\|^2 \leq \check C \frac{\varepsilon^{\frac 12}}{ \lambda_i^2} \, \|{\mathbf f}\|_2^2\,.
\end{equation}
By  \eqref{eq:17a}  we now obtain that
\begin{equation}
\label{eq:18}
  \|\hat{u}_s^\prime\|_2^2 +   \|(\omega^2-\lambda)\hat{u}_s\|_2^2+\varepsilon^{-1}\|\hat{u}_s\|_2^2  \leq
  C\left(|\langle(\omega^2-\bar \lambda)\hat{u}_s,g\rangle| +\frac{\varepsilon^{\frac 12}}{ \lambda_i^2}\|{\mathbf f}\|_2^2\right)\,.
\end{equation}
We note that \eqref{eq:24} implies in Case 1
\begin{multline}
\label{eq:24case1}
  |\langle(\omega^2-\bar{\lambda})\hat{u}_s,g\rangle|  \leq C\Big(\Big[1+\frac{\varepsilon^{1/2}\lambda_{r,m}^{1/2} +\lambda_{r,m}^{-\frac 12}}{\lambda_i}  \Big]\|{\bf
    f}\|_2+ \varepsilon^{-1/2} \|\hat{u}_s\|_2\Big)\|{\bf
    f}\|_2
\end{multline}
Substituting \eqref{eq:24case1} into \eqref{eq:18} yields,  for any
  $\delta>0$, the existence of positive $C$, and $\varepsilon_0$  such
that for $\varepsilon \in (0,\varepsilon_0)$, it holds that 
\begin{equation}
\label{eq:39}
  \|\hat{u}_s^\prime\|_2^2 +
  \|(\omega^2-\lambda)\hat{u}_s\|_2^2+\varepsilon^{-1}\|\hat{u}_s\|_2^2 \leq C\Big(\Big[1+
  \frac{\varepsilon^{\frac 12}}{ \lambda_i^2} +  \frac{1}{\lambda_i} \Big]\|{\mathbf f}\|_2^2+
  \varepsilon^{-1/2} \|\hat{u}_s\|_2\|{\bf f}\|_2\Big)\,, 
\end{equation}
and hence
\begin{equation}
\label{eq:19}
  \|\hat{u}_s\|_2^2 \leq C\varepsilon\Big[1+
  \frac{\varepsilon^{\frac 12}}{ \lambda_i^2} +\frac{1}{\lambda_i} \Big]\|{\mathbf f}\|_2^2\,, 
\end{equation}
establishing, thereby, \eqref{eq:9} in this case. \\

\paragraph{Case 2: $\lambda_r>\delta^2 \varepsilon^{-\frac 12 }/8 \,.$ }~\\

We begin by the simple observation that in this case $\lambda_{r,m}=\lambda_r$.
Then we use \eqref{eq:116}, \eqref{eq:112}, and \eqref{eq:Cauchy} to
conclude that
\begin{displaymath}
    \Big|\Big\langle2\omega
    \hat{u}_s,\frac{\hat{u}_s^\prime}{\omega^2-\lambda}\Big\rangle\Big|\leq
    \left(\lambda_r^{-1/2}\|(\omega^2-\lambda_r)\hat{u}_s\|_2+  2 \lambda_r^{1/2}\|\hat{u}_s\|_2
    \right)\, \Big(\|\hat{u}_s\|_2 +   
  \frac{C}{ \lambda_i}\|{\bf f}\|_2\Big) \,,
\end{displaymath}
from which we conclude that
\begin{multline*}
  \Big|\Big\langle2\omega
  \hat{u}_s,\frac{\hat{u}_s^\prime}{\omega^2-\lambda}\Big\rangle\Big|\leq
  \lambda_r^{-3/4}\|(\omega^2-\lambda_r)\hat{u}_s\|_2^2\\
  + (2\lambda_r^{1/2}+ \lambda_r^{-1/4})\|\hat{u}_s\|_2^2 +
  \frac{C}{ \lambda_i}\|{\bf f}\|_2\left( 2 \lambda_r^{1/2}\|\hat{u}_s\|_2 + \lambda_r^{-1/2}\|(\omega^2-\lambda_r)\hat{u}_s\|_2\right)\,.   
\end{multline*}
Substituting the above into \eqref{eq:22} then yields
\begin{multline}
\label{eq:23}
  \|\hat{u}_s^\prime\|_2^2 +(\varepsilon^{-1}-
  \lambda_i^2- 2\lambda_r^{1/2}-\lambda_r^{-1/4})\|\hat{u}_s\|_2^2  +
  (1-\lambda_r^{-3/4})\|(\omega^2-\lambda_r)\hat{u}_s\|_2^2 \\
 \leq  |\langle(\omega^2-\bar{\lambda})\hat{u}_s,g\rangle|+   \frac{C}{ \lambda_i}\|{\bf
    f}\|_2\left(2 \lambda_r^{1/2}\|\hat{u}_s\|_2
  + \lambda_r^{-1/2}\|(\omega^2-\lambda_r)\hat{u}_s\|_2\right)\,.    
\end{multline}
 From \eqref{eq:23}, observing that $\epsilon^{-1} -\lambda_i^2\geq 0$,  we can conclude   that
\begin{displaymath}
  \|(\omega^2-\lambda_r)\hat{u}_s\|_2 \leq C\Big(|\langle(\omega^2-\bar{\lambda})\hat{u}_s,g\rangle|^\frac 12 + \lambda_r^{1/4}\|\hat{u}_s\|_2 +
  \frac{\lambda_r^{-1/2}}{ \lambda_i}\|{\bf
    f}\|_2+\frac{\lambda_r^{1/4}}{ \lambda_i^{1/2}}\|{\bf f}\|_2^{1/2}\|\hat{u}_s\|_2^{1/2} \Big)\,.
\end{displaymath}

Substituting the above into \eqref{eq:23} yields
\begin{multline}
\label{eq:106}
  \|\hat{u}_s^\prime\|_2^2 +(\varepsilon^{-1}-
  \lambda_i^2- 2\lambda_r^{1/2}- \lambda_r^{-1/4})\|\hat{u}_s\|_2^2  +
  (1-\lambda_r^{-3/4})\|(\omega^2-\lambda_r)\hat{u}_s\|_2^2 \leq\\
\leq   C|\langle(\omega^2-\bar{\lambda})\hat{u}_s,g\rangle|+   \frac{C}{ \lambda_i}\|{\bf
    f}\|_2\Big(2 \lambda_r^{1/2}\|\hat{u}_s\|_2+ \frac{\lambda_r^{-1}}{\lambda_i}\|{\bf
    f}\|_2  +\frac{\lambda_r^{-1/4}}{ \lambda_i^{1/2}}\|{\bf f}\|_2^{1/2}\|\hat{u}_s\|_2^{1/2}\Big)\,.    
\end{multline}
We recall from \cite[Proposition 3.1]{aletal10} that the
ground state energy of the anharmonic  oscillator 
\begin{displaymath}
  -\frac{d^2}{dx^2} + (\frac 12 x^2-\beta)^2
\end{displaymath}
acting on $\R$, behaves as $\beta \to + \infty$ as $\sqrt{2 \beta}$. Hence, we
can conclude, after dilation, that for sufficiently small $\varepsilon$,
\begin{displaymath}
   \|\hat{u}_s^\prime\|_2^2 +  (1-\lambda_r^{-3/4})\|(\omega^2-\lambda_r)\hat{u}_s\|_2^2 \geq
   2[1-\lambda_r^{-3/4}]^{1/2}\lambda_r^{1/2} (1- C \lambda_r^{-1}) \|\hat{u}_s\|_2^2
   \geq2(\lambda_r^{1/2}-\lambda_r^{-1/4}) \|\hat{u}_s\|_2^2\,.
\end{displaymath}
Substituting the above into \eqref{eq:106} yields
\begin{displaymath}
  (\varepsilon^{-1}- \lambda_i^2 -3\lambda_r^{-1/4})\|\hat{u}_s\|_2^2 \leq C |\langle(\omega^2-\bar{\lambda})\hat{u}_s,g\rangle|+\frac{C}{ \lambda_i}\|{\bf
    f}\|_2\Big(\lambda_r^{1/2}\|\hat{u}_s\|_2+ \frac{\lambda_r^{-1}}{\lambda_i}\|{\bf
    f}\|_2  \Big)\,.    
\end{displaymath}
from which we conclude (recall that $ |\lambda_i|<(1-2\delta^4)^{1/2}\varepsilon^{-1/2}$) that for sufficiently small $\varepsilon$
\begin{equation}
\label{eq:108}
  \|\hat{u}_s\|_2^2 \leq C\varepsilon  \Big((
  |\langle(\omega^2-\bar{\lambda})\hat{u}_s,g\rangle|+ \frac{\varepsilon\lambda_r+\lambda_r^{-1}}{ \lambda_i^2}\|{\bf
  f}\|_2^2   \Big) \,.
\end{equation}
We note that \eqref{eq:24} reads in this case
\begin{multline}
\label{eq:24case2}
  |\langle(\omega^2-\bar{\lambda})\hat{u}_s,g\rangle|  \leq C\Big(\Big[1+\frac{\varepsilon^{1/2}\lambda_{r}^{1/2}}{\lambda_i}   \Big]\|{\bf
    f}\|_2+  \big( \varepsilon^{-1/2}+ \varepsilon^{1/2}\lambda_{r}^{1/2} \big)\|\hat{u}_s\|_2\Big)\|{\bf f}\|_2
\end{multline}
We now use \eqref{eq:108} and \eqref{eq:24case2}  to obtain that
\begin{displaymath}
   \|\hat{u}_s\|_2^2 \leq C\varepsilon  \Big(
  \big(\varepsilon^{-1/2}+\varepsilon^{1/2}\lambda_r^{1/2}\big)
  \|\hat{u}_s\|_2 \|{\bf f}\|_2   + \Big[1+\frac{\varepsilon^{1/2}\lambda_r^{1/2}}{\lambda_i}+
  \frac{\varepsilon\lambda_r+\lambda_r^{-1}}{ \lambda_i^2}\Big]\|{\bf f}\|_2^2   \Big)    \,.
\end{displaymath}
Consequently,
\begin{equation}
\label{eq:113}
  \|\hat{u}_s\|_2 \leq C\varepsilon^{1/2}\Big[1+\frac{\varepsilon^{1/2}\lambda_r^{1/2}}{ \lambda_i}\Big]\|{\bf
  f}\|_2 \,.
\end{equation}

\paragraph{Case 3: $\lambda_r<0$. }\strut \\
Here we write, as in \eqref{eq:20} (with $\eta_\epsilon=1$),
\begin{multline}
\label{eq:97}
  \Big\|\frac{\hat{u}_s^\prime}{\omega^2-\lambda}[\omega^2-\lambda_r]^{1/2}\Big\|_2^2 
 + \varepsilon^{-1} \Big\|\frac{ \hat{u}_s}{\omega^2-\lambda}[\omega^2-\lambda_r]^{1/2}\Big\|_2^2 +
  \|\omega\hat{u}_s\|_2^2 -\lambda_r \|\hat{u}_s\|_2^2= \Re\langle\hat{u}_s,g\rangle \,,
\end{multline}
from which we conclude that
\begin{displaymath}
   \|\omega\hat{u}_s\|_2^2 \leq |\langle\hat{u}_s,g\rangle |\,.
\end{displaymath}
From the above and \eqref{eq:112}
 we can conclude that
\begin{displaymath}
  \Big|\Big\langle2\omega \hat{u}_s,\frac{(\hat{u}_s)^\prime}{\omega^2-\lambda}\Big\rangle\Big|\leq
   |\langle\hat{u}_s,g\rangle |^{1/2}\Big(\|\hat{u}_s\|_2 +
  \frac{2}{ \lambda_i}\|\bf f\|_2\Big) \,. 
\end{displaymath}
Substituting the above into \eqref{eq:22} 
yields
\begin{multline}\label{eqadd}
   \|(\hat{u}_s)^\prime\|_2^2 +[\varepsilon^{-1}- \lambda_i^2)]\|\hat{u}_s\|_2^2
   + \|(\omega^2-\lambda_r)\hat{u}_s\|_2^2 \\ \leq|\langle\hat{u}_s,g\rangle |^{1/2}\Big(\|\hat{u}_s\|_2 +
  \frac{2}{ \lambda_i}\|{\bf f}\|_2\Big) 
  +|\langle(\omega^2-\bar{\lambda})\hat{u}_s,g\rangle| \,.
\end{multline}
Here \eqref{eq:24} reads (note that $\lambda_{r,m}=1$ in
the present case)
\begin{equation}
\label{eq:24Case3}
  |\langle(\omega^2-\bar{\lambda})\hat{u}_s,g\rangle|  \leq C\Big(\Big[1+\frac{\epsilon^\frac 12 +1}{\lambda_i}   \Big]\|{\bf
    f}\|_2+  \varepsilon^{-1/2}  \|\hat{u}_s\|_2\Big)\|{\bf f}\|_2
\end{equation}
We now use \eqref{eq:15}, 
\eqref{eqadd},  and \eqref{eq:24Case3}  to obtain that
\begin{displaymath}
  \|\hat{u}_s\|_2^2\leq C_\delta\varepsilon\Big(\Big[\|\hat{u}_s\|_2+ \frac{1}{ \lambda_i}\|{\bf
    f}\|_2\Big]^{3/2} \|{\bf   f}\|_2^{1/2}+\Big[1+\frac{\varepsilon^{1/2} +1}{\lambda_i}\Big]\|{\bf
    f}\|_2^2+ \varepsilon^{-1/2}\|\hat{u}_s\|_2\|{\bf
    f}\|_2 \Big)\,. 
\end{displaymath}
Hence,
\begin{displaymath}
  \|\hat{u}_s\|_2 \leq C_\delta\varepsilon^{1/2}\Big(1+\frac{1}{ \lambda_i^{3/4}} +\frac{1}{\lambda_i} \Big)\|{\bf f}\|_2\,.
\end{displaymath} 
This completes the proof of the proposition.\end{proof} 
\begin{remark} \label{Rem3.3} Notice that in this third case, since $\B_\epsilon$ is accretive, we have also 
  \begin{displaymath}
    \| {\bf u}\|_2 \leq \frac{C}{|\lambda_r|}\|{\bf f}\|_2
  \end{displaymath}
 which is correct without any limitation on the value of $\lambda_i$ as in
 the statement of Proposition  \ref{prop:strip-estimate}.
  \end{remark}
  
As a consequence of Proposition \ref{prop:strip-estimate}, we get under the same assumptions:
\begin{proposition}
\label{cor:spect}
  Let 
  \begin{displaymath}
    S(\varepsilon,\delta)=\{\lambda\in \C\setminus\overline{\R_+} \,| \, |\Im\lambda|<(1-\delta)\varepsilon^{-1/2}\,\}.
  \end{displaymath}
  Then, for any $\delta>0$ there exists $\varepsilon_0$ such that for all
  $0<\varepsilon<\varepsilon_0$ we have
\begin{displaymath} 
\sigma(\check{\B}_\varepsilon)\cap S(\varepsilon,\delta)=\emptyset \,.  
\end{displaymath}
\end{proposition}
\begin{proof}
  Proposition \ref{prop:strip-estimate} establishes boundedness of $u_s$ only. To
  prove boundedness of ${\bf u}$ one needs, therefore to prove, in
  addition, the boundedness of $\hat{u}_d$ and $\hat{u}_3$ for all $\lambda\in
  S(\varepsilon,\delta)$. To this end we first observe that by (\ref{eq:5}c) it
  holds that
\begin{equation}
 \label{eq:114}
\|\hat{u}_3\|_2 \leq C |\lambda_i|^{-1}\Big(1+\frac{[1+\varepsilon^{1/2}(\lambda_r)_+^{1/2}]}{| \lambda_i|}\Big)\|{\bf f}\|_2 \,.
  \end{equation}
To prove boundedness  of $\hat{u}_d$ we first observe that  by (\ref{eq:5}a)-(\ref{eq:5}b) 
\begin{displaymath}
(\omega^2-\lambda) \hat u_d -\hat u'_s = \hat f_1-\hat f_2\,,
\end{displaymath}
from which we conclude that
 \begin{displaymath}
    \|\hat u_d\| \leq  \frac{1}{|\lambda_i|} (\|\hat u'_s\| + \|{\bf f}\|)\,.
 \end{displaymath}
We now observe that by \eqref{eq:39}, \eqref{eq:106}, and
\eqref{eqadd} we have that
 \begin{displaymath}
   \|\hat u'_s\|_2  \leq C \Big(1+\frac{[1+\varepsilon^{1/2}(\lambda_r)_+^{1/2}]}{| \lambda_i|}\Big)\|{\bf f}\|_2 \,.
 \end{displaymath}
Combining the above yields
  \begin{equation}
\label{eq:21}
  \|{\bf u} \| \leq C \Big(1 +\frac{1}{|\lambda_i|} \Big)
  \Big(1+\frac{[1+\varepsilon^{1/2}(\lambda_r)_+^{1/2}]}{| \lambda_i|}\Big)\|{\bf
    f}\|_2\,,
  \end{equation}
 or, equivalently, that for any $\lambda\in\C$ satisfying $ 0 < | \lambda_i|\leq
  (1-2\delta^4)^{1/2}\varepsilon^{-1/2}$ it holds that
   \begin{equation}
\label{eq:21a}
  \| (\check B_\varepsilon -\lambda)^{-1}\| \leq C \Big(1 +\frac{1}{|\lambda_i|} \Big)
  \Big(1+\frac{[1+\varepsilon^{1/2}(\lambda_r)_+^{1/2}]}{| \lambda_i|}\Big)\,. 
  \end{equation}
 \end{proof}
  We can deduce from \eqref{eq:21a} the following bound on $ \| \B_\epsilon
-\Lambda)^{-1}\|$.
 \begin{corollary}
For any $\Lambda\in\C$ satisfying $ 0 < | \Lambda_i |\leq (1-2\delta^4)^{1/2} $ it
holds that
    \begin{equation}
\label{eq:21b}
  \| \B_\epsilon -\Lambda)^{-1}\| \leq C \epsilon^{-\frac 23}  \Big(1 + \epsilon^{\frac 23} \frac{1}{|\Lambda_i|} \Big)
  \Big(1+\epsilon^{ \frac 23}\frac{[1+\epsilon^{1/3}(\Lambda_r)_+^{1/2}]}{| \Lambda_i|}\Big)\,. 
  \end{equation}
 \end{corollary}

\subsection{Point spectrum}\label{sss331}
We now establish the existence of a point spectrum for $\check{\B}_\varepsilon$
for $|\lambda_i| >\varepsilon^{-1}/2$.  As before, we can assume $ \lambda_i>\varepsilon^{-1}/2$.
Our main result, which is included in the statement of Theorem
\ref{thm:unbounded}, is:
\begin{proposition}
  \label{prop:point-exist}
For every $k\in\N$ there exist  $\varepsilon_k>0$ and $R_k>0$, such that
$B(\lambda_k,R_k\varepsilon^{\frac 3 4})\cap \sigma(\check{\B}_\varepsilon)\neq\emptyset$ for all $0<\varepsilon<\varepsilon_k$,
where $\lambda_k = i\varepsilon^{-\frac 12} + \frac{2k-1}{2} (1+i) \varepsilon^\frac 14$. 
\end{proposition}
\subsubsection{Preliminary reduction}
We begin with the following  substitution
\begin{equation}
\label{eq:53}
  \hat{u}_s=(\omega^2-\lambda)^{1/2}v\,,
\end{equation}
in  \eqref{eq:8} (for ${\bf f}=0$) to obtain $\Mg_\lambda v =0$ with 
\begin{equation}
  \label{eq:38}
\Mg_\lambda  \overset{def}{=}-\frac{d^2}{d\omega^2} +\Big[(\omega^2-\lambda)^2+\varepsilon^{-1}+\frac{2\omega^2+\lambda}{(\omega^2-\lambda)^2}\Big]
\end{equation}
Note that the last term $\frac{2\omega^2+\lambda}{(\omega^2-\lambda)^2}$ is, as
$\lambda_i>\varepsilon^{-1/2}$, $C^\infty$ and bounded. Hence $\mathcal M_\lambda$ is a bounded
perturbation of the anharmonic oscillator:
\begin{equation}\label{eq:defmg0}
\Mg_\lambda^0= -\frac{d^2}{d\omega^2} + (\omega^2 -\lambda)^2 +\epsilon^{-1}\,.
\end{equation}
Note that  $\Mg_\lambda^0$, or the anharmonic
oscillator, has been intensively studied \cite{Sim} (in the form
$-d^2/d\omega^2 + \omega^2 +\beta \omega^4$) and later, in the above form,  
in   \cite{mo95,ha78} for real values of $\lambda$. \\
\\
 It has been established  (see \cite{HeRo} and
references therein) that $\Mg_\lambda$ is for all $\lambda\not\in\R_+$,  a closed operator 
whose domain satisfies 
\begin{equation}
\label{eq:41}
  D(\Mg_\lambda)= \{ u\in L^2(\R)\,|\, \Mg_\lambda u \in L^2(\R)\} =\{ u\in H^2(\R)\,|\,\omega^4u\in L^2(\R)\}\,, 
\end{equation}
and is maximally accretive.\\
We now observe the following:
\begin{lemma}
For any $\lambda\in\C\setminus\R_+$
  \begin{equation}
\label{eq:123}
    0\in\sigma(\Mg_\lambda)\Leftrightarrow\epsilon^{2/3}\lambda\in\sigma(\B_\epsilon) \,.
  \end{equation}
\end{lemma}
\begin{proof}
  We first prove that $0\in
  \sigma_p(\Mg_\lambda)\Leftrightarrow\lambda\in\sigma_p(\check{\B}_\varepsilon)$ (the point spectrum).  Suppose that
  ${\bf u}\in D(\check{\B_\varepsilon})$ is an eigenfunction associated with
  $\lambda\in\sigma_p(\check{\B}_\varepsilon)$. Then, let 
  \begin{displaymath}
    v(\omega)=\frac{1}{(\omega^2-\lambda)^{1/2}}\FF(u_1+u_2)(\omega) \,.
  \end{displaymath}
Clearly, $v\in L^2(\R)$ and, by the construction of $\Mg_\lambda$ we have that
$\Mg_\lambda v=0$. Consequently, $v\in D(\Mg_\lambda)$ and hence
$0\in\sigma_p(\Mg_\lambda)$. (Note that for any $\lambda\in\C\setminus\R_+$ the spectrum of
$\Mg_\lambda$ is discrete.)

Conversely, suppose that $0\in \sigma_p(\Mg_\lambda)$ and let $v\in D(\Mg_\lambda)$ denote an
associated eigenfunction. Set
\begin{equation}\label{eq:vu}
  \hat{u}_s=(\omega^2-\lambda)^{1/2}v\,.
\end{equation}
As  $v\in D(\Mg_\lambda)$ it follows from \eqref{eq:41}  that
$\hat{u_s}\in L^2(\R)$.  We now define, as in (\ref{eq:5}c) and
\eqref{eq:6}
 \begin{equation}
\label{eq:revfor}
 \hat u_d =\frac{1}{\omega^2 -\lambda} \frac{d \hat u_s}{d \omega }\,,\, \hat u_3 =
 \frac{1}{\sqrt 2} \, \varepsilon^{-\frac 12} \frac{\hat u_s}{\omega^2-\lambda}\,. 
 \end{equation}
From \eqref{eq:12} (with $g=0$) we can conclude that $ \hat u_d
\in L^2(\R)$. As $\hat u_3\in L^2(\R)$ we may conclude that
\begin{displaymath}
  {\bf u}=
  \begin{bmatrix}
    \frac{1}{2}\FF^{-1}(\hat{u}_s+\hat{u}_d) \\
    \frac{1}{2}\FF^{-1}(\hat{u}_s-\hat{u}_d) \\
    \FF^{-1}(\hat{u}_3)
  \end{bmatrix}
\in L^2(\R)\,.
\end{displaymath}
Consequently, since $(\check{\B}_\varepsilon-\lambda) {\bf u}=0$ we obtain that $
{\bf u}\in D(\check{\B}_\varepsilon)$ and hence $\lambda\in\sigma(\check{\B}_\varepsilon)$. 

Since $\epsilon^{2/3}\check{\B}_\varepsilon$ is obtained from $\B_\epsilon$ via unitary
dilation and rotation we obtain that $\epsilon^{2/3}\lambda\in\sigma(\B_\epsilon )$. The
proof of \eqref{eq:123} is now completed by using Proposition
\ref{prop.pointsp} together with the fact that $(-\epsilon^2\Delta+ix-\lambda)^{-1}$
is compact.  
\end{proof}

\subsubsection{Heuristics and dilation}\label{sss332}
We begin by making some intuitive observations on $\mathcal M_\lambda^0$.
To this end we return to our initial spectral parameter $\Lambda$
(see \eqref{eq:49}):
\begin{equation}
\label{eq:47}
\lambda =  \varepsilon^{-\frac 12} \Lambda \,,
\end{equation}
which, when substituted into \eqref{eq:defmg0}, yields
\begin{displaymath}
\widetilde \Mg_{\Lambda}^0:= - \frac{d^2}{d\omega^2} + \varepsilon^{-1} (1+\Lambda^2)  - 2 \Lambda \varepsilon^{-\frac 12} \omega ^2 + \omega^4 \,.
\end{displaymath}
Here and in the following we use the term ``critical value'' for any
$\Lambda\in\C$ for which ${\rm ker}\,\Mg_{\Lambda}^0\neq\{0\}$.
The above form suggests it would be plausible to look for critical values
near $\Lambda=i$ (so that $\varepsilon^{-1}(1+\Lambda^2)$ is of the same order of
$2 \Lambda \varepsilon^{-\frac 12} \omega ^2 $). Hence, we set
\begin{subequations}
\label{eq:46}
  \begin{equation}
\Lambda = i + \varepsilon^{\frac{3}{4}} \mu
\end{equation}
 and 
 \begin{equation}
 \omega = \varepsilon^\frac
18 \tilde \omega 
\end{equation}
to obtain, after multiplication by $\varepsilon^\frac 14$,
\begin{equation}
\widehat \Mg_{\mu}^0:= - \frac{d^2}{d\tilde \omega^2}   - 2 i \tilde  \omega ^2 +  2i  \mu  +  \varepsilon^{\frac{3}{4}}
\mu^2  -  2 \mu\varepsilon^{3/4}\tilde \omega^2 +\varepsilon^{\frac 34} \tilde \omega^4 \,.
\end{equation} 
\end{subequations}
Neglecting small terms (in the limit $\varepsilon\to0$) we thus expect $- 2i \mu$
to be an eigenvalue of the complex harmonic oscillator $-d^2/d\tilde
\omega^2-2 i \tilde \omega^2$ (cf. \cite{hebook13}).  We now apply the
previous rescaling \eqref{eq:47} and \eqref{eq:46} to $\mathcal M_\lambda$
to obtain, (for convenience we return to the parameter $\epsilon =\varepsilon^{\frac 34}$)
\begin{equation} \label{rescprob}
\widehat \Mg_{\mu,\epsilon} := - \frac{d^2}{d\tilde \omega^2}  - 2 i \tilde  \omega ^2  +  2i  \mu  + \epsilon 
\mu^2  - 2 \mu\epsilon \tilde \omega^2 + \epsilon \tilde \omega^4+ \epsilon \Phi(\epsilon,\tilde \omega,\mu) \,,
\end{equation}
where
\begin{equation}\label{rescproba}
\Phi(\epsilon,\tilde \omega,\mu) :=\frac{(i+ 2 \epsilon  \tilde \omega^2 + 
  \epsilon \mu)}{(- i + \epsilon \tilde \omega^2 - \epsilon \mu)^{2}}  \,.
\end{equation}
Clearly,
\begin{lemma}
For any $\Lambda \in\C\setminus\R_+$ such that $\Lambda = i+\epsilon\mu $ 
  \begin{equation}
\label{eq:123bis}
    0\in\sigma(\widehat \Mg_{\mu,\epsilon})\Leftrightarrow\Lambda\in\sigma(\B_\epsilon) \,.
  \end{equation}
\end{lemma}
Hence it remains to find critical values $\mu\in \mathbb C$ such that $\widehat
\Mg_{\mu,\epsilon}$ has a non trivial kernel.
 
\subsubsection{Formal asymptotics}
By \eqref{rescprob} $\widehat
{\mathcal M}_{\mu,\epsilon}$ is close to the complex harmonic
oscillator. In the following we present a formal expansion relying on
that intuition. 
 \begin{proposition}
 For any $k\in \mathbb N^*$, there exists sequences
 $\{\mu_{k,\ell}\}_{l=1}^\infty\subset\C$ and $\{f_{k,\ell}\}_{l=1}^\infty \subset\mathcal
 S(\mathbb R)$ such that, as $\epsilon\to0$,
 \begin{subequations}
\label{defmuk}
    \begin{equation}
  \mu_k \sim \sum_{\ell >0} \epsilon^\ell  \mu_{k,\ell}\,,
  \end{equation}
  \begin{equation}
   f_k = \sum_{\ell \geq 0} \epsilon^\ell  f_{k,\ell}\,,
   \end{equation}
and
  \begin{equation}
\label{eq:as}
   \widehat \Mg_{\mu_k,\epsilon} f_k \sim 0\,.
  \end{equation}
 \end{subequations}
   In particular,  we have
\begin{displaymath}
  \widehat  \Mg_{\mu_{k,\epsilon}} f_{k,0} = \epsilon \,  r_k^{(0)} (\tilde \omega,\epsilon)\,,
\end{displaymath}
where $\omega \mapsto \exp (e^{i \frac \pi 4}
\,\omega^2/\sqrt{2}) \, r_k^{(0)} (\omega,\epsilon)$ is polynomially bounded. \\
For $\epsilon=0$, we have $$\widehat \Mg_{\mu_{k,0}}f_{k,0}=0\,.$$
     \end{proposition}
  
     \begin{proof}
For the leading order we have by (\ref{defmuk}c)
 \begin{equation}
\label{eq:as0}
 \widehat  \Mg_{\mu_k,0} f_{k,0} =\Big(- \frac{d^2}{d\omega^2} - 2 i\omega^2 + 2 i \mu_{k,0}\Big) f_{k,0} =0\,.
 \end{equation}
The eigenvalues of the complex harmonic oscillator are well-known
(cf. for instance \cite[Proposition 14.12]{hebook13}). Hence,
\begin{equation}\label{defmuk0}
 \mu_{k,0} =\frac{ (2k-1)}{\sqrt{2}} e^{i\frac \pi 4}\,.
 \end{equation} 
  We can normalize the corresponding eigenmode by setting
\begin{displaymath}
\int |f_{k,0}|^2(\tilde \omega)  d\tilde \omega =1\,.
\end{displaymath}
 The coefficient of $\epsilon$  assumes the form
\begin{equation}\label{eq:as2}
 \Mg_{\mu_k,0} f_{k,1} = -2 i \mu_{k,1} f_{k,0} -
 (\mu_{k,0}^2-2\mu_{k,0}\omega^2+\omega^4-1)f_{k,0}\,. 
\end{equation}
A necessary and sufficient condition for the existence of solution for
\eqref{eq:as2} is obtained by taking the inner product with $\bar
f_{k,0}$, yielding
 \begin{equation}\label{defmuk2}
 \mu_{k,1} =
 \frac{\int_\R([\omega^2-\mu_{k,0}]^2-1)f_{k,0}^2\,d\omega}{2i\int_\R f_{k,0}^2\,d\omega}
 \end{equation}
 Under \eqref{defmuk2}, there exists a unique solution  $f_{k,1}$ of
 \eqref{eq:as2} if we add the condition
 \begin{displaymath}
   \int_\R f_{k,0} (\omega) f_{k,1} (\omega) d\omega =0 \,.
 \end{displaymath}
 We can then continue by recursion, to prove \eqref{defmuk0} for any
 order. To prove that $\exp (e^{i \frac \pi 4}
 \,\omega^2/\sqrt{2}) \, r_k^{(0)} (\omega,\epsilon)$ is polynomially bounded we use the
 well-known properties of Hermite functions to conclude that
\begin{displaymath}
   \Big|\exp \Big(\frac{1}{\sqrt{2}} \exp\Big(i \frac \pi 4\Big)
\,\omega^2\Big)f_{k,0}\Big|\leq C|\omega|^k \,.
\end{displaymath}
Then, by direct substitution we obtain that
\begin{displaymath}
     \Big|\exp \Big(\frac{1}{\sqrt{2}} \exp\Big(i \frac \pi 4\Big)
\,\omega^2\Big)r_k^0\Big|\leq C|\omega|^{k+4}\,. 
\end{displaymath}
  \end{proof}

  We have formally established that, for sufficiently small
  $\epsilon$, $\sigma(\B_\epsilon)$ should contain a sequence of points $\Lambda_k\sim i  +\epsilon  \mu_k$.\\
  In the following, we attempt to rigorously prove these formal
  estimates. Two of the difficulties we face in the forthcoming
  rigorous analysis are:
  \begin{enumerate}
  \item It involves a non-linear spectral problem.
  \item It involves the spectral analysis of a non- selfadjoint operator.
  \end{enumerate}
  Non-self-adjointness prohibits the use of the spectral theorem
  to estimate error terms. To mitigate this problem we use analytic
  dilation so that the leading order operator, converts from the
  complex harmonic oscillator into the real harmonic oscillator.

 \subsubsection{Analytic dilation.}
      We begin by recalling from \eqref{eq:38}
   \begin{equation*}
\Mg_\lambda  \overset{def}{=}-\frac{d^2}{d\omega^2} +\Big[(\omega^2-\lambda)^2+\varepsilon^{-1}+\frac{2\omega^2+\lambda}{(\omega^2-\lambda)^2}\Big]
\end{equation*}
  Let $\theta\in\mathbb R$. We introduce the unitary dilation operator
\begin{equation}
\label{eq:40}
u\longmapsto (U(\theta) u)(x)= e^{-\theta/2} \, u( e^{-\theta} x)\,.
\end{equation}
We then define
\begin{equation}
\label{eq:119}
\Mg_{\lambda,\theta}  := U(\theta)^{-1}\Mg_\lambda U(\theta) =-e^{-2 \theta}
\frac{d^2}{d\omega^2}+V_{\lambda,\theta} \,,
\end{equation}
where
\begin{displaymath}
V_{\lambda,\theta} =  (e^{2\theta}\omega^2-\lambda)^2+\varepsilon^{-1}
+\frac{2e^{2\theta}\omega^2+\lambda}{(e^{2\theta}\omega^2-\lambda)^2}\,.
\end{displaymath}
Similarly, we define 
\begin{equation}
\label{eq:119a}
\Mg_{\lambda,\theta}^0  := -e^{-2 \theta}
\frac{d^2}{d\omega^2}+V^0_{\lambda,\theta} \,,
\end{equation}
where
\begin{displaymath}
V^0_{\lambda,\theta} =  (e^{2\theta}\omega^2-\lambda)^2+\varepsilon^{-1}
\end{displaymath}
$\Mg^0_{\lambda,\theta}$ can be extended to the strip $|\Im \theta| < \frac \pi 6$,
using \eqref{eq:119a}, as a globally quasi-elliptic operator (see
\cite{HeRo0},\cite{HeSj}) whose principal term (in the sense of this theory) is $-e^{-2 \theta}
\frac{d^2}{d\omega^2} +  e^{4\theta}\omega^4$.\\

It is then a rather standard matter to show that its spectrum is
independent of $\theta$.  Hence it remains necessary to control the effect
of the "perturbation term"
\begin{displaymath}
 \phi (\omega,\lambda,\theta)= \frac{2e^{2\theta}\omega^2+\lambda}{(e^{2\theta}\omega^2-\lambda)^2}\,.
\end{displaymath}
We observe that if $|\Im \theta| < \frac \pi 8 + \frac {\epsilon_0}2 $ and $\pi/4
+\epsilon_0<\arg \lambda<7\pi/4 -\epsilon_0$ (for some $\epsilon_0 >0$) and $|\lambda| >0$, then
$(\omega,\theta) \mapsto \phi(\omega,\lambda,\theta)$ remains bounded. Consequently, $\Mg_{\lambda,\theta}$
is sectorial and possesses a compact resolvent. We now use standard
arguments: We observe that for $\theta\in\R$, since $U(\theta)$ is a unitary
operator, $\sigma(\Mg_{\lambda,\theta})$ is independent of $\theta$. By analytic
continuation it must also be constant for $-\frac \pi 8 -\frac{\epsilon_0}{2}
<\Im \theta<\frac \pi 8+ \frac{\epsilon_0}{2} $ (see \cite[Section VI.1.3]{ka80},
\cite[Section 12]{HeSj} or \cite{AHP2}). Hence we have obtained,
for  $\epsilon_0=\frac \pi 8$,
  \begin{proposition}
Let $\lambda\in\C\setminus \{ 0 \}$ satisfy $3\pi/8 <\arg \lambda<13 \pi/8$, and let $\theta$
be such that $|\Im \theta| < \frac {3\pi}{ 16}$. Then,
\begin{equation}\label{eq:120a}
\sigma (\Mg_{\lambda,\theta})=\sigma (\Mg_{\lambda})
  \end{equation} 
\end{proposition}
In particular, it holds that
\begin{equation}
    \label{eq:120}
 0\in   \sigma (\Mg_{\lambda,\theta})\Leftrightarrow 0 \in \sigma (\Mg_{\lambda})\,.
 \end{equation}
Setting $\theta=\frac \pi 8$, we obtain
\begin{corollary}
  For all $\lambda\in\C\setminus \{0\}$
  satisfying $3\pi/8<\arg \lambda<13\pi/8$, it holds that 
\begin{equation}
  \label{eq:121}
  0\in\sigma(\Mg_{\lambda,\frac \pi 8})\Leftrightarrow\varepsilon^{1/2}\lambda\in\sigma(\B_\epsilon) \,.
\end{equation}
\end{corollary}

Finally, we apply \eqref{eq:47} and \eqref{eq:46} to $\Mg_{\lambda,\frac \pi
  8}$ to obtain (using the original parameter $\epsilon$) the operator
\begin{equation} 
\label{eq:43}
\widehat \LL_{\mu,\epsilon} := - \frac{d^2}{ds^2}  + 2 s^2  +  2 e^{3i\frac \pi 4}   \mu  + \epsilon
e^{i\frac \pi 4} \mu^2 - 2i \mu\epsilon s ^2 + e^{3i \frac \pi 4} \epsilon s^4 + \epsilon \, e^{i \frac \pi 4} \Phi (\epsilon, e^{i\frac \pi 8} s,\mu)   \,.
\end{equation}
We search for $\mu$ values for which $\widehat \LL_{\mu,\epsilon} $ has a non
trivial kernel.

\subsubsection{Weighted estimates}
\begin{lemma}\label{lemwe1}
  Let $r>0$, $\mu \in B(0,r)$ and $(g,w)$ a pair in $L^2(\mathbb R) \times D
  ( \widehat \LL_{\mu,\varepsilon}) $ such that $ \widehat \LL_{\mu,\epsilon} w= g$.
  Suppose that $\|e^{|\cdot|}g\|_2<\infty$. Then, there exist $\epsilon_0>0$ and
  $C>0$ such that for all $0<\epsilon<\epsilon_0$ it holds that
  \begin{equation}
  \label{eq:44}
  \|e^{|\cdot|}w\|_2\leq C(\|w\|_2+\|e^{|\cdot|}g\|_2)\,.
  \end{equation}
\end{lemma}
\begin{proof}
 Let
$\varphi (s) =\sqrt{1 + s^2}$ and   $\varphi_n (s)
  =\tilde \eta (s/n)\sqrt{1 + s^2}$ (where $\tilde \eta$ is defined in
  \eqref{eq:11}). Integration by parts yields that for any
$\alpha>0$, there exists $C_\alpha(r)>0$ such that, for any $n\geq 1$
\begin{displaymath}
  \begin{array}{l}
\Re \left (  e^{-3i\pi/8} \langle e^{2 \varphi_n(s)} \widehat \LL_{\mu,\varepsilon} w,w\rangle\right) \\
\quad \geq \cos (\frac {3\pi}{8} ) [ \| (e^{\varphi_n} w)' \|^2 + 2 \| s e^{\varphi_n} w\|^2 +  \epsilon \|s^2 e^{\varphi_n} w\|^2  ] -\alpha  \| (e^{\varphi_n} w)' \|^2  \\
\qquad  \qquad  - C_\alpha  (  \epsilon \| s e^{\varphi_n}w\|^2 + \|  e^{\varphi_n} w\|^2)\\
\quad \geq \cos (\frac {3\pi}{8} ) \| s e^{\varphi_n} w\|^2  - C  \| e^{\varphi_n} w\|^2\,.
\end{array}
\end{displaymath}
To obtain the second inequality we use the uniform boundedness of $\Phi
(\epsilon, e^{i\frac \pi 8} s,\mu) $ (see \eqref{rescproba}).  To obtain the
last inequality we choose first $0<\alpha>\cos(3\pi/8)/2$ and then  $\epsilon_0$
small enough so that $C_\alpha\epsilon<\cos(3\pi/8)$ for all $0<\epsilon<\epsilon_0$. We
can now conclude that there exist $C>0$ and $C_1>0$ such that
\begin{displaymath}
\| e^{\varphi_n}g\| \,\| e^{\varphi_n} w \| \geq C\, \| e^{\varphi_n } w\|^2 - C_1\, \|w\|^2\,.
\end{displaymath}
  Finally, we can take the limit as $n\to +\infty$.
\end{proof}

\begin{lemma} 
  Let $k\in\N$ and $\mu_{k,0}$ be given by \eqref{defmuk}. Let further
  $g\in L^2(\R,e^{|\cdot|})$ and $(\mu, w)\in\partial B(\mu_{k,0},r)\times D ( \widehat
  \LL_{\mu,\epsilon})$ satisfy 
  \begin{equation}
    \label{eq:48}
\widehat \LL_{\mu,\epsilon} w= g\,. 
  \end{equation}
Then, there exist
  positive $\epsilon_0$, $r_0$, $C_0$ and $C$ such that for all $(\epsilon,r)$
  satisfying \break $0<\epsilon\leq \epsilon_0$, $C_0 \epsilon  \leq r \leq r_0$, it holds that
  \begin{equation}
    \label{eq:45}
\|w\|_2\leq \frac{C}{r} (\epsilon \|e^{|\cdot|}g\|_2+\|g\|_2) \,.
  \end{equation}
\end{lemma}
\begin{proof}
Using  \eqref{eq:43}, we rewrite \eqref{eq:48} in the form
\begin{equation} \label{eq:43bis}
\Big(  - \frac{d^2}{ds^2}  + 2 s^2  +  2 e^{3i\frac \pi 4}   \mu \Big) w = g  + \epsilon  R(s,\epsilon,\mu) w   \,,
\end{equation}
where
\begin{equation}\label{eq:43ter}
R(s,\varepsilon,\mu)  := - 2i \mu s ^2 + e^{3i \frac \pi 4} s^4 + e^{i\frac \pi 4}
\mu^2+ e^{i \frac \pi 4} \Phi (\epsilon, e^{i\frac \pi 8} s,\mu) ) 
\end{equation}

Using the spectral theorem for the harmonic oscillator, we immediately
obtain that for $\varepsilon$ and $r$ small enough we have 
\begin{displaymath}
\|w\|_2\leq \frac{C}{r} \left( \|g\|_2 + \epsilon  \|w\| + \epsilon \| e^\varphi w\|\right)\,.
\end{displaymath}
Using Lemma \ref{lemwe1} we then deduce
\begin{displaymath}
\|w\|_2\leq \frac{C}{r} \left( \|g\|_2 + \epsilon \|w\| + \epsilon \| e^\varphi g\|\right)\,.
\end{displaymath}
Assuming that $\epsilon /r$ is small enough (which is obtained via a suitable
choice of $C_0$) we obtain \eqref{eq:45}.
\end{proof}

\subsubsection{Proof of Proposition \ref{prop:point-exist}}
\begin{lemma}
  \label{lem:point-exist}
  For every $k\in\N$ there exist $\epsilon_k>0$ and $R_k>0$, such that for
  all $0<\epsilon<\epsilon_k$ there exists $\mu'_k(\epsilon) \in B(\mu_{k,0},R_k\epsilon) $ for
  which $\widehat{\mathcal L}_{\mu'_k(\epsilon)}$ has a non trivial 
  kernel.
\end{lemma}
\begin{proof}
  Let $w_k$ denote the $k$-th normalized eigenfunction of the harmonic
  oscillator  $- d^2/ds^2  + 2 s^2$, corresponding to the
  eigenvalue $(2k-1)/\sqrt{2}$. By \eqref{eq:43} we have that
\begin{displaymath} 
g_k:=  \widehat {\mathcal L}_{\mu,\epsilon} w_k = [(2k-1) \sqrt{2} +  2 e^{3i\frac \pi 4}   \mu ] w_k  + \epsilon  R(s,\epsilon,\mu) w_k   \,.
\end{displaymath}
We now define
\begin{displaymath}
\mu =\mu_{k,0} + \rho  e^{i\theta}\,, \, \nu(\mu,k) =  [(2k-1) \sqrt{2} +
2e^{3i\frac \pi 4}   \mu] \,,   
\end{displaymath}
where $\mu_{k,0}$ is defined in \eqref{defmuk0}, to obtain
\begin{equation} \label{eq:43a}
  \widehat {\mathcal L}_{\mu,\epsilon} w_k = \nu (\mu,k,\theta) w_k  + \epsilon  R(s,\epsilon,\mu) w_k   \,.
\end{equation}
Note that
\begin{displaymath}
\nu(\mu,k,\theta) =  2 \rho e^{3i\frac \pi 4 +\theta}  +\mathcal O  (\epsilon) \,.
\end{displaymath}
Suppose, for a contradiction,  that $ \widehat {\mathcal
  L}_{\mu,\varepsilon}$ is invertible for all $\mu\in B(\mu_{k,0},r)$. Then in this ball, $\mu
\mapsto   \widehat {\mathcal L}_{\mu}^{\,-1}$ is an holomorphic family of
compact operators acting on $L^2(\mathbb R)$. 
By \eqref{eq:43a}, we may write
\begin{equation} \label{eq:43b}
\frac{1}{\nu(\mu,k,\theta) } w_k = (\widehat {\mathcal L}_{\mu,\epsilon})^{-1} w_k  +
\frac{\epsilon}{\nu (\mu,k,\theta) }  (\widehat {\mathcal L}_{\mu,\epsilon})^{-1}  R(s,\epsilon,\mu)
w_k   \,. 
\end{equation}
We now take the scalar product with $w_k$ and integrate along a circle of radius $r/2$ centered at $\mu_{k,0}$, assuming that 
\begin{displaymath}
r \geq C \epsilon
\end{displaymath}
 for some, sufficiently large, $C>0$, 
\begin{multline}
  \label{eq:52}
\Big|\int_{\partial B(\mu_{k,0} ,r/2)}\langle w_k, \widehat{\mathcal  L}_{\mu,\epsilon}^{-1}w_k \rangle\,d\mu
- \frac 12 e^{- i \frac{3\pi}{4}} \Big|\leq \hat C \frac{\epsilon}{r}
\Big|\int_{\partial B(\mu_{k,0},r)/2} \langle w_k, \widehat{\mathcal  L}_{\mu,\epsilon}^{-1} R w _k\rangle\,d\mu\Big|  \,.
\end{multline}

We now turn to estimate $\|\widehat{\LL}_{\mu,\varepsilon}^{-1}R w_k\|_2$. To this end we use
\eqref{eq:45} to establish that
\begin{equation}
\label{eq:55a}
\|\widehat{\LL}_{\mu,\epsilon}^{-1}R w_k\|_2 \leq \frac{C}{r}(\epsilon \|e^{|\cdot|} R w_k\|_2+\|R w_k\|_2)\,.
\end{equation}
Using the well-known exponential decay of $w_k(\omega)$ as $|\omega|\to\infty$
(an Hermite function \cite{abst72}) and the polynomial
  boundedness of $R$ we obtain
\begin{equation}
\label{eq:55}
\|\widehat{\LL}_{\mu,\epsilon }^{-1}R w_k\|_2 \leq \frac{C_k}{r}\,.
\end{equation}
Finally, by the assumed holomorphy of $\widehat{\LL}_{\mu,\epsilon}^{-1}$,
\begin{equation}
  \label{eq:52f}
\frac 12  = \Big|\int_{\partial B(\mu_k,r/2)}\langle w_k, \widehat{\mathcal  L}_{\mu,\epsilon}^{-1}w_k \rangle\,d\mu
- \frac 12 e^{- i \frac{3\pi}{4}} \Big|\leq \hat C  \frac{\epsilon}{r}
\end{equation}
Letting $r= R_k \epsilon$ leads to a contradiction for
sufficiently large $R_k$ and $0< \epsilon \leq \epsilon_k$.  This completes the proof
of the lemma.
\end{proof}

\begin{remark} 
\label{rem:original-spectrum}
  By \eqref{eq:123bis} and \eqref{eq:121} it follows that for
  $\Lambda=i + \epsilon \mu$ we have that
  \begin{displaymath}
     0\in\sigma(\widehat{\mathcal  L}_{\mu,\epsilon})\Leftrightarrow \Lambda \in\sigma(\B_\epsilon) \,.
  \end{displaymath}
By Lemma \ref{lem:point-exist} we then obtain that for each $k\in\N$
there exists $R_k>0$ and \break $\Lambda_k\in B(i+\mu_{k,0}\epsilon,R_k\epsilon^2)$ which belongs
to $\sigma(\B_\epsilon)$. We can now easily conclude \eqref{eq:105} from
\eqref{eq:49}.   
\end{remark}
\subsection{Reminder on the complex harmonic oscillator}
We recall that the complex harmonic
oscillator $\hg$, constitutes the principal part of $ \Mg_{\mu,\epsilon}$.  We recall that 
  \begin{equation}
    \label{eq:30}
\hg = -\frac{d^2}{d\xi^2} - 2 i\xi ^2 \,,
  \end{equation}
 is  defined on $D(\hg)=H^2(\R)\cap L^2(\R;\xi^2\,d\xi)$, and that its spectrum is
 \begin{equation}
   \label{eq:33}
 \sigma (\hg):=\N_{-\pi/4}=\{\sqrt{2} e^{-i\pi/4}(2k-1)\}_{k=1}^\infty\,.
 \end{equation}
Moreover, its numerical range is $\mathbb C_{+-} =\{z\in \mathbb C, \Re z\geq 0\,,\, \Im z \leq 0\}$.
\begin{proposition}\label{prop3.15}
The following estimates for the resolvent of $\hg$ hold:
\begin{enumerate}
\item If $z\notin \mathbb C_{+-}$, then $z\in \rho(\hg):=\mathbb C \setminus \sigma(\hg)$ and 
  \begin{subequations}
  \label{eq:31}
    \begin{equation}
\| (\hg -z)^{-1} \| \leq d(z, \mathbb C_{+-})^{-1} \,.
\end{equation}
\item For any compact set $K\in \mathbb C$, there exists a constant $C_K$ such that, for any $z\in K\cap \rho(\hg)$ we have
\begin{equation}
\| (\hg -z)^{-1} \| \leq  C_K \Big(\frac{1}{d(z,\sigma(\hg))}+1\Big)\,.
\end{equation}
\item There exist $\delta_0
  >0$ and $B_0 >0$ such that for all  $z\in \rho(\hg)$ such that 
  \begin{equation}
\label{eq:34}
    \Re z \geq 0 \mbox{  and } |\Im z| \leq \delta_0 (\Re z)^3 \,,
  \end{equation}
it holds that
    \begin{equation}
\|(\hg - z )^{-1}\|\leq B_0 \Big(\frac{1}{d(z,\sigma(\hg))}+
  \frac{1}{1+|\Re z |^{1/3}} \Big)\,.
\end{equation}
\end{subequations}
\end{enumerate}
\end{proposition}
\begin{proof}~\\
The first item is a consequence of the sectoriality of $\hg$  (See \cite[Remark 14.14]{hebook13}).\\
 For the second item, by  \cite[Proposition
14.12]{hebook13}, we may use Riesz-Schauder theory to conclude that
for sufficiently large $N$ we have
\begin{displaymath}
  (\hg-z)^{-1}= \sum_{n=1}^N
  \frac{\Pi_n}{ ( \sqrt{2} e^{-i\pi/4}(2n-1)-z)} + T_N(z)\,,
\end{displaymath}
where $\Pi_n$ denotes the projection on the $n$'th eigenfunction of
$\hg$ and $T_N(z)$ is holomorphic in $K$. \\
Finally, we prove the third item. To this end,
let $v\in D(\hg)$ and $g\in L^2(\R)$ satisfy
\begin{displaymath}
  (\hg-z)v=g\,.
\end{displaymath}
  Applying a  Fourier transform to \eqref{eq:30} yields
  \begin{equation}\label{eq:31aa} 
     \Big(- 2 \frac{d^2}{dx^2}  + i(x^2- \Re z )\Big)\hat{v}=
     i\hat{g}+\Im z \, \hat{v} \,.
  \end{equation}
  Using \cite[Proposition 14.13]{hebook13} yields that there exists
  $C>0$ such that for $\Re z  \geq C$ we have
  \begin{displaymath}
       \|v\|_2\leq\frac{C}{\Re z ^\frac 13}(\|g\|_2+ |\Im z  | \|v\|_2)\,.
  \end{displaymath}
By \eqref{eq:34} we can conclude that for sufficiently small $\delta_0$
\begin{equation}
\label{eq:32}
   \|v\|_2\leq\frac{C}{\Re z ^{1/3}}\|g\|_2\,.
\end{equation}
For $\Re z \leq C$ it holds that $|\Im z| \leq C^\frac 13 \delta_0$, and hence
$z$ belongs to a compact set in $\C$ and one can conclude the proof of
proposition  from (\ref{eq:31}b).
\end{proof}

\subsection{Application to resolvent estimates}

We continue the analysis of the spectral properties of $\B_\epsilon$ by
obtaining resolvent estimates for $\B_\epsilon$ in a set in the form
\begin{displaymath}
{\mathcal V} (\epsilon,\varrho):=\{\Lambda \in \mathbb C\,,\, \Lambda_i\geq 1 /2\,,\,  
  \Lambda_r\leq\varrho\epsilon \}\,.
\end{displaymath}
As $\Lambda=i+\epsilon\mu$ it holds that 
\begin{displaymath}
\Lambda_r=\epsilon \mu_r\,, \Lambda_i = 1+\epsilon \mu_i\,,
\end{displaymath}
Equivalently we may write, 
  \begin{displaymath}
    \Omega(\epsilon,\varrho) = \Big\{\mu \in \mathbb C\,,\, \mu_i\geq -\frac{1}{2\epsilon}\,,\, \mu_r\leq\varrho  \Big\}\,,
  \end{displaymath}
  where $\mu_r=\Re\mu$ and $\mu_i=\Im\mu$.\\
We begin by rewriting \eqref{rescprob} in the form
\begin{displaymath}
  \widehat \Mg_{\mu,\epsilon} := - \frac{d^2}{d\tilde \omega^2}  - 2i(1+\epsilon\mu_i)\tilde \omega^2+i[2\mu_r
  +i\mu_i(2+\epsilon\mu_i-2i\epsilon\mu_r)] + \epsilon(\tilde \omega^2-\mu_r)^2 + \epsilon \Phi(\epsilon,\tilde \omega,\mu) \,,
\end{displaymath}
 We further apply the rescaling
\begin{equation}
\label{eq:64}
  \xi =\tilde{\omega}\, [1+\epsilon\mu_i]^{1/4}
\end{equation}
to obtain, after division by $[1+\epsilon \mu_i]^\frac 12$, 
\begin{equation}\label{eq:25}
  \Mg_{\mu,\epsilon}= - \frac{d^2}{d\xi^2}  - 2i\xi^2-z_0+
  \frac{\epsilon}{[1+\epsilon\mu_i]^{1/2}}[( [1+\epsilon\mu_i]^{-1/2}\xi^2 -\mu_r)^2 +
  \Phi(\epsilon, [1+\epsilon\mu_i]^{-1/4}\xi,\mu) \,,
\end{equation}
where
\begin{equation}\label{eq:25a}
  z_0(\mu,\epsilon) : = \frac{\mu_i(2+\epsilon\mu_i)-2i\mu_r(1+\epsilon\mu_i) }{[1+\epsilon\mu_i]^{1/2}}
\end{equation}
Note that for $\mu \in \Omega(\epsilon,\varrho)\cap\{  \mu_i\geq 16\varrho^3\delta_0^{-1}\} \cap \{|\mu_r|\leq \varrho\}$,  we have, \\
\begin{equation}\label{eq:25b}
 |\Im z_0|^3   \leq 8 \varrho^3 \Big(\frac{1}{\mu_i}+\epsilon\Big) \Re z_0 \,.
\end{equation}
We then have 
\begin{lemma}
For any $\varrho >0$, there
  exist $\epsilon_0 >0$ and $B_1 >0$ such that if $\epsilon \in (0,\epsilon_0]$ and $\mu
  \in \Omega (\epsilon,\varrho)$, then
    \begin{equation}
  \label{eq:31a}
\|(\hg -z_0)^{-1}\|\leq B_1 \Big(\frac{1}{d(z_0, \sigma (\hg))}+
  \frac{1}{1+|\Re z_0|^{1/3}} \Big)\,.
\end{equation}
\end{lemma}
\begin{proof}
  The lemma easily follows Proposition \ref{prop3.15}. \\
  For  $\mu_i>16\varrho^3\delta_0^{-1}$ it follows that $\Re z_0>0$ and that there exists
  $\epsilon_0(\delta_0,\varrho)$ such that for $0 < \epsilon \leq \epsilon_0(\varrho,\delta_0)$, we have
\begin{displaymath}
  8 \varrho^3 \Big(\frac{1}{\mu_i}+\epsilon\Big) \leq \delta_0\,.
\end{displaymath}
Consequently, we may conclude \eqref{eq:31a} from (\ref{eq:31}c) in
the case $|\mu_r| \leq \varrho$.  In the case $\mu_r < -\varrho$ we use
(\ref{eq:31}a).  Hence, it remains to treat the case when $ \epsilon \mu_i \in
(-\frac 12,\epsilon16\varrho^3\delta_0^{-1})$. In this case $\Im z_0$ is bounded and $\Re z_0$ has the
sign of $\mu_i$. If $\mu_i\geq - C$, for some $C>0$, then we can apply
(\ref{eq:31}b).  Otherwise when $\mu_i < - C$ we may use (\ref{eq:31}a)
once again to complete the proof.
\end{proof}
\begin{lemma} 
\label{lem:resolv-estim}
Let for   $R>0$ and $\epsilon>0$, 
\begin{equation}
\label{eq:63}
\check \Omega(\epsilon,R)=\{z\in\C \,| \, d(z,\sigma(\hg) )\geq R\epsilon\}\,.
\end{equation}
Then, for all $\varrho >0$, there exist $R_0>0$ and $C >0$ such
that, for all  $R_0<R<\epsilon^{-1}$, 
$\mu \in \Omega(\epsilon,\varrho)$ for which $z_0(\mu,\epsilon)\in \check \Omega
(\epsilon,R)$ and for every pair $(\tilde v,\tilde g) \in D(\Mg_{\mu,\epsilon})\times 
L^2(\R;(\xi^2+1)^2\,d\xi)$ satisfying 
\begin{equation}\label{eq:relvg}
\Mg_{\mu,\epsilon}\tilde v=\tilde g\,,
\end{equation}
we have
\begin{equation}
\label{eq:26}  
  \|\tilde v \|_2\leq \frac{C}{R\epsilon}\,  \left( \|\tilde g\|_2 + \epsilon \|\xi^2 \tilde g\|_2 \right) \,. 
\end{equation}
\end{lemma}
\begin{proof}
The proof is based on presenting $ \Mg_{\mu,\epsilon}$, given by
\eqref{eq:25}, as a perturbation of $\hg$.   We preliminary observe that by global estimates for the
  quartic oscillator  (see for example Helffer-Robert \cite{HeRo0,HeRo}) $\xi^6 \tilde v$, $\xi^4 \tilde v'$, and $\xi^2 \tilde v''$ belong to $L^2$, 
  if $\xi^2 \tilde g \in L^2$. \\
Using the fact that $1+\epsilon
\mu_i=\Lambda_i$ we may write
\begin{equation}\label{eq:99z}
 \Mg_{\mu,\epsilon} =\hg -z_0 + \epsilon \Lambda_i^{-\frac 12} ( \Lambda_i^{-\frac{1}{2}} \xi^2 -\mu_r)^2 +\epsilon \Lambda_i^{-\frac 12} \Phi\,.
\end{equation}
We first attempt to estimate the effect of the perturbation term $\epsilon
\Lambda_i ^{-1/2}\,(\Lambda_i^{-\frac{1}{2}} \xi^2 -\mu_r)^2\,.$ To this end we first
observe that
\begin{equation*}
 \hg - z_0 = -\frac{d^2}{d\xi^2} -\Re z_0  - 2i \Lambda_i^{\frac 12}   (\Lambda_i^{-\frac 12} \xi^2 -\mu_r)\,.
\end{equation*}
Thus, to estimate the perturbation term, we evaluate, having in mind
that \break $(\Lambda_i^{-1/2}\xi^2-\mu_r)^3 \tilde v \in L^2$, the quantity $
-\Lambda_i^{1/2}\Im\langle (\Lambda_i^{-1/2}\xi^2-\mu_r)^3\,\tilde v,\Mg_{\mu,\epsilon} \tilde
v\rangle$ to obtain
  \begin{multline*}
    -\Lambda_i^{1/2}\Im\langle ( \Lambda_i^{-1/2}\xi^2-\mu_r)^3 \tilde v,\tilde g \rangle= 2\Lambda_i\|(\Lambda_i^{-1/2}\xi^2-\mu_r)^2\,\tilde v\|_2^2 \\ - 6\Im\langle
     ( \Lambda_i^{-1/2}\xi^2-\mu_r)^2\,\tilde v, \xi \tilde v^\prime \rangle-\epsilon\Im\Big\langle
    ( \Lambda_i^{-1/2}\xi^2 - \mu_r)^3\,\tilde v,\Phi(\epsilon, \Lambda_i^{-1/4}\xi,\mu) \tilde v\Big\rangle \,,
  \end{multline*}
implying that
    \begin{multline*}
  2\Lambda_i\|(\Lambda_i^{-1/2}\xi^2-\mu_r)^2\,\tilde v\|_2^2 \leq  \Lambda_i^{1/2}\|(\Lambda_i^{-1/2}\xi^2-\mu_r)^2 \tilde v\|_2\,\| (\Lambda_i^{-1/2}\xi^2-\mu_r) \tilde g \|_2  \\ 
  +  6 \|(\Lambda_i^{-1/2}\xi^2-\mu_r)^2 \tilde v\|_2\,\| \xi \tilde v^\prime\|_2 \\ 
+ \epsilon \|(\Lambda_i^{-1/2}\xi^2- \mu_r)^2 \tilde v\|_2 \,\| \Phi(\epsilon,\Lambda_i^{-1/4}\xi,\mu) (\Lambda_i^{-1/2}\xi^2- \mu_r) \tilde v \|_2 \,.
  \end{multline*}
  Next we attempt to bound $ \Phi(\epsilon,\Lambda_i^{1/4}\xi,\mu)$
  under the assumptions of the lemma. Recall from \eqref{rescproba}
  that 
  \begin{equation}\label{rescprobaz}
\Phi(\epsilon,\omega,\mu) :=\frac{(i+ 2 \epsilon \omega^2 + 
  \epsilon \mu)}{(- i + \epsilon \omega^2 - \epsilon \mu)^{2}}  \,.
\end{equation}
For all $\Lambda\in \mathcal V (\epsilon,\varrho)$ it holds that
\begin{displaymath}
\mathbb R_+\ni \tau \mapsto \frac{( 2\tau + \Lambda)}{(\tau  - \Lambda)^{2} }\,,
\end{displaymath} 
is uniformly bounded. Consequently, uniform boundedness of
$\Phi(\epsilon,\omega,\mu)$ follows as well. 
Hence, there exists $C>0$ such that 
\begin{equation}
  \label{eq:27}
\|(\Lambda_i^{-1/2}\xi^2-\mu_r)^2\,\tilde v\|^2_2 \leq \frac{C}{\Lambda_i^2}(\|\xi \tilde v^\prime\|_2^2+ \Lambda_i
\|(\Lambda_i^{-1/2}\xi^2-\mu_r) \,\tilde g\|_2^2 + \epsilon^2\|\tilde v\|_2^2)\,.
\end{equation}
To bound $\|\xi \tilde v^\prime\|_2^2$  we use the identity
\begin{multline*}
 \Re\langle \xi^2 \tilde v, \tilde g \rangle=
 \|\xi \tilde v^\prime\|_2^2 - \|\tilde v\|_2^2 - \Re z_0\,\|\xi \tilde v\|_2^2 \\  
+\epsilon\Lambda_i^{-1/2}\|\xi(\Lambda_i^{-1/2}\xi^2-\mu_r) \tilde v\|_2^2 +\epsilon\, \Re \Big\langle
    \xi^2\tilde v,\Phi(\epsilon,\Lambda_i^{-1/4}\xi,\mu) \tilde v\Big\rangle \,,
\end{multline*}
from which we obtain
\begin{equation}
\label{eq:28}
  \|\xi \tilde v^\prime\|_2^2 \leq (\Re z_0)_+\, \|\xi \tilde v\|_2^2 + C(\|\tilde v\|_2^2 +
  \|\xi^2\tilde g\|_2^2)\,.
\end{equation}
We now use the identity
\begin{displaymath}
   -\Im\langle \tilde v, \tilde g \rangle= 2  \|\xi \tilde v\|_2^2 -2 \mu_r \Lambda_i^{1/2} \|\tilde v\|_2^2 -\epsilon\Lambda_i^{-1/2}\Im\Big\langle \tilde v,\Phi(\epsilon,\Lambda_i^{-1/4}\xi,\mu) \tilde v\Big\rangle \,,
\end{displaymath}
to obtain
\begin{displaymath}
 2 \|\xi \tilde v\|_2^2 \leq    \|\tilde v\|_2 \,  \| \tilde g\|_2   + [2 (\mu_r)_+   \Lambda_i^{1/2}+C  \epsilon\Lambda_i^{-1/2}]\|\tilde v\|_2^2 \,,
\end{displaymath} 
which implies, having in mind that $\Lambda_i>\frac 12$ and $\mu_r\leq \varrho\,$, 
\begin{equation}
\label{eq:50}
\| \xi \tilde  v\|_2^2   \leq      \hat C \, ( \Lambda_i^\frac 12  \|\tilde v\|_2^2 + \Lambda_i ^{-\frac 12} \| \tilde g\|_2^2)  \,.
\end{equation}
Substituting the above into \eqref{eq:28} yields
\begin{displaymath}
   \|\xi \tilde v^\prime\|_2^2 \leq C\,\Big([\Lambda_i^{1/2}(\Re z_0)_++1]\| \tilde v\|_2^2 + 
  \|(\xi^2  +\epsilon^{-1/2}\Lambda_i^{1/2})\tilde g\|_2^2\Big)\,,
\end{displaymath}
which yields, upon substitution into \eqref{eq:27},
\begin{equation}
  \label{eq:29}
\|(\Lambda_i^{-1/2}\xi^2-\mu_r)^2 \tilde v\|_2^2\leq C\,
\Big(\Lambda_i^{-1}\|(\Lambda_i^{-1/2}\xi^2+\epsilon^{-1/2}) \tilde g\|_2^2 + \Big[\frac{(\Re z_0)_+}{\Lambda_i^{3/2}}+1\Big] \| \tilde v\|_2^2 \Big)\,. 
\end{equation}
To complete the proof,  we observe from \eqref{eq:99z} applied to $\tilde v$,  that
\begin{displaymath}
  (\hg-z_0)\tilde v
  =\frac{\epsilon}{\Lambda_i^{1/2}}[(\Lambda_i^{-1/2}\xi^2-\mu_r )^2 +
  \Phi(\epsilon,\Lambda_i^{-1/4}\xi,\mu)] \tilde v +\tilde g\,.
\end{displaymath}
By \eqref{eq:31} and \eqref{eq:29}  it holds that
\begin{equation}\label{eq:29a}
  \| \tilde v\|_2 \leq C \Big(\frac{1}{d(z_0,\sigma(\hg))}+
  \frac{1}{1+ |\Re z_0|^{1/3}}
  \Big)\Big[ \epsilon\Lambda_i^{-1/2} \Big(\frac{(\Re
    z_0)_+}{\Lambda_i^{3/2}}+1\Big)^{1/2}\| \tilde v\|_2 +
  \frac{\epsilon}{\Lambda_i}\|\Lambda_i^{-1/2}\xi^2 \tilde g\|_2+\|\tilde g\|_2\Big]
  \,.  
\end{equation}
We now  show that for any $\eta >0$ there exist $\epsilon_0$ and $R_0$ such
that, for all $(\epsilon ,R) \in (0,\epsilon_0]\times  [R_0,+\infty)$, for $\mu \in \Omega
(\varrho,\epsilon,R)$ we have  
\begin{equation}\label{eq:313new}
\delta(\epsilon,\mu):= \Big(\frac{1}{d(z_0,\sigma(\hg))}+  \frac{1}{1+ |\Re
  z_0|^{1/3}}  \Big) \epsilon\Lambda_i^{-1/2} \Big(\frac{(\Re
  z_0)_+}{\Lambda_i^{3/2}}+1\Big)^{1/2}\leq \eta \,, 
\end{equation}
  where
\begin{displaymath}
  \Omega (\varrho,\epsilon, R):=\Omega (\varrho,\epsilon)\cap \check \Omega (\epsilon,R)\,.
\end{displaymath}
  We consider four different cases.\\
  \begin{enumerate}
  \item If $\mu_i \geq 1$, we observe that by \eqref{eq:25b} and the
    location of $\sigma (\hg)$,  there exists $C_0>0$ such that if $\Re z_0
    \geq C_0$, then $d(z_0,\sigma(\hg)) \geq \frac 1 C_0 \Re z_0$. Thus, if
    $d(z_0,\sigma(\hg))\leq1$,  we obtain, using the fact that $\Lambda_i \geq 1$ for
    $\mu_i\geq0$, that $\delta(\epsilon,\mu)\leq C\epsilon$. For $d(z_0,\sigma(\hg))\geq1$ and $\Re z_0
    \geq C_0$ it holds that  
  \begin{equation}
\label{eq:56}
 \delta(\epsilon,\mu) \leq \hat C \epsilon (1+ |\Re z_0|^\frac{1}{6} \Lambda_i^{-\frac 54})\,.
\end{equation}
As
 \begin{equation*}
0 < \Re  z_0= \frac{\mu_i(2+\epsilon\mu_i) }{[1+\epsilon\mu_i]^{1/2}}= \frac 1 \epsilon (\Lambda_i^2 -1) \Lambda_i^{-\frac 12}\leq \frac{1}{\epsilon} \Lambda_i ^{\frac 32} \,,
\end{equation*}
we obtain that
\begin{displaymath}
 \delta(\epsilon,\mu) \leq \check C \, \epsilon^\frac 56\,.
\end{displaymath}
 \item 
In the case $\Re z_0 \leq C_0$, we have
\begin{equation}
\label{eq:58}
  \delta(\epsilon,\mu) \leq \hat C  \epsilon (1 + d(z_0,\sigma(\hg)^{-1}) \leq C \Big(\epsilon + \frac 1R \Big) \,.
\end{equation}
 \item If $0 \leq \mu_i \leq 1$ and  $d(z_0,\sigma(\hg)\leq1/2$ there exists
  $C>0$ such that $\Re z_0\leq C\, |\Im z_0|$. Furthermore, it holds that
  $\Im z_0<0$ and hence $0\leq \mu_r\leq\varrho$. Consequently, 
  \begin{displaymath}
    |\Im z_0| \leq C\varrho \,,
  \end{displaymath}
and hence $\Re z_0\leq C$ and we may proceed as in item 2.\\
 Otherwise, if
$d(z_0,\sigma(\hg)\geq1/2$, we may invoke \eqref{eq:56} and proceed as in item
1.
 \item Finally, if $\mu_i <0$, we have $\Re z_0 <0$ and $\frac 12 \leq
   \Lambda_i \leq 1$ and hence we can obtain \eqref{eq:58} once again. 
 \end{enumerate}
  For sufficiently small $\eta >0$ we can conclude the existence of
 positive $C$, $\epsilon_0$ and $R_0$ such that for all $(\epsilon ,R) \in
 (0,\epsilon_0] \times   [R_0,+\infty)$ 
  and $\mu \in \Omega (\varrho,\epsilon,R)$  it holds that
 \begin{equation}
\label{eq:26c}
  \| \tilde v \|_2 \leq  C
  \left( d(z_0,\sigma (\hg))^{-1} + \frac{ 1}{| \Re z_0|^{\frac 13} +1}
  \right)  \left( \| \tilde g\|_2 + \epsilon \Lambda_i^{-\frac 32} \|\xi^2 \tilde g\|_2 \right) \,.  
\end{equation}
We can now easily verify \eqref{eq:26}. 
\end{proof}
Coming back to the resolvent of $\tilde {\B} _\epsilon$, we now prove:
\begin{proposition}
  There exists $R_0>0$ and $C>0$, such that, for any ${\bf f}\in
  H^2(\R,\C^3)$, any $(\epsilon,R)$ satisfying $R_0<R<\epsilon^{-1}$, and $\Lambda\in
  \Vg(\epsilon,\varrho)$ for which $\mu\in\check{\Omega}(\epsilon,R)$ it holds that
  \begin{equation}
\label{eq:35}
    \|(\tilde{\B}_\epsilon-\Lambda)^{-1} { \bf f}\|_2\leq \frac{C}{R\epsilon^{5/3}}(\|  {\bf
      f} \|_2+\epsilon^{2} \| {\bf f}_\perp^{\prime\prime}\|_2 + (1+\Lambda_i^2)^{-1/2}\|x  {\bf f}_\perp\|_2) \,,
  \end{equation}
 where ${\bf f}_\perp=(f_1, f_2,0)$.
\end{proposition}
\begin{proof}
Consider (see \eqref{eq:65}) a triple $(\tilde {\bf u}, {\bf f}, \Lambda)$
such that 
\begin{displaymath}
  (\widetilde \B_\epsilon - \Lambda) \tilde{\bf u} = {\bf f} =\epsilon^{\frac 23}  \tilde{\bf  f } \,.
\end{displaymath}
   We first recall  that for $\Lambda_r\leq-1$ we have (see Remark \ref{Rem3.3})
  \begin{displaymath}
   \epsilon^2 \|\nabla \tilde {\bf u}\|_2^2 -\Lambda_r\|\tilde {\bf u}\|_2^2= \Re \langle \tilde  {\bf u},{\bf f}\rangle \,,
  \end{displaymath}
and hence
\begin{displaymath}
  \| \tilde {\bf u}\|_2 \leq  |\Lambda_r|^{-1} \| {\bf f}\|_2\,.
\end{displaymath}
Consequently, we have for $\epsilon^{- \frac 23} \Lambda_r= \lambda_r < -1$
\begin{equation}\label{eq:62a}
   \|(\tilde{\B}_\epsilon-\Lambda)^{-1} { \bf f}\|_2 \leq  \epsilon^{-\frac 23} \| {\bf f}\|_2\,,
\end{equation}
which implies \eqref{eq:35} in this case.\\
Through the rest of this proof we assume in addition that $\lambda_r>-1\,.$  \\

From \eqref{eq:64} which reads $\xi =\tilde \omega \Lambda_i^\frac 14$, and
  \eqref{eq:46} which reads $\omega = \epsilon^{\frac 16} \tilde \omega$, we get $\xi
  = \epsilon^{-1/6} \Lambda_i ^{1/4}\,\omega $. From \eqref{eq:53} we then conclude that
  \begin{displaymath}
 \lambda =\epsilon^{-2/3} \Lambda\,,\,   \tilde v(\xi)= (\omega^2-\lambda)^{-\frac 12} \hat{u}_s(\omega) \mbox{ and } \tilde{g}(\xi)=\Lambda_i^{-1/2}\epsilon^{1/3} g(\omega)\,.
  \end{displaymath}

\paragraph{Estimation of $\hat u_s$ for $\lambda_r>-1$.}~\\
Recall that $\hat{u}_s$ is introduced in \eqref{eq:defhatu}, and that
the definition of $g$ is given in \eqref{eq:8b} which reads
\begin{equation}\label{eq:defgnew}
   g:=  \hat{f}_s + \epsilon^{-2/3}\frac{\hat{f}_3}{\omega^2-\lambda}
 +  \frac{d}{d\omega}\Big(\frac{\hat{f}_d}{\omega^2-\lambda}\Big)\,.
\end{equation}
 Thus, by using \eqref{eq:26} and \eqref{eq:50} we obtain, for $R_0 < R < \frac 1 \epsilon$,
 \begin{displaymath}
\| (\epsilon^{\frac 13} \Lambda_i ^{-\frac 12} \xi^2 +1)^\frac 12  \tilde v(\xi) \|_2
\leq  \frac{C}{R\epsilon}\,   \left( \| \tilde g\|_2 + 
    \epsilon\Lambda_i^{-3/2} \| \xi^2 \tilde g\|_2\right) + C \epsilon^\frac 16
\Lambda_i^{-\frac 12} \| \tilde g\|\,, 
\end{displaymath}
which implies for a new constant $C>0$
\begin{displaymath}
\| (\epsilon^{\frac 13} \Lambda_i ^{-\frac 12} \xi^2 +1)^\frac 12  \tilde v(\xi) \|_2
\leq  \frac{C}{R\epsilon} \,   \left( \| \tilde g\|_2 + 
    \epsilon\Lambda_i^{-3/2} \| \xi^2 \tilde g\|_2\right) \,,
\end{displaymath}
and hence
\begin{equation*}
  \Big\|\Big[\frac{\omega^2+1}{\omega^2-\lambda}\Big]^{1/2}\hat{u}_s\Big\|_2 \leq  \frac{C}{R} \epsilon^{-\frac 23}   
   \Lambda_i^{-1/2}  \left(  \| g\|_2 +
     \epsilon^{2/3}\Lambda_i^{-1}  \| \omega^2 g  \|_2 \right)  \,.   
\end{equation*}
As for all $\lambda\in\C$ satisfying $-1\leq\lambda_r \leq 1 $ and $\lambda_i \geq 1$
\begin{displaymath} 
  \Big|\frac{\omega^2+1}{\omega^2-\lambda}\Big|\geq \frac{1}{\sqrt{2} \, \lambda_i} \,,
\end{displaymath}
we obtain, for sufficiently large $R_0$ that, for $R_0 < R < \frac 1 \epsilon$,
\begin{equation}\label{eq:119b}
  \|\hat{u}_s\|_2 \leq   \frac{C}{R\epsilon} \left(  \| g\|_2 +
     \epsilon^{2/3}\Lambda_i^{-1}  \| \omega^2 g  \|_2 \right) \,.   
\end{equation}
We now estimate of the right hand side in \eqref{eq:119b}
using \eqref{eq:defgnew}. \\[1.5ex]

{\bf  Estimation of terms involving $\hat f_3$.}\\[1.5ex]
 We begin by observing that 
\begin{displaymath}
\epsilon^{-\frac 23} \Big\| \frac{\hat f_3}{\omega^2-\lambda}\Big\|_2 \leq \epsilon^{-\frac 23} \lambda_i^{-1} \|\hat f_3\|_2= \Lambda_i^{-1}  \|\hat f_3\|_2\,.
\end{displaymath}
Next, to estimate $\Lambda_i^{-1} \Big\| \frac{\omega^2 \hat
  f_3}{\omega^2-\lambda}\Big\|_2$ we observe that, for $\lambda_r \in (-1,+1)$ and
$\lambda_i \geq 1$, it holds that 
\begin{equation}
\label{ineqaba} 
\Big|\frac{\omega^2
}{\omega^2-\lambda}\Big| \leq 2\,.
  \end{equation}
Hence, we can conclude that
\begin{displaymath}
\Lambda_i^{-1}  \Big\| \frac{\omega^2 \hat f_3}{\omega^2-\lambda}\Big\|_2 \leq 2 \Lambda_i^{-1} \| \hat f_3\|_2\,.
\end{displaymath}

{\bf Estimation of terms involving $\hat f_d$.}\\[1.5ex]
 We first write
\begin{displaymath}
\frac{d}{d\omega}\Big(\frac{\hat{f}_d}{\omega^2-\lambda}\Big)= \frac{\hat{f}_d'}{\omega^2-\lambda} - \frac{2 \omega \hat{f}_d}{(\omega^2-\lambda)^2}\,.
\end{displaymath}
For the first term on the right-hand-side we conclude that
\begin{displaymath}
\| \frac{\hat{f}_d'}{\omega^2-\lambda}\|_2 \leq \epsilon^\frac 23 \Lambda_i^{-1}\| \hat{ f}'_d \|_2\,.
\end{displaymath}
Then, with the aid \eqref{ineqaba}, valid for $\lambda_r \in (-1,+1)$ and
$\lambda_i \geq 1$, we obtain that
\begin{displaymath}
\Big|\frac{2 \omega }{(\omega^2-\lambda)^2}\Big|= 2 \Big|\frac{ \omega^2 }{(\omega^2-\lambda)}\Big|^\frac 12 \, \Big|\frac{1 }{(\omega^2-\lambda)^{\frac 32} }\Big|\leq 2 ^{\frac 32} \epsilon\Lambda_i^{-3/2} \,.
\end{displaymath}
Consequently,
\begin{displaymath}
\Big\| \frac{2 \omega \hat{f}_d}{(\omega^2-\lambda)^2}\Big\|_2  \leq  2 ^{\frac 32} \epsilon\Lambda_i^{-3/2}\| \hat f_d\|_2\,.
\end{displaymath}
Summarizing the above yields
\begin{equation}\label{eq:ctrlg}
\|g\|_2 \leq C \big[\| \hat{f}_s\|_2 + \Lambda_i^{-1}  
   \|\hat{f}_3\|_2
 +  \epsilon^{2/3} \Lambda_i^{-1}  \|\hat{f}_d^\prime\|_2+\epsilon  \Lambda_i^{-\frac 32} \|\hat{f}_d\|_2\big]\,.
\end{equation}

Next, using \eqref{ineqaba} once again yields 
\begin{displaymath}
\epsilon^\frac 23 \Lambda_i^{-1}  \Big\| \frac{\omega^2 \hat{f}_d'}{\omega^2-\lambda}\Big\|_2
\leq 2  \epsilon^\frac 23 \Lambda_i^{-1} \| \hat{f}_d'\|_2\,. 
\end{displaymath}
Finally, it holds that
\begin{displaymath}
 2 \epsilon^\frac 23 \Lambda_i^{-1}  \Big\| \frac{\omega^3
   \hat{f}_d}{(\omega^2-\lambda)^2}\Big\|_2 \leq 2^\frac 52  \epsilon \Lambda_i^{-\frac 32}
 \| \hat f_d\|_2\,. 
\end{displaymath}
In conclusion, we have proved the existence of $C>0$ and $R_0>0$ such
that for $R_0 < R < \frac 1 \epsilon$, $\Lambda\in \Vg(\epsilon,\varrho)\cap\{\Lambda_r>-\epsilon^{2/3}\}$ for which
$\mu\in\check{\Omega}(\epsilon,R)$
\begin{equation}
\label{eq:62b}
  \|\hat{u}_s\|_2\leq   \frac{C}{R\epsilon  } \big[\|(\epsilon^{2/3}\Lambda_i^{-1}\omega^2+1)\hat{f}_s\|_2 +
   \|\hat{f}_3\|_2
 +  \epsilon^{2/3} \hat \Lambda_i^{-1} \|\hat{f}_d^\prime\|_2+\epsilon \|\hat{f}_d\|_2\big]\,.  
\end{equation}

\paragraph{Estimation of $\hat{u}_3$.}
 As in the proof of Proposition \ref{cor:spect}  we write (\ref{eq:5}c) in the form 
\begin{displaymath}
\hat{u}_3 = \frac{\epsilon^{- 2/3}}{(\omega^2-\lambda)\sqrt{2}}\hat{u}_s +  \frac{\hat{f}_3}{\omega^2 -\lambda}  \,,
\end{displaymath}
which implies
\begin{displaymath}
\|\hat{u}_3 \|_2 \leq  \frac{1}{|\Lambda_i| \sqrt{2}}\,\| \hat{u}_s\|_2   +
\epsilon^{2/3} \frac{1}{|\Lambda_i|} \, \|\hat{f}_3\|_2\leq C  \Lambda_i^{-1}   (\|
\hat{u}_s\|_2   +  \epsilon^{2/3}\, \|\hat{f}_3\|_2) \,. 
\end{displaymath}
Hence, using the fact that $R\epsilon \leq 1$, we obtain by \eqref{eq:62b}
  \begin{equation}
\label{eq:59}
  \|\hat{u}_3\|_2\leq   \frac{C}{R\epsilon  } \big[\|(\epsilon^{2/3}\Lambda_i^{-1}\omega^2+1)\hat{f}_s\|_2 +
   \|\hat{f}_3\|_2
 +  \epsilon^{2/3} \hat \Lambda_i^{-1} \|\hat{f}_d^\prime\|_2+\epsilon \|\hat{f}_d\|_2\big]\,.  
\end{equation}

  \paragraph{Estimation of $\hat{u}_d$}
We begin with an estimate of $\hat u'_s$.  By \eqref{eq:22} which reads
   \begin{multline*}
 \Re\langle(\omega^2-\bar{\lambda})\hat{u}_s,g\rangle =   \|(\hat{u}_s)^\prime\|_2^2 + 
 \Re \Big\langle2\omega \hat{u}_s,\frac{(\hat{u}_s)^\prime}{\omega^2-\lambda}\Big\rangle\\  +(\varepsilon^{-1}- \lambda_i^2)\|\hat{u}_s\|_2^2 +
  \|(\omega^2-\lambda_r)\hat{u}_s\|_2^2 \,,
\end{multline*}
we may conclude that
  \begin{multline*}
 \|(\omega^2-\bar{\lambda})\hat{u}_s\|_2\, \| g \|_2 \geq   \|(\hat{u}_s)^\prime\|_2^2 
 -2 \| \hat{u}_s\|_2 \Big\| \frac{ \omega (\hat{u}_s)^\prime}{\omega^2-\lambda}\Big\| _2 +(\varepsilon^{-1}- \lambda_i^2)\|\hat{u}_s\|_2^2  +  \|(\omega^2-\lambda_r)\hat{u}_s\|_2^2 \,.
\end{multline*}
Hence,
 \begin{equation*}
 \|(\hat{u}_s)^\prime\|_2^2  \leq C \left(   \| g \|^2_2  +  \| \frac{ \omega (\hat{u}_s)^\prime}{\omega^2-\lambda}\|^2_2  + \lambda_i^2 \|\hat{u}_s\|_2^2\right) \,.
\end{equation*}
Using \eqref{ineqaba}, we get 
\begin{displaymath}
 \Big\| \frac{ \omega (\hat{u}_s)^\prime}{\omega^2-\lambda}\Big\|_2 \leq \sqrt{2} \lambda_i^{-\frac 12}  \| \hat{u}_s^\prime\|_2 \leq C \epsilon^{\frac 13} \| \hat{u}_s^\prime\|_2\,.
\end{displaymath}
Consequently, we can conclude that under the assumed conditions on
$\Lambda,\epsilon, R$ it holds that
\begin{equation}\label{3.115}
  \|\hat{u}_s^\prime\|_2^2\leq C(\|g\|_2^2+ \lambda_i^2\|\hat{u}_s\|_2^2) \,.
\end{equation}
We can, thus, deduce  from  \eqref{eq:6} that
\begin{equation*}
\|\hat u_d\| \leq C  \epsilon^{2/3} \Lambda_i^{-1}  \left( \| \hat f_d\| +\| \hat u_s'\|\right) \,,
\end{equation*}
which leads to
\begin{equation}\label{3.115a}
\|\hat u_d\| \leq C  \epsilon^{2/3} \Lambda_i^{-1}  \left( \| \hat f_d\|  + \|g\|_2\right)   + C \|\hat{u}_s\|_2  \,.
\end{equation}

Substituting  \eqref{eq:ctrlg}   into \eqref{3.115a} yields
\begin{equation}\label{3.115b}
\|\hat u_d\| \leq C  \epsilon^{2/3}   \left( \| \hat f_d\|  + \| \hat{f}_s\|_2 +   
   \|\hat{f}_3\|_2
 +  \epsilon^{2/3}  \Lambda_i^{-1}   \|\hat{f}_d^\prime\|_2 \right) + C \|\hat u_s\|_2  \,.
\end{equation}
With the aid of \eqref{eq:62b} we then conclude
\begin{equation}
\label{eq:62bc}
  \|\hat{u}_d\|_2\leq   \frac{C}{R\epsilon} \big[\|(\epsilon^\frac 23 \omega^2+1)\hat{f}_s\|_2 +
  \|\hat{f}_3\|_2
 +  \epsilon^{2/3}\Lambda_i^{-1}\|\hat{f}_d^\prime\|_2+ \|\hat{f}_d\|_2\big]\,.  
\end{equation}
Combining the above with \eqref{eq:62b} and  \eqref{eq:59} then yields
for $R_0< R < \epsilon^{-1}$ 
  \begin{equation}
\label{eq:59f} 
  \|\hat{\bf u}\|_2\leq   \frac{C}{R\epsilon  } \big[\|(\epsilon^{2/3}\omega^2+1)\hat{f}_s\|_2 +
   \|\hat{f}_3\|_2
 +  \epsilon^{2/3} \Lambda_i^{-1} \|\hat{f}_d^\prime\|_2+ \|\hat{f}_d\|_2\big]\,.  
\end{equation}
In term of the original variables $({\bf u},{\bf f})$ \eqref{eq:59f}
reads 
  \begin{displaymath}
 \| {\bf u}\|_2 \leq     \frac{C}{R \epsilon^\frac 53} \left( \epsilon^2\|  {\bf
     f}_\perp'' \|_2 + \|{\bf f}\|_2
       + \Lambda_i^{-1}  \| x {\bf f}_\perp\|_2\, \right)\,.
  \end{displaymath}
\end{proof}

\subsection{Proof of Theorem \ref{thm:unbounded}} 
We now complete the proof of  Theorem \ref{thm:unbounded}. Let ${\bf u}\in
D(\B_\epsilon)$ and ${\bf f}\in L^2(\R,\C^3)$ satisfy $(\B_\epsilon-\Lambda){\bf u}={\bf
  f}$ for some $\Lambda=\Lambda_r+i \Lambda_i\in\C$. As above we consider
only the case $ \Lambda_i>0$.\\
{\bf The case $0< \Lambda_i<1/2$.}\\ Here we have by
 \eqref{eq:21b} for $\Lambda_r\leq \varrho \epsilon$ 
\begin{equation}
  \label{eq:61}
\|(\B_{\epsilon}-\Lambda)^{-1}\| =\|(\tilde{\B}_{\epsilon}-\Lambda)^{-1}\| \leq C\epsilon^{2/3}\Big(\epsilon^{-2/3}+ \frac{1}{\Lambda_i}\Big)^2 \,.
\end{equation}
\\
{\bf The case $ \Lambda_i>1/2$}  \\ Here we attempt to use \eqref{eq:35}, to which end we first
observe that
\begin{displaymath}
  \Re \langle{\bf u},(\tilde{\B}_\epsilon+1-i\Lambda_i){\bf u}\rangle=\epsilon^2\|\frac{d {\bf
    u}}{dx} \|_2^2 + \|{\bf u}\|_2^2\,.
\end{displaymath}
It follows that
\begin{equation}
\label{eq:81}
  \|(\tilde{\B}_\epsilon+1-i \Lambda_i)^{-1}\|+\epsilon \|(\tilde{\B}_\epsilon+1-i \Lambda_i)^{-1}\|_{\LL(L^2,H^1)}\leq 3 \,.
\end{equation}
 Integration by parts yields that, for any ${\bf w} \in D(\tilde{\B}_\epsilon)$, 
\begin{displaymath}
  - \Re\Big\langle\frac{d^2{\bf w} }{dx^2},(\tilde{\B}_\epsilon+1-i\Lambda_i){\bf w} \Big\rangle =
  \epsilon^2\Big\|\frac{d^2{\bf w}}{dx^2} \Big\|_2^2 -
\Big\|\frac{d{\bf w}}{dx} \Big\|_2^2   - \Im\langle \frac{d w_1}{dx},w_1\rangle+ \Im\langle\frac{dw_2}{dx},w_2\rangle
\end{displaymath}
Consequently, we obtain that
\begin{equation}
\label{eq:37}
  \|(\tilde{\B}_\varepsilon+1-i \Lambda_i)^{-1}\|_{\LL(L^2,H^2)}\leq \frac{C}{\epsilon^2 }\,.
\end{equation}
Finally, recall that by \eqref{eq:99}
\begin{displaymath}
  D(\tilde{\B}_\epsilon) = \{ {\bf u}\in H^2(\R,\C^3)\,|\,x{\bf u}_\perp\in
  L^2(\R,\C^3)\} \,,
\end{displaymath}
where ${\bf u}_\perp=(u_1,u_2,0)$.\\
For $s>0$,   we equip $D(\tilde{\B}_\epsilon)$ with the norm
\begin{displaymath} 
  \|{\bf u}\|_{(\B,s,\epsilon)}= \epsilon^2 \|{\bf u}''\|_{2}+ \epsilon\|{\bf u}'\|_{2} + \|{\bf u}\|_{2} +  s^{-1}\|x{\bf u}_\perp\|_2\,,
\end{displaymath}
and denote this normed space by $D_{s,\epsilon}$. \\
An integration by parts yields (see \eqref{eq:124})
\begin{multline*}
  \Im\langle(-xu_1,xu_2,0)\,,\, (\tilde{\B}_\varepsilon+1-i \Lambda_i){\bf u}\rangle\\ = \|x{\bf
    u}_\perp\|_2^2  +\Lambda_i(\langle xu_1,u_1\rangle-\langle xu_2,u_2\rangle)+ \frac{1}{\sqrt{2}}(\Im\langle x(u_2-u_1)\,, \, u_3\rangle)
  \\ +\epsilon^2(\Im\langle u_2,u_2^\prime\rangle-\Im\langle u_1,u_1^\prime\rangle)\,.
\end{multline*}
 Consequently,  combining the above with \eqref{eq:81}, we get 
\begin{equation}\label{eq:3.33}
\Lambda_i^{-1}\,   \|x \, \pi_\perp (\tilde{\B}_\varepsilon+1-i \Lambda_i)^{-1}\|\leq C\,,
\end{equation}
where $\pi_\perp$ denotes the projection on the two first components.

We now choose $s=\Lambda_i$.
Combining \eqref{eq:3.33} with \eqref{eq:37} yields
\begin{equation}
\label{eq:122} 
  \|(\tilde{\B}_\varepsilon+1-i \Lambda_i)^{-1}\|_{\LL(L^2,D_{\Lambda_i,\epsilon})}\leq C \,.
\end{equation}

Let $\Lambda\in\Vg(\epsilon,\varrho)$ satisfy $z_0(\mu,\epsilon)\in \check \Omega(\epsilon,R)$ where
$z_0$ is given by \eqref{eq:25a}, $\check \Omega$ by \eqref{eq:63}, and
$R_0$ is sufficiently large, so that \eqref{eq:35} holds true for $R_0 \leq R <\frac 1 \epsilon$.  We
begin by observing that  \eqref{eq:35} implies
\begin{displaymath}
  \|(\tilde{\B}_\varepsilon-\Lambda)^{-1}\|_{\LL(D_{\Lambda_i,\epsilon}\,,\,L^2)}\leq\frac{C}{R\epsilon^{5/3}}\,.
\end{displaymath}
We now
use the resolvent identity
\begin{displaymath}
  (\tilde{\B}_\epsilon-\Lambda)^{-1}=  (\tilde{\B}_\epsilon+1-i \Lambda_i)^{-1}+(\Lambda_r+1)(\tilde{\B}_\epsilon-\Lambda)^{-1} (\tilde{\B}_\epsilon+1-i \Lambda_i)^{-1}\,,
\end{displaymath}
to establish that
\begin{displaymath}
  \|(\tilde{\B}_\epsilon-\Lambda)^{-1}\|\leq  \|(\tilde{\B}_\epsilon+1-i
  \Lambda_i)^{-1}\|+C\|(\tilde{\B}_\epsilon-\Lambda)^{-1}\|_{\LL(D_{\Lambda_i,\epsilon},L^2)}
  \|(\tilde{\B}_\epsilon+1-i \Lambda_i)^{-1}\|_{\LL(L^2,D_{\Lambda_i,\epsilon})} \,.
\end{displaymath}
We may conclude from the above  that:
\begin{proposition}\label{prop3.20}
 Let $\rho >0$. There exist $R_0>0 $  and $\epsilon_0>0$ such that for any
 $0<\epsilon<\epsilon_0$, $\Lambda \in \mathcal V(\epsilon,\rho)$,  $ R_0  < R < \frac{1}{ \epsilon}$,
 $\Lambda_i > \frac 12$, and   $z_0(\mu,\epsilon)\in \check \Omega(\epsilon,R)$ it holds that
   \begin{equation}
\label{eq:62}
 \|(\tilde{\B}_\epsilon-\Lambda)^{-1}\| \leq  C \Big(1 + \frac{1}{R\epsilon^\frac 53}\Big)\,.
\end{equation}
\end{proposition} 

We finally provide a more explicit condition in guaranteeing the
validity of the assumptions of proposition \ref{prop3.20}. We
introduce for $\hat R>0$, $\epsilon >0$
\begin{equation}
\hat D^+(\epsilon,\hat R) =\{ \Im \Lambda > 0\,,\, d( \Lambda - i ,   \epsilon \sigma(\hg^*) > \hat R \epsilon^2\}\,,
\end{equation}
Note that for all $0<\epsilon $ the  set $\{ \Lambda \in \mathbb C
  \,,\, d( \Lambda,   \epsilon  \sigma(\hg^*)) < \frac 12 \epsilon\} $ is a union of
disjoint disks. 
\begin{lemma}\label{lem3.21}
  Let $\rho >0$. There exist $\hat R_0>1$ and $\epsilon_0>0$, such that for
  any $ \hat R_0 < \hat R < 1/(\sqrt{2}\epsilon)$, and $\Lambda \in\mathcal
  V(\epsilon,\rho)\cap \mathcal D^+(\epsilon,\rho, \hat R)$, we have that
  $z_0(\mu,\epsilon)\in \check \Omega(\epsilon,R)$ for all $R\geq \sqrt{2}\hat R$.
\end{lemma}
\begin{proof} 
  Let $1<\hat{R}_0<\hat{R}<[\epsilon\sqrt{2}]^{-1}$ and $0<R\leq\sqrt{2}(\hat
  R-2N_\varrho^2)$, where $N_\varrho$ is given by \eqref{eq:131}. (Note that
  $R<\sqrt{2}\hat R<\epsilon^{-1}$.) Suppose for a contradiction that for
  some $n\leq N_\varrho$, it holds that $\Lambda \in\mathcal V(\epsilon,\rho)\cap \mathcal
  D^+(\epsilon,\rho, \hat R)$ but
  \begin{equation}
\label{eq:129}
     d(z_0,(2n-1)[1-i])\leq R\epsilon\,.
  \end{equation}
Then,   $ |\Re z_0- (2n-1)| \leq R\epsilon$ which can be rewritten in the form
\begin{equation}
\label{eq:70}
  \Big|\frac{\Lambda_i^2-1}{\Lambda_i^{1/2}}-(2n-1)\epsilon\Big|\leq R\epsilon^2\,.
\end{equation}
Clearly, $\Lambda_i\geq1$, otherwise we have
\begin{displaymath}
   \Big|\frac{\Lambda_i^2-1}{\Lambda_i^{1/2}}-(2n-1)\epsilon\Big|\geq(2n-1)\epsilon \,,
\end{displaymath}
contradicting \eqref{eq:70} as $R\epsilon < 1$.   \\
For $\Lambda_i\geq1$ we have, as $( \Lambda_i+1)\geq2\Lambda_i^{1/2}$,
\begin{displaymath}
   R\epsilon^2\geq\frac{\Lambda_i^2-1}{\Lambda_i^{1/2}}+ (2n-1)\epsilon\geq 2(\Lambda_i-1)-(2n-1)\epsilon \,,
\end{displaymath}
and hence
\begin{equation}
\label{eq:128}
1 \leq   \Lambda_i\leq 1+ \frac{2n-1}{2}\epsilon + \frac{R}{2}\epsilon^2 \,.
\end{equation}
Returning to \eqref{eq:70}, we write
\begin{displaymath}
-R\epsilon^2\leq\frac{\Lambda_i^2-1}{\Lambda_i^{1/2}}-(2n-1)\epsilon\leq \Big(2+
\frac{2n-1}{2}\epsilon + \frac{R}{2}\epsilon^2\Big)(\Lambda_i-1) -(2n-1)\epsilon\,,
\end{displaymath}
which leads to (as $n\leq N_\varrho$ where $N_\varrho$ is defined by \eqref{eq:131})
\begin{displaymath}
  \Lambda_i \geq 1+ \frac{(2n-1)\epsilon-R\epsilon^2}{2+
N_\rho \epsilon + \frac{R}{2}\epsilon^2}\geq 1+\frac{1}{2}
[(2n-1)\epsilon-R\epsilon^2]\Big[1-N_\rho \epsilon - \frac{R}{2}\epsilon^2\Big]  \,,
\end{displaymath}
and hence
\begin{displaymath}
  \Lambda_i - 1 \geq \frac{2n-1}{2}\epsilon - \left( \frac R2  +  N_\rho \Big(N_\rho  +
     \frac{R}{2}\epsilon\Big) \right) \epsilon^2\geq \frac{2n-1}{2}\epsilon -
  \Big(\frac{R}{2}+2N_\varrho^2\Big)\epsilon^2  \,.
\end{displaymath}
Combining the above with \eqref{eq:128} yields 
\begin{equation}
\label{eq:132}
    \Big|\Lambda_i-1- \frac{2n-1}{2}\epsilon\Big|\leq  \Big(\frac{R}{2}+2N_\varrho^2\Big) \epsilon^2 \,. 
\end{equation}
By \eqref{eq:129} we have, in addition to \eqref{eq:70},
\begin{equation}\label{eq:140z}
|\Im z_0 + (2n-1)| \leq R \epsilon\,.
\end{equation}
In view of \eqref{eq:132} we can rewrite 
\begin{equation*}
\Big| \frac 12(2n-1) -\mu_r \Big| \leq  \Big(\varrho N_\varrho + \frac{R}{2}\Big) \epsilon\,,
\end{equation*}
or equivalently
\begin{displaymath}
| \frac 12(2n-1)\epsilon  -\Lambda_r | \leq  \Big(N_\varrho^2  + \frac{R}{2}\Big) \epsilon^2\,.
\end{displaymath}
Combining the above with \eqref{eq:132} yields
\begin{equation}
\label{eq:130}
 d( \Lambda,  i - (1+i)  \frac{2n-1}{2}\epsilon)  \leq \sqrt{2}\Big(2N_\varrho^2  +
 \frac{R}{2}\Big) \epsilon^2<\hat{R}\epsilon^2\,.
\end{equation}
A contradiction.
\end{proof}

Combining Lemma \ref{lem3.21} with Proposition \ref{prop3.20} (applied
with $R=\sqrt{2} \hat R$) establishes \eqref{eq:107} for $\Lambda_i>1/2$.
We then use \eqref{eq:61} to obtain \eqref{eq:107} for all $\Lambda_i>0$,
$\Lambda_r<\varrho\epsilon$ and $ \Lambda\in \hat D^+(\hat R,\varrho, \epsilon)$.
Remark~\ref{rem:original-spectrum} and Lemma \ref{lem:point-exist}
together establish \eqref{eq:105}.  \newpage
\section{Finite intervals}
\label{sec:3}
As stated in Section \ref{sec:case-with-boundary} we expect that the
behaviour of $\B_\epsilon$ acting on $\R$, can intuitively explain the
behavioir of $\B_\epsilon$ on a bounded interval, in the limit $\epsilon\to0$.  In
the following, we focus on the part of $\sigma(\B_\epsilon)$ that tends to
$\R_+$ (which is in the spectrum of $\B_\epsilon$ on $\R$ by Remark
\ref{rem:essential-unit}).
\subsection{The problem}
We now define the operator $\B_\epsilon^I$ whose associated differential operator is given
by \eqref{eq:1} in the interval $I=(a,b)$, and its domain by
\begin{displaymath}
  D(\B_\epsilon^I)=H^2((a,b),\C^3)\cap H^1_0((a,b),\C^3) \,.
\end{displaymath}
Let $({\bf u},\Lambda)\in D(\B_\epsilon^I)\times\C$ denote an eigenpair of
$\B_\epsilon^I$. 
 Let further $\LL_\pm=-\epsilon^2\frac{d^2}{dx^2}\pm ix$ be defined on
$H^2(a,b)\cap H^1_0(a,b)$. Then by (\ref{eq:65}) (with ${\bf f}=0$)
and \eqref{eq:124} we may set
\begin{displaymath}
  u_1=-(\LL_--\Lambda)^{-1}u_3 \quad ; \quad
  u_2=-(\LL_+-\Lambda)^{-1}u_3  \,,
\end{displaymath}
and hence by (\ref{eq:65})
and \eqref{eq:124}  we obtain
\begin{subequations}
  \label{eq:67}
  \begin{equation}
(\PP_\Lambda -\Lambda)u_3=0 \,,
\end{equation}
where 
\begin{equation}
   \PP_\Lambda
   \overset{def}{=}-\epsilon^2\frac{d^2}{dx^2}+\frac{1}{2}[(\LL_--\Lambda)^{-1}+(\LL_+-\Lambda)^{-1}]
\end{equation}
\end{subequations}
is defined on $$D(\PP_\Lambda)=H^2((a,b),\C)\cap H^1_0((a,b),\C)\,.$$
 Note that
$(\LL_\pm-\Lambda)^{-1}$ is well defined, for sufficiently small $\epsilon$,
whenever $ \Re\Lambda<\epsilon^{2/3}|\nu_1|/2$  where $\nu_1$ denotes the leftmost zero of
Airy's function \cite{abst72}. \\
 Some intuition can be gained by considering the operator
 \begin{multline}\label{eq:67ter}
 \A_\Lambda =(\LL_- -\Lambda) (\PP_\Lambda-\Lambda)(\LL_+- \Lambda) =\\- \epsilon^2  \Big[(\LL_- -\Lambda)
 \frac{d^2}{dx^2} (\LL_+- \Lambda) +  \frac{d^2}{dx^2}\Big] -\Lambda[1+
 (\LL_--\Lambda)(\LL_+ - \Lambda)] \,,
\end{multline}
where $ D(\A_\lambda)=\{ w\in H^6((a,b)\,,\,
w(a)=w(b)=w''(a)=w''(b)=0\}$. Note that for $\Lambda\leq C\epsilon^2$ it holds that
$0\in\sigma(\A_\Lambda)\Leftrightarrow0\in\sigma(\PP_\Lambda-\Lambda)$ since  $(\LL_\pm -\Lambda)$ is invertible. 
Assuming $\Lambda=\Lambda_0 \epsilon^2$ where $\Lambda_0$ is independent of $\epsilon$ we get 
\begin{equation}\label{eq:68bis}
(\epsilon^{-2}A_\Lambda)\big|_{\epsilon=0}= - x \frac{d^2}{dx^2} x - \frac{d^2}{dx^2} -\Lambda_0(1+x^2)
\end{equation}
Let $\tilde{\A}_\Lambda=(1+x^2)^{-1}(\epsilon^{-2}A_\Lambda)\big|_{\epsilon=0}$which is
selfadjoint on $L^2((a,b), (1+x^2))$.  By the foregoing discussion, we
expect $\Lambda_0$ to be the ground state of $\tilde{\A}_\Lambda$. Denote the
principal eigenfunction by $w_0$. We expect that $u_3\approx(\LL_+-\Lambda)w_0$
and hence, by setting $\epsilon=0$ once again we can conclude that $u_3\approx
xw_0$. In fact, we shall rigorously corroborate this approximation in
the sequel.

The above heuristical argument suggests in addition that for any $C>0$
there exists $\epsilon_C$ such that for all $\epsilon \in (0,\epsilon_C]$, we have
$\sigma(\B_\epsilon^I) \cap \{ \Re \Lambda < C \epsilon^2\} \subset\R_+$.  Furthermore, while we do
not prove that in the following, one may expect to obtain for any
fixed $j\geq1$, that $\Lambda_j \epsilon^2 +o(\epsilon^2)\in\sigma(\B_\epsilon^I)$ where $\Lambda_j$ is
an eigenvalue of the operator $\tilde{\A}_\Lambda$.

\subsection{Upper bound for the bottom of the real spectrum} 
For $\Lambda\in\R$, $\PP_\Lambda$ with domain $D(\PP_\Lambda):=H^2(a,b)\cap H^1_0(a,b)$ is a selfadjoint
operator on $L^2(a,b)$ with compact resolvent.
Consequently,  we may
define for $\Lambda \in \mathbb R$, 
\begin{equation}
  \label{eq:42}
\nu(\Lambda)= \inf_{u\in H^1_0([0,1])\setminus\{0\} }\frac{\langle u,(\PP_\Lambda
  -\Lambda)u\rangle}{\|u\|_2^2}  \,.
\end{equation}
It can be easily verified that if $\nu(\Lambda)=0$ then $\Lambda\in\sigma(\B_\epsilon^I)$.\\

Furthermore, if $\Lambda\in\sigma(\B_\epsilon^I) \cap \mathbb R$ then there exists $u\in
D(\PP_\Lambda)$ such that $(\PP_\Lambda -\Lambda)u=0$, and hence, from the definition
of $\nu$ we also learn that $\nu(\Lambda)\leq0$ in this case.  It can be easily
verified that for any $u$ in $D(\PP_\Lambda)$ and $\Lambda \in \mathbb R$
\begin{equation}
\label{eq:82}
  \langle u,\PP_\Lambda u\rangle=
  \epsilon^2\Big[\|u^\prime\|_2^2+ \frac{1}{2}(\|w^\prime_+\|_2^2+\|w^\prime_-\|_2^2)\Big]-\frac{\Lambda}{2}(\|w_+\|_2^2+\|w_-\|_2^2)\,,
\end{equation}
where $w_\pm =(\mathcal L_\pm -\Lambda)^{-1} u$.\\

Before obtaining bounds on $\nu(\Lambda)$ we need the following auxiliary
lemma 
\begin{lemma}
\label{lem:auxiliary-schrodinger}
Let $K>0$. There exist positive $\epsilon_0$ and $C$, such that for all $0<\epsilon<\epsilon_0$,
for any real $\Lambda\leq K\epsilon^2$, for any triple $(w_0,w_-,w_+)$
s.t.   $w_0 \in H^2(a,b)\cap H^1_0(a,b)$ and $w_\pm =(\mathcal L_\pm -\Lambda)^{-1} (xw_0)$, it holds
that, 
\begin{equation}
  \label{eq:83}
\|\tilde{w}_\pm\|_2 +
\epsilon^{2/3}\|\tilde{w}^\prime_\pm\|_2\leq   C\epsilon^{4/3}\| w_0\|_{2,2} \,,
\end{equation}
where $\|\cdot \|_{2,2}$ denotes the norm in $H^2(a,b)$, and 
\begin{equation}
\label{eq:125}
    \tilde{w}_\pm=w_\pm \pm i  w_0 \,.
  \end{equation}
\end{lemma}
\begin{proof}
  It can be easily verified that  $ \tilde{w}_\pm\in H^1_0(a,b)$ and  that
  \begin{displaymath} 
    (\LL_\pm -\Lambda)\tilde{w}_\pm = -i(\pm \epsilon^2 w_0^{\prime\prime} \pm \Lambda\, w_0) \,.
  \end{displaymath}
We can now establish \eqref{eq:83} by using either \cite[Proposition 
5.2]{almog2019stability} or  \cite[Theorem 1.1]{hen15}. Here we use
the fact that $K\epsilon^2 < \epsilon^{\frac 23} |\nu_1|/2$ for sufficiently small $\epsilon$.\\
\end{proof}

Lemma \ref{lem:auxiliary-schrodinger} allows us to obtain an upper
bound for $\nu(\Lambda)$. 
\begin{lemma}
  \label{lem:upper-bound}
Let 
\begin{subequations}
\label{eq:84}
  \begin{equation}
\rho_0 = \inf_{w\in H^1_0(a,b)\setminus\{0\}} \frac{I(w)}{\|[x^2+1]^{1/2}w\|_2^2}\,,
\end{equation}
where
\begin{equation}
  I(w)= \|(xw)^\prime\|_2^2+\|w^\prime\|_2^2 \,.
\end{equation}
\end{subequations}
  There exist positive $\epsilon_0$ and $r_+$ such that for all $0<\epsilon<\epsilon_0$,
  one can find $\Lambda_1\in\sigma(\B_\epsilon^I)\cap\R$ satisfying
   \begin{equation}
\label{eq:85}
\Lambda_1< \rho_0 \epsilon^2(1+r_+\epsilon^{2/3})\,.
   \end{equation}
\end{lemma}
\begin{proof} ~\\
 Note first  that
\begin{displaymath}
\rho_0>\frac{\pi^2}{(b-a)^2}
\end{displaymath} 
since for any $w\in H^1_0(a,b)\setminus\{0\}$
\begin{equation}
\label{eq:126}
    \frac{I(w)}{\|[x^2+1]^{1/2}w\|_2^2}\geq 
    \frac{\pi^2(\|xw\|_2^2+\|w\|_2^2)}{(b-a)^2\|[x^2+1]^{1/2}w\|_2^2}=\pi^2/(b-a)^2\,.
\end{equation}
Equality in \eqref{eq:126} is achieved when both
$w_0=Cxw_0=\hat{C}\sin\pi ((x-a)/(b-a))$ which is clearly impossible.\\
  Let $w_0$ denote a minimizer of \eqref{eq:84} satisfying
  $\|(x^2+1)^{1/2}w_0\|_2=1$. (The proof that $w_0$ exists is rather
  standard, and is therefore omitted.) Then, $w_0$ must satisfy
  \begin{displaymath}
    -(x^2+1)w_0^{\prime\prime}-2xw_0^\prime-\rho_0(x^2+1)w_0=0 \,.
  \end{displaymath}
  Notice that the above balance is identical with \eqref{eq:68bis}
  with $\Lambda_0=\rho_0$, $w_0$ being the ground state of the operator
  $\tilde{\A}_\Lambda$.  It can be readily verified from the above that
\begin{displaymath}
  \|w_0^{\prime\prime}\|_2\leq  \|w_0^\prime\|_2+ \rho_0 \,.
\end{displaymath}
From the fact that $w_0$ is a minimizer of \eqref{eq:84} we readily
conclude that $$\|w_0^\prime\|_2\leq\rho_0^{1/2}$$ and hence
\begin{equation}
  \label{eq:86}
 \|w_0^{\prime\prime}\| \leq  (\rho_0^{1/2}+\rho_0) \,.
\end{equation}

Next, we select $\rho_0 <K$ and $\Lambda\leq K\epsilon^2$. As in \eqref{eq:125} we
set  $\tilde w_\pm =w_\pm \pm i w_0$ and  we obtain by \eqref{eq:83} and
\eqref{eq:86} that there exist positive $\epsilon_0$ and $C$ such that for
all $0<\epsilon<\epsilon_0$
\begin{displaymath}
\big|\|w_\pm\|_2 -\|w_0\|_2\big|+ \big|\|w_\pm^\prime\|_2 -\|w_0^\prime\|_2\big|\leq
\|\tilde{w}_\pm\|_2 + \|\tilde{w}_\pm^\prime\|_2 \leq C\epsilon^{2/3}\,.
\end{displaymath}
Consequently, with $u=x w_0$, 
\begin{displaymath}
   \langle u,(\PP_\Lambda -\Lambda)u\rangle\leq 
  \epsilon^2I(w_0)-\Lambda +r_+\epsilon^{8/3}=\rho_0 \epsilon^2-\Lambda+r_+\epsilon^{8/3}\,.
\end{displaymath}
It follows that for all $K\epsilon^2 > \Lambda>\rho_0 \epsilon^2 +r_+\epsilon^{8/3}$ we have $\nu(\Lambda)<0$.
By \eqref{eq:82}  and since for any $w\in H^1_0(a,b)$ we have
$\|w^\prime\|_2^2\geq (\pi^2/(b-a)^2)\, \|w\|_2^2$, it follows that $\nu(\pi^2\epsilon^2/(b-a)^2)>0$, and
hence, by the continuity of $\nu(\Lambda)$ we have that
$$\sigma(\B_\epsilon^I)\cap(\pi^2\epsilon^2/(b-a)^2,\rho_0  \epsilon^2 +r_+\epsilon^{8/3})\neq\emptyset\,.$$ 
\end{proof}
\subsection{Lower bound for the bottom of the real spectrum}
To obtain a lower bound for $\Lambda_1$ we rewrite \eqref{eq:82} in the form, with  $w_\pm=(\LL_\pm-\Lambda)^{-1}u$, 
\begin{subequations}
  \label{eq:71}
  \begin{equation}
\langle u,(\PP_\Lambda -\Lambda)u\rangle = \frac{1}{2}\big([\epsilon^2\Jg_\Lambda^+(w_+)-\Lambda]\|w_+\|_+^2+[\epsilon^2\Jg_\Lambda^-(w_-)-\Lambda]\|w_-\|_-^2\big)\,,
\end{equation}
where
\begin{equation}
  \|w\|_\pm^2=\Big\|\Big(-\epsilon^2\frac{d^2}{dx^2}\pm ix-\Lambda\Big)w\Big\|_2^2+\|w\|_2^2\,,
\end{equation}
and
\begin{equation}
  \Jg_\Lambda^\pm(w)= \frac{\| \frac{d}{dx} (-\epsilon^2\frac{d^2}{dx^2}\pm ix-\Lambda)w\|_2^2+ \|w^\prime\|_2^2}{\|w\|_\pm^2}\,.
\end{equation}
\end{subequations}
Here we have used that
  \begin{displaymath}
    \|u^\prime\|_2^2=\frac{1}{2}\Big[\Big\| \frac{d}{dx}
    \Big(-\epsilon^2\frac{d^2}{dx^2}+ ix-\Lambda\Big)w_+\Big\|_2^2+ \Big\|
    \frac{d}{dx} \Big(-\epsilon^2\frac{d^2}{dx^2}- ix-\Lambda\Big)w_-\Big\|_2^2 \Big]
  \end{displaymath}
and
\begin{displaymath}
   \|u\|_2^2=\frac{1}{2}\Big[\Big\| 
    \Big(-\epsilon^2\frac{d^2}{dx^2}+ ix-\Lambda\Big)w_+\Big\|_2^2+ \Big\|
    \Big(-\epsilon^2\frac{d^2}{dx^2}- ix-\Lambda\Big)w_-\Big\|_2^2 \Big]\,.
\end{displaymath}
Note that since $u\in H^1_0(a,b)$ we obtain that $w^{\prime\prime}_\pm\in H^1_0(a,b)$
as well. Hence $w_\pm$ belongs to the domain of the form $\Jg_\Lambda^\pm$
which is  defined by 
\begin{displaymath}
  D(\Jg_\Lambda^\pm)= \{\, w\in H^3((a,b),\C)\cap H^1_0(a,b) \,| \,
  w^{\prime\prime}\in H^1_0(a,b)\,\} \,.
\end{displaymath}
We now define
\begin{displaymath}
 \mu(\Lambda) = \inf_{w\in D(\Jg^+_\Lambda)\setminus\{0\}} \Jg_\Lambda^+(w) \,.
\end{displaymath}
As $\Jg_\Lambda^+(w)=\Jg_\Lambda^-(\bar{w})$ there is no need  to define a
different minimization problem for $\Jg_-$. 
 As in \eqref{eq:126} one can  establish that
  $\Jg_\Lambda^\pm(w)>\pi^2/(b-a)^2$  for all $w\in D(\Jg_\Lambda^\pm)$.
Consider then a sequence $\{w^{(k)}\}_{k=1}^\infty\subset D(\Jg_\Lambda^\pm)$ of unity norm,
i.e., $\|w^{(k)}\|_\pm=1$, satisfying $\Jg_\Lambda^\pm(w^{(k)})\leq C$ for some $C>\pi$. By
considering an appropriate subsequence, we may assume that $w^{(k)}$ is
weakly convergent in $D(\Jg_\Lambda^\pm)$ and strongly convergent in
$H^2(a,b)\cap H^1_0(a,b)$ denote the weak limit by $w$. By the strong
convergence we must have $\|w\|_\pm=1$. Furthermore, by the weak
convergence we have
\begin{displaymath}
  \Big\langle\frac{d}{dx} \Big(-\epsilon^2\frac{d^2}{dx^2}\pm ix-\Lambda\Big)w^{(k)},\frac{d}{dx}
  \Big(-\epsilon^2\frac{d^2}{dx^2}\pm ix-\Lambda\Big)w\Big\rangle\to \Big\|\frac{d}{dx}
  \Big(-\epsilon^2\frac{d^2}{dx^2}\pm ix-\Lambda\Big)w\Big\|_2^2 \,,
\end{displaymath}
and hence
\begin{displaymath}
  \Big\|\frac{d}{dx}   \Big(-\epsilon^2\frac{d^2}{dx^2}\pm
  ix-\Lambda\Big)w\Big\|_2^2\leq\liminf \Big\|\frac{d}{dx}
  \Big(-\epsilon^2\frac{d^2}{dx^2}\pm ix-\Lambda\Big)w^{(k)}\Big\|_2^2\,.
\end{displaymath}
By the strong convergence
\begin{displaymath} 
  \|w^\prime\|_2^2 = \lim_{k\to +\infty}  \|(w^{(k)})^\prime\|_2^2 \,, 
\end{displaymath}
and hence
\begin{displaymath}
 \liminf \Jg_\Lambda^\pm(w^{(k)})\geq \Jg_\Lambda^\pm(w) \,.
\end{displaymath}
If $\{w^{(k)}\}_{k=1}^\infty$ is a minimizing sequence, we may immediately
conclude from the the above lower semicontinuity that $w$ is a
minimizer.\\

The foregoing discussion leads us to state the following:
 \begin{lemma}
 \label{lem:lower-bound}
   Let $\rho_0$ be defined by \eqref{eq:84}. There exists positive $\epsilon_0$
   and $r_-$ such that for all $0<\varepsilon<\varepsilon_0$, and every
   $\Lambda_1\in\sigma(\B_\epsilon^I)\cap\R$ we have
   \begin{equation}
     \label{eq:66}
 \rho_0 \epsilon^2(1-r_-\epsilon^2)<\Lambda_1\,.
   \end{equation}
 \end{lemma}
 \begin{proof}
   Let $w_1$ denote a
   unity norm (i.e., $\|w_1\|_+=1$) ground state of $\Jg_\Lambda^+$, associated
   with $\mu(\Lambda)$ for some $\Lambda$ in the interval $(\pi^2\epsilon^2/(b-a)^2,\rho_0 \epsilon^2+r_+\epsilon^{8/3})$.  Clearly,
\begin{equation}
\label{eq:75}
  -\epsilon^2(\LL_--\Lambda)\frac{d^2}{dx^2}(\LL_+-\Lambda)w_1-\epsilon^2w^{\prime\prime}_1-\mu(\Lambda)[1+(\LL_--\Lambda)(\LL_+-\Lambda)]w_1=0 \,.
\end{equation}
Taking the inner product of \eqref{eq:75} with $-w^{\prime\prime}_1$ yields for the real part
\begin{multline}
\label{eq:76}
\epsilon^2 \Re\langle(\LL_+-\Lambda)w_1^{\prime\prime},\frac{d^2}{dx^2}(\LL_+-\Lambda)w_1\rangle+\|w_1^{\prime\prime}\|_2^2 \ - \mu(\Lambda)[\|w_1^\prime\|_2^2- \Re\langle(\LL_+-\Lambda)w_1^{\prime\prime},(\LL_+-\Lambda)w_1\rangle]=0 \,.
\end{multline}
For the first term on the left-hand-side we write
\begin{displaymath}
  \Re\langle(\LL_+-\Lambda)w_1^{\prime\prime},\frac{d^2}{dx^2}(\LL_+-\Lambda)w_1\rangle=
  \|(\LL_+-\Lambda)w_1^{\prime\prime}\|_2^2 + \Re\langle(\LL_+-\Lambda)w_1^{\prime\prime},2iw_1^\prime\rangle\,.
\end{displaymath}
It follows that
\begin{displaymath}
  \Re\langle(\LL_+-\Lambda)w_1^{\prime\prime},\frac{d^2}{dx^2}(\LL_+-\Lambda)w_1\rangle \geq -\|w_1^\prime\|_2^2\geq-\mu(\Lambda) \,.
\end{displaymath}
For the last term on the left-hand-side of \eqref{eq:76} it holds that
\begin{displaymath}
   -\Re\langle(\LL_+-\Lambda)w_1^{\prime\prime},(\LL_+-\Lambda)w_1\rangle=
  \Big\|\frac {d}{dx} (\LL_+-\Lambda)w_1 \Big\|_2^2 +\Re\langle2iw_1^\prime,(\LL_+-\Lambda)w_1\rangle \,,
\end{displaymath}
and hence (recalling that $\|w_1\|_+=1$)
\begin{displaymath}
  \Re\langle(\LL_+-\Lambda)w_1^{\prime\prime},(\LL_+-\Lambda)w_1\rangle \leq \Jg_\Lambda^+(w_1) + 1=\mu(\Lambda)+1\,.
\end{displaymath}
Consequently, we obtain from \eqref{eq:76} that
\begin{displaymath}
  \|w_1^{\prime\prime}\|_2^2\leq2\mu(\Lambda)(\mu(\Lambda)+1) \,.
\end{displaymath}
Since for any $\epsilon$-independent function $\tilde{w}$ satisfying
$\|\tilde{w}\|_+=1$, $K>0$, and $\Lambda<K\epsilon^2$, $\Jg_\Lambda(\tilde{w})$ is
bounded, as $\epsilon\to0$, there must exist a positive constant $\hat C$ such that
\begin{displaymath}
\mu(\Lambda): =\Jg_\Lambda(w_1)\leq\hat{C}\,.
\end{displaymath}
Consequently, there
exists $C>0$ and $\epsilon_0$ such that for any $\epsilon \in (0,\epsilon_0]$ and $\Lambda \in
(\frac{\pi^2}{(b-a)^2} , \epsilon^2 \rho_0\epsilon^2+r_+\epsilon^{8/3})$,
\begin{equation}
\label{eq:77}
   \|w_1^{\prime\prime}\|_2\leq C \,.
\end{equation}
Next, we write
\begin{equation}
\label{eq:80}
  \Big\|\frac {d}{dx} (-\epsilon^2\frac{d^2}{dx^2}+ix-\Lambda)w_1\Big\|_2^2= \epsilon^4\|w_1^{(3)}\|_2^2 +
  \Big\|\frac {d}{dx} (ix-\Lambda)w_1\Big\|_2^2-2\epsilon^2\Re\langle\frac {d}{dx} (ix-\Lambda)w_1,w_1^{(3)}\rangle 
\end{equation}
For the last term on the right-hand-side we have
\begin{displaymath}
  - \Re\langle\frac {d}{dx} (ix-\Lambda)w_1,w_1^{(3)}\rangle =
  -\Lambda\|w_1^{\prime\prime}\|_2^2-2\Im\langle w_1^\prime,w_1^{\prime\prime}\rangle \,.
\end{displaymath}
Consequently, by \eqref{eq:77} we can conclude the existence of $C>0$,
independent of $\epsilon$, such that
\begin{displaymath}
  - 2\epsilon^2\Re\langle\frac {d}{dx} (ix-\Lambda)w_1,w_1^{(3)}\rangle \geq -C\epsilon^2 \,. 
\end{displaymath}
Substituting the above into \eqref{eq:80} yields, for some $C>0$
\begin{displaymath}
  \Big\|\frac {d}{dx} (-\epsilon^2\frac{d^2}{dx^2}+ix-\Lambda)w_1\Big\|_2^2 \geq  \Big\|\frac {d}{dx} (ix-\Lambda)w_1\Big\|_2^2
  -C\epsilon^2 \geq \|(xw_1)^\prime\|_2^2-\tilde C \epsilon^2 \,.
\end{displaymath}
Hence, by \eqref{eq:71} we have
\begin{equation}
\label{eq:69}
  \mu(\Lambda)=\Jg_\Lambda^+(w_1) \geq \|(xw_1)^\prime\|_2^2+\|w_1^\prime\|_2^2-
\tilde C \epsilon^2\geq\rho_0\|[x^2+1]^{1/2}w_1\|_2^2-C\epsilon^2 \,.
\end{equation}
By \eqref{eq:77} and the fact that $\Lambda\in(\pi^2\epsilon^2/(b-a)^2,\rho_0 \epsilon^2+r_+\epsilon^{8/3})$ it holds that
\begin{displaymath}
  1=\|w_1\|_+\leq \|[x^2+1]^{1/2}w_1\|_2 + C\epsilon^2\,.
\end{displaymath}
Substituting the above into \eqref{eq:69} yields that there exists
$r_->0$ such that
\begin{equation}
  \label{eq:72}
\mu(\Lambda)\geq \rho_0-r_-\epsilon^2 \,.
\end{equation}
We can now conclude from \eqref{eq:71} that
\begin{displaymath}
  \langle u,(\PP_\Lambda -\Lambda)u\rangle \geq \frac{1}{2}[\rho_0 \epsilon^2-r_-\epsilon^4 -\Lambda](\|w_+\|_+^2+\|w_-\|_-^2)\,,
\end{displaymath}
which immediately leads to \eqref{eq:66}.
\end{proof}

\subsection{Proof of Theorem \ref{prop:spectrum-bounded}}
  Let $\Lambda=\Lambda_r+i \Lambda_i\in\C$ satisfy $\Re\Lambda<\rho_0\epsilon^2-r_-\epsilon^4$. Let further $u\in D(\PP_\Lambda)$
  and $w_\pm=(\LL_\pm-\Lambda)^{-1}u$. It can be easily verified that
  \begin{displaymath}
    \Re\langle u,(\LL_\pm-\Lambda)^{-1}u\rangle = \epsilon^2\|w_\pm^\prime\|_2^2-\Lambda_r\|w_\pm\|_2^2 \,.
  \end{displaymath}
Consequently we may write that
\begin{multline*}
  \Re\langle u,(\PP_\Lambda -\Lambda)u\rangle=
  \frac{1}{2}\big([\epsilon^2\Jg_{\Lambda_r}^+(w_+)-\Lambda_r]\|w_+\|_+^2+
  [\epsilon^2\Jg_{\Lambda_r}^-(w_-)-\Lambda_r]\|w_-\|_-^2\big)\,,  
\end{multline*}
With the aid of \eqref{eq:72} we then conclude that, for sufficiently
small $\epsilon$,
\begin{equation}
\label{eq:54}
  \Re\langle u,(\PP_\Lambda -\Lambda) u\rangle \geq \frac{1}{2}(\rho_0 \epsilon^2-\Lambda_r -r_-\epsilon^4)(\|w_+\|_2^2+\|w_-\|_2^2) \,,
\end{equation}
yielding
\begin{equation}
\label{eq:73}
  \liminf_{\epsilon\to0}\epsilon^{-2}\Re \sigma(\B_\epsilon^I)\geq\rho_0\,. 
\end{equation}
By Lemma \ref{lem:upper-bound} it holds that
\begin{displaymath}
    \limsup_{\epsilon\to0}\epsilon^{-2}\Re \sigma(\B_\epsilon^I)\leq \rho_0\,,
\end{displaymath}
which together with \eqref{eq:73} achieves the proof of Theorem
\ref{prop:spectrum-bounded}.

\def\cprime{$'$}

\end{document}